\documentclass[peerreview,journal,draftcls,onecolumn,12pt]{IEEEtran}
\usepackage{amssymb}
\usepackage{amsmath}
\usepackage{rotating}
\usepackage{multirow}
\usepackage{multicol}
 \usepackage{subfigure}
\usepackage{lipsum}
\usepackage[ruled,vlined]{algorithm2e}
\usepackage{booktabs}
\usepackage{array}
\usepackage{bbm}
\usepackage{tabu}
\usepackage{ctable} 

\newcolumntype{M}[1]{>{\arraybackslash}m{#1}}
\newcolumntype{N}{@{}m{0pt}@{}}
\renewcommand{\arraystretch}{1.5}

\newtheorem{theorem}{\textbf{Theorem}}

\newtheorem{definition}{\textbf{Definition}}
\newtheorem{remark}{\textbf{Remark}}

\begin{document}
\title{Large-Scale Cloud Radio Access Networks \\with Practical Constraints: \\Asymptotic Analysis and Its Implications}
\author{\smaller{Kyung~Jun~Choi,~\IEEEmembership{Student~Member,~IEEE},~and~Kwang~Soon~Kim$^\dag$,~\IEEEmembership{Senior~Member,~IEEE}}
\thanks{
This work was supported in part by ICT R\&D program of MSIP/IITP \{B0101-16-1367, Next Generation WLAN System with High Efficient Performance\}, and in part by Basic Science Research Program through the National Research Foundation of Korea (NRF) funded by the Ministry of Education, Science and Technology (NRF-2014R1A2A2A01007254).

The authors are with the Department of Electrical and Electronic Engineering, Yonsei University, 50 Yonsei-ro, Seodaemun-gu, Seoul 120-749, Korea.}
\thanks{$^\dag$: Corresponding author (ks.kim@yonsei.ac.kr)}
}
\markboth{Choi and Kim: Large-Scale Cloud Radio Access Networks with Practical Constraints}{in preparation} 
\maketitle
\vspace{-5mm}
\begin{abstract}
Large-scale cloud radio access network (LS-CRAN) is a highly promising next-generation
cellular network architecture whereby lots of base stations (BSs) equipped with a massive antenna array are
connected to a cloud-computing based central processor unit via digital front/backhaul links. 
This paper studies an asymptotic behavior of downlink (DL) performance of a LS-CRAN with three practical constraints: 1) limited transmit power, 2) limited front/backhaul capacity, and 3) limited pilot resource. As an asymptotic performance measure, the scaling exponent of the signal-to-interference-plus-noise-ratio (SINR) is derived for interference-free (IF), maximum-ratio transmission (MRT), and zero-forcing (ZF) operations. Our asymptotic analysis reveals four fundamental operating regimes and the performances of both MRT and ZF operations are fundamentally limited by the UL transmit power for estimating user's channel state information, not the DL transmit power. We obtain the conditions that MRT or ZF operation becomes interference-free, i.e., order-optimal with three practical constraints. Specifically, as higher UL transmit power is provided, more users can be associated and the data rate per user can be increased simultaneously while keeping the order-optimality as long as the total front/backhaul overhead is $\Omega(N^{\eta_{\rm{bs}}+\eta_{\rm{ant}}+\eta_{\rm{user}}+\frac{2}{\alpha}\rho^{\rm{ul}}})$ and $\Omega(N^{\eta_{\rm{user}}-\eta_{\rm{bs}}})$ pilot resources are available. 
It is also shown that how the target quality-of-service (QoS) in terms of SINR and the number of users satisfying the target QoS can simultaneously grow as the network size increases and the way how the network size increases under the practical constraints, which can provide meaningful insights for future cellular systems.

\end{abstract}
\clearpage

\section{Introduction}
\subsection{Motivation}
\IEEEPARstart{R}{ecently}, mobile data traffic is explosively and continuously rising due to smart phone and tablet users and it is expected that next-generation (a.k.a. the 5th generation) cellular networks will offer a 1000x increase in network capacity as well as a 1000x increase in energy-efficiency in the following decade to meet such an excessively high user demand \cite{AndrewsWhat}. The most prominent way to increase the network capacity is the \emph{network densification} \cite{BhushanNetwork} by either implementing more antennas at base stations (BSs) called large-scale antenna system (LSAS) \cite{RusekScaling} or adding more small cells at hot spot areas called ultra-dense network (UDN) \cite{GotsisUDN}. By using a lot of antennas at each BS, the LSAS can exploit massive spatial dimension to generate a sharp beam and thus it can provide higher spectral efficiency and also lower power consumption \cite{NgoEnergy}. However, the performance of the LSAS is highly limited by the accuracy of channel state information (CSI) and varies over the geographical locations of users. On the other hand, the UDN exploits spatial reuse obtained from deploying more small BSs so that it can offer geographically uniform performance to each user due to the reduction of the access distance between users and BSs \cite{HwangHolistic}, \cite{HoydisGreen}. However, its performance is limited by uncoordinated out-of-cell interference so that an efficient interference management is a key challenge for the UDN, while managing the costs in terms of the front/backhaul overhead and computational complexity for a network operator.

Large-scale cloud radio access network (LS-CRAN) is recently regarded as a novel wireless cellular network architecture to unify the above two architectures and is able to implement the interference handling mechanisms introduced for the long-term evolution (LTE) or LTE-Advanced, such as the enhanced inter-cell interference coordination \cite{HuangIncreasing} and the coordinated multi-point transmission \cite{IrmerCoordinated}, \cite{CheckoCloud}. 
In the LS-CRAN, lots of BSs each with a massive antenna array are connected to a virtualized central processing unit (CPU) via dedicated front/backhaul link and some of the baseband processing functionality of each BS is migrated to the CPU. As a result, it is expected that the performance of the LS-CRAN is much higher than that of a conventional network.

A natural objective of the LS-CRAN is to maximize the number of supporting users while preserving their quality-of-service (QoS) requirements under practical constraints. If the LS-CRAN can manage the inter-user interference and inter-BS interference ideally, the network can be transformed into a multi-access (UL) or broadcast (DL) channel without any interference \cite{CaireAchievable}. However, unfortunately in practical systems, such an ideal situation is hard to be realized due to various practical limitations so that the maximum number of supporting users with a QoS requirement is carefully considered under such practical constraints.

There are three dominant practical constraints in the LS-CRAN. First, the network total transmit power should be maintained at an appropriate level not only for saving power consumption (and operating cost) \cite{LinToward}, but also for adopting cheap radio-frequency (RF) components in each BS and user \cite{BjornsonMassive}. Second, the capacity of the front/backhaul links connecting the BSs and the CPU is usually limited as pointed out in \cite{BiermannHow} and installing dedicated high-capacity front/backhaul link for every BS causes too much cost, especially in a dense deployment \cite{ChanclouOptical}. Last, the pilot resource for the LS-CRAN to acquire the CSI and handle the interference caused from a large number of BSs, antennas, and users by using the maximum ratio transmission (MRT) \cite{LoMaximum} or the zero-forcing (ZF) \cite{SpencerZero} operation, needs to be limited.

Then, a naturally raised question is that how many users can be supported by the LS-CRAN with a QoS requirement under these three practical constraints? 
Also, how do the target QoS and the number of users satisfying the target QoS simultaneously grow as the network size increases?
The answers are ready to be addressed in this paper.

\subsection{Related Works}
Although the performance of an LS-CRAN has been widely investigated in literature \cite{IrmerCoordinated}, \cite{MarschLarge} via simulations, an intuitive analytic result considering practical limitations is not available yet. The exact distribution of the signal-to-interference-plus-noise ratio (SINR) of the MRT or ZF operation is derived in \cite{BasnayakaPerformance1} for the case of two distributed BSs and its Laplace approximated distribution is derived in $\cite{BasnayakaPerformance2}$ for a general LS-CRAN. In \cite{CheikhAnalytical}, a simpler form of the SINR distribution is derived by approximating the SINR as a Gamma random variable. As parallel works, the average achievable rate of the MRT or ZF operation is derived in \cite{YangPerformance} for a single BS case, and in \cite{WangAsymptotic} for the case of fully-distributed or fully-colocated antennas as a function of the BS and user locations. In order to understand the network behavior more intuitively, stochastic geometry has been recently adopted to remove such dependency and describe important network metrics in a probabilistic way. In \cite{TanbougiTractable}-\cite{GiovanidisStochastic}, the SINR distributions of cooperative transmission schemes are provided by adopting stochastic geometry, but the analysis resorts on multiple numerical integrals. In \cite{HuangAnalytical}, an asymptotic analysis on the tail of the SINR distribution is provided by applying large-deviation theory and stochastic geometry. In \cite{LeeSpectral},  the average achievable rate of the MRT and ZF successive interference cancellation is investigated and its scaling laws are obtained in a dense random network with multiple receive antennas. However, the effects of the limited transmit power, the limited front/backhaul capacity, or the limited pilot resource have not been jointly considered.

As the network is densified, the effect of erroneous CSI becomes more crucial because all of users and antennas are jointly processed in the CPU so that small CSI error can dramatically degrade the overall network performance. In an LS-CRAN, there exist three main sources causing erroneous CSI:  the limited transmit power or the limited pilot resource in the process of CSI acquisition and the limited capacity of front/backhaul links. In \cite{NgoEnergy}, the effect of the CSI acquisition error on the energy and spectral efficiency is analyzed for the DL single-cell LSAS using the MRT or ZF operation. In \cite{MarschUL}, an information-theoretic capacity bound is derived for UL CRAN by modeling the CSI error as an additive Gaussian noise. In \cite{WangPerformance} and \cite{KongMultipair}, the average achievable rate and its asymptotic analysis are obtained in multi-antenna relay networks by considering the effect of the CSI acquisition errors.
In \cite{XiaFundamentals}, authors model and study the front/backhaul link by using stochastic geometry. The achievable rate with various front/backhaul data compression strategies is analyzed in the UL scenario \cite{SimeoneUL} and this result is extended to the DL scenario \cite{SimeoneDL}. In \cite{ZhouUplink}, the achievable rate of a cooperative transmission scheme is derived by considering backhaul burdens.
 However, these results become quite intractable if applied to the LS-CRAN with large number of antennas and users so that more insightful and intuitive analysis for the LS-CRAN is required. 


\subsection{Contributions}
The main objective of this paper is to characterize the asymptotic behavior of the LS-CRAN under three practical constraints: (1) the limited total transmit power, (2) the limited front/backhaul capacity, and (3) the limited pilot resource. The major contributions of this paper are summarized as follows.

\begin{itemize}
\item
The scaling exponents of the signal-to-interference-plus-noise ratio (SINR) defined and derived as a function of the numbers of BS, BS antennas, and single-antenna users, and the UL/DL transmit power for the interference-free (IF), MRT, and ZF operations, which are stated in Theorems 1-3. 

\item Based on the derived scaling exponents, four fundamental regimes are distinguished according to the UL transmit power (or the CSI acquisition error): (i) the extremely high (UL transmission) power regime (\textsf{EH}), (ii) the high power regime (\textsf{H}), (iii) the medium  power regime (\textsf{M}), and (iv) the low power regime (\textsf{L}). The scaling exponents of the three operations are derived  according to the four regimes with insightful discussions.

\item The limited front/backhaul capacity can be successfully transformed into partial associations, where each BS serves only users affordable with its front/backhaul capacity. The analytical result shows that such limited front/backhaul capacity and pilot resource constrain the network behavior in a similar way as the limited transmit power, as stated in Theorems 4 and 5.

\item In order to characterize the network-wise performance of the LS-CRAN, we define the number of supportable users as the maximum number of users satisfying a pre-determined SINR requirement. Then, we derive the tradeoff in the scaling exponents of them as the network size increases under the practical limitations, as stated in Theorem 6.


\end{itemize}

\subsection{Organization and Notation}
The reminder of this paper is organized as follows. In Section II, the LS-CRAN system model is described and the scaling exponent of the SINR is defined as a network performance measure. Section III provides some candidate operations suitable for practical implementation and the corresponding scaling exponents are analyzed by considering practical limitations to provide an intuitive understanding on the LS-CRAN in Section IV. In Section V, the number of supportable users is defined as another important network metric and analyzed to provide an insight on the asymptotic network behavior of the LS-CRAN. Finally, conclusion is given in Section VI.

Matrices and vectors are respectively denoted by boldface uppercase and lowercase characters. The superscript ${( \cdot )^*}$, ${( \cdot )^T}$ and ${( \cdot )^H}$ denote the conjugate, transpose and conjugate transpose, respectively. $\left| \cdot \right|$, $\mathbb{E}[ \cdot ]$, ${\text{Tr}}(\cdot)$, and $\delta(n)$ stand for the cardinality of a set, statistical expectation, trace of a square matrix, and Kronecker delta function, respectively. Also, $\mathbf{1}_{A\times 1}$ and $\mathbf{0}_{A\times 1}$ are the $A\times 1$ all one vector and all zero vector, respectively. Further, $\mathrm{diag}(a_1,a_2,\cdots,a_n)$ denotes the diagonal matrix whose $(k,k)$th element is $a_k$, and $\mathcal{CN}(\mu,\sigma^2)$ denotes the circularly symmetric complex Gaussian distribution with mean $\mu$ and variance $\sigma^2$. For a $C\times 1$ vector $\mathbf{x}_{ab}$,  $[{{\bf{x}}_{ab}}]_{a = 1,b = 1}^{A,B}$ denotes the $AC\times B$ matrix $\left[ {{{[{\bf{x}}_{11}^H,\cdots,{\bf{x}}_{A1}^H]}^H},\cdots,[{\bf{x}}_{1B}^H,\cdots,{\bf{x}}_{AB}^H]^H} \right]$. 

In this paper, we only deal with a scaling exponent, $s$, of a random measure, $f(n)$, as $n\to\infty$ in a probabilistic sense, which is mathematically defined as follows. The scaling exponent of $f(n)$ is $s$ in probability if 
\begin{equation}\label{eq_A1}
\lim_{n\to\infty}\Pr\left(  \left|\frac{\log f(n)}{\log n} - s\right| < \epsilon \right) = 1
\end{equation}
holds for any positive finite $\epsilon$. 
Also a slightly abused notation $s=\pm \infty$ is used if there is no finite positive $\epsilon$ satisfying (\ref{eq_A1}) for any finite $s$.
Let $s$ and $t$ be the scaling exponents of $f(n)$ and $g(n)$, respectively. Then, the following modified order notations are used through this paper. 
\begin{itemize}
\item[1)] $f(n)  = o(g(n))$ or $f(n)\ll g(n)$, if $s<t$,
\item[2)] $f(n)  = O(g(n))$, if $s\le t$,
\item[3)] $f(n)  = \omega(g(n))$  or $f(n)\gg g(n)$, if $s>t$,
\item[4)] $f(n)  = \Omega(g(n))$, if $s\ge t$, and
\item[5)] $f(n)  = \Theta(g(n))$ or $f(n)\asymp g(n)$, if $s=t$.
\end{itemize}



{\renewcommand{\arraystretch}{1}
\begin{table}
\center\caption{Summary of definitions}
\begin{tabular}{ M{3cm} M{7cm} }
\addlinespace
\toprule 
{\bf Symbol} & {\bf Definition} \\
\toprule
$L$, $M$, $K$& the numbers of BSs, BS antennas, or users, respectively\\
\hline
$N$& the network size, $N=LM$\\
\hline
$\mathcal{X}=\{X_1,\cdots,X_L\}$& {the set of BSs, where $X_l$ is the location of BS $l$}\\
\hline
$\mathcal{U}=\{U_1,\cdots,U_K\}$& {the set of users, where $U_k$ is the location of user $k$}\\
\hline
$\mathcal{X}_k$& {the set of BSs associated to user $k$}\\
\hline
$\mathcal{U}_l$& {the set of users associated to BS $l$}\\
\hline
$R_{\rm{th}}$&{the association range}\\
\hline
$\eta_{\rm{bs}}$, $\eta_{\rm{ant}}$, $\eta_{\rm{user}}$& scaling exponents for the number of BSs, BS antennas or users, respectively\\
\hline
$N_{\rm{PA}}$, $\upsilon_{\rm{PA}}$ & the number of associated BSs per user and its scaling exponent
\\ \hline
$T$, $\upsilon_{\rm{PR}}$ & the number of pilots and its scaling exponent
\\ \hline
$\alpha$ & Pathloss exponent ($\alpha>2$) \\
\hline
$\mathbf{g}_{lk}=\sqrt{\beta_{lk}}\mathbf{h}_{lk} $& $M\times 1$ wireless channel between BS $l$ and user $k$, where $\beta_{lk}$ is a long-term fading, $\mathbf{h}_{lk}$ is a short-term fading \\
\hline
$P_k^{\rm{ul}}$, $P_k^{\rm{dl}}$, $\rho^{\rm{ul}}$, $\rho^{\rm{ul}}$& the UL or DL transmit power for user $k$, respectively and their scaling exponents\\
\hline
$\mathsf{SNR}_k^\Upsilon$, $\mathsf{SIR}_k^\Upsilon$, $\mathsf{SINR}_k^\Upsilon$ & {SNR, SIR or SINR of user $k$ when operation $\Upsilon$ is applied, respectively and their scaling exponents }\\
\bottomrule
\end{tabular}
\end{table}
}
\begin{figure*}
\centering
\subfigure[System model]{\includegraphics[width=8cm]{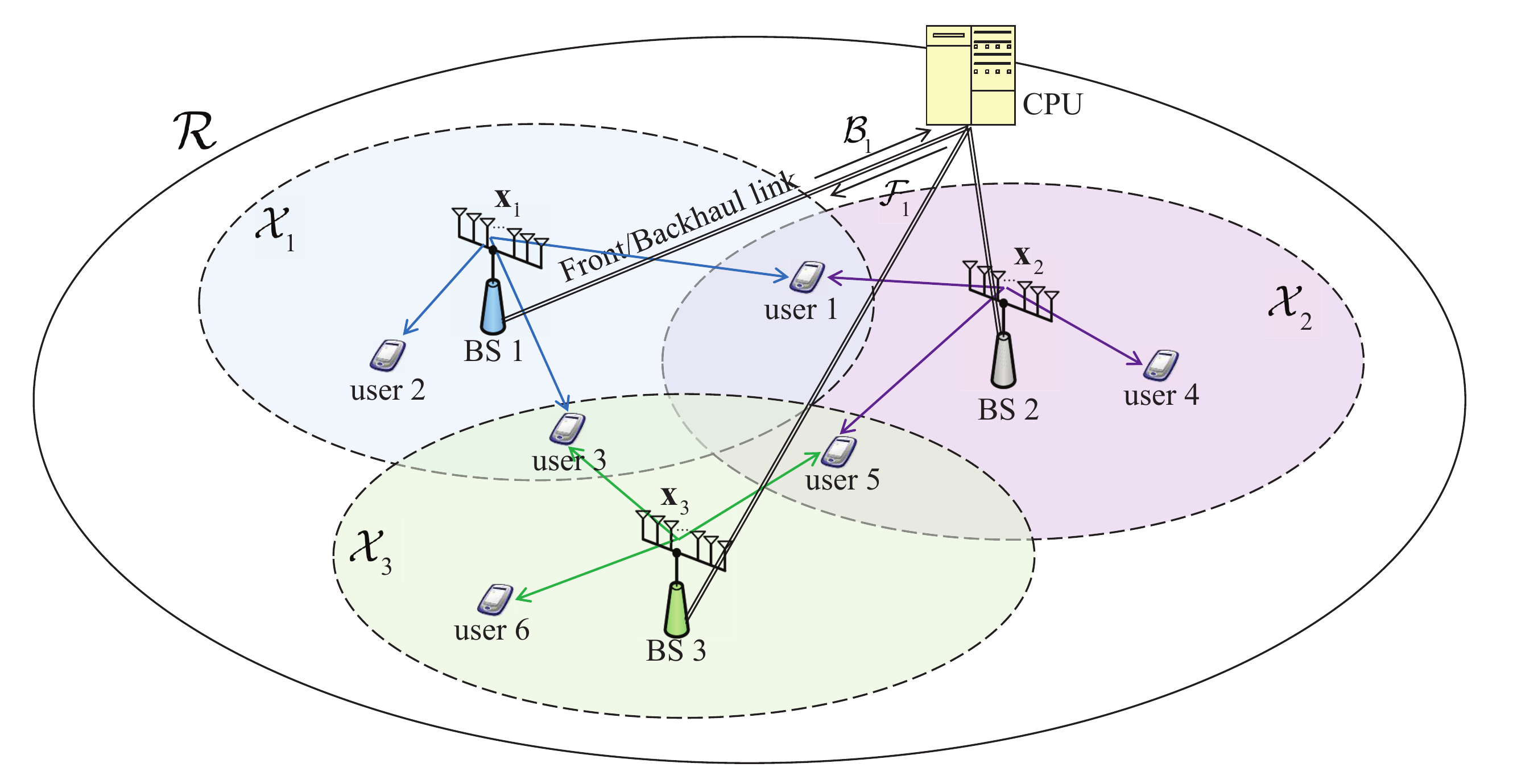}}\quad
\subfigure[DL frame structure]{\includegraphics[width=8cm]{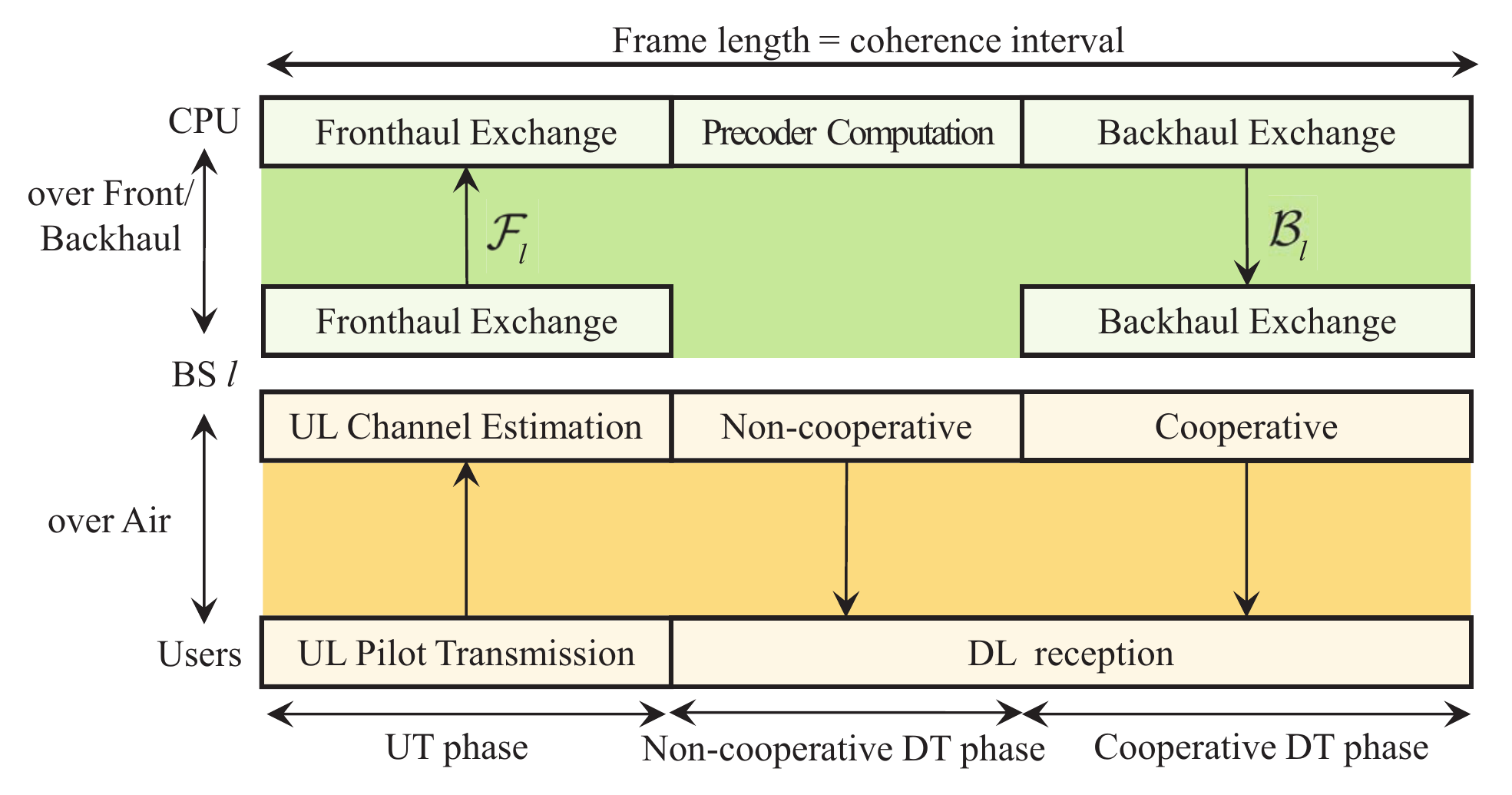}}\\
\caption{System model and DL frame structure for LS-CRAN}
\end{figure*}

\section{System Model and Performance Measure}
\subsection{LS-CRAN Model and Frame Structure}
Consider an LS-CRAN system as illustrated in Fig. 1 (a). Suppose that $L$ BSs with $M$ antennas and $K$ users with a single antenna are uniformly distributed on a finite region $\mathcal{R}$. The sets of BSs and users are denoted as $\mathcal{X}=\{X_1,X_2,\cdots,X_L\}$ and $\mathcal{U}=\{U_1,U_2,\cdots,U_K\}$, respectively. With a slight abuse of notations, $X_l$ and $U_k$ are used as the locations of BS $l$ and user $k$, respectively. It is  assumed that each of users is associated with neighboring BSs and the set of the BSs serving user $k$ is defined as
\begin{equation}
\mathcal{X}_k = \{X_l \in \mathcal{X} | |U_k - X_l| \le R_\text{th}\},
\end{equation}
where $R_\text{th}$ is the association range and the set of users associated with BS $l$ is denoted as $\mathcal{U}_l=\{U_k\in\mathcal{U}|X_l\in\mathcal{X}_k\}$. To guarantee each of users is served by at least one BS, $\mathcal{X}_k\ne\emptyset$ for all $k$. Note that $\mathcal{U}_l$  and $\mathcal{U}_{l\rq{}}$ are not necessarily disjoint and typical users are associated with multiple BSs for being served cooperatively. 
In this paper, a network is called \emph{fully associated}, if $\mathcal{X}_k = \mathcal{X}$ for all $k$. Otherwise, the network is called \emph{partially associated}.

The BSs are connected to the CPU via high-speed dedicated front/backhaul link for enabling a cooperative transmission operation. For inter-signaling between the BSs and the CPU, the sets of front/backhaul information are defined as $\mathcal{F}_l$ and $\mathcal{B}_l$, where $\mathcal{F}_l$ is the fronthaul information set (from CPU to BS $l$) and $\mathcal{B}_l$ is the backhaul information set (from BS $l$ to CPU). The elements of $\mathcal{F}_l$ and $\mathcal{B}_l$ are closely related according to a specific transmission operation, which will be described in Section III.

The DL frame structure operating in time-division duplex (TDD) mode is shown in Fig. 1(b) and has three phases as follows. UL training (UT) phase, non-cooperative data transmission (DT) phase, and cooperative DT-phase. In the UT-phase, the instantaneous CSI of each user is estimated by receiving user\rq{}s pilot (or reference) signal at each associated BS independently. After the UT-phase, the non-cooperative DT-phase is performed first and followed by the cooperative DT-phase. In the non-cooperative DT-phase, each BS separately transmits the data symbols to the associated users without any front/backhaul exchange. In the cooperative DT-phase, the BSs jointly transmit the data symbols to the associated users with the aid of the exchanged information via the front/backhaul link. For simplicity, the cooperative DT-phase is only focused in this paper. 

\subsection{Channel Model}
Let ${{\mathbf{g}}^H_{lk}}$ denote the  $1 \times M$ flat-fading DL channel vector from BS $l$ to user $k$, which can be written as\footnote{A narrow-band flat-fading channel is assumed because wideband frequency-selective channels may be decomposed into multiple narrow-band channels using modulation schemes such as the orthogonal frequency division multiplexing.}
\begin{equation}\label{channel_model} 
{{\mathbf{g}}^H_{lk}} = \sqrt {{\beta _{lk}}} {{\mathbf{h}}^H_{lk}},
\end{equation}
where ${{\mathbf{h}}_{lk}} \in {\mathbb{C}^{M \times 1}}$ is the short-term CSI whose elements are independent and identically distributed (i.i.d.) $\mathcal{CN}(0,1)$ and ${\beta _{lk}}( \ge 0)$ is the long-term CSI depending on the path-loss and shadowing. 
The long-term CSI between BS $l$ and user $k$ is modeled as $\beta_{lk} = \left|X_l-U_k\right|^{-\alpha}$, where $\alpha (> 2)$ is the wireless channel path-loss exponent.\footnote{
By using the random displacement theorem similarly as in \cite{DhillonDownlink}, this model can include the shadowing effect.}

It is assumed that the short-term CSI of each user remains constant within a given frame but independent across different frames, while the long-term CSI does not vary during a much longer interval. Further, it is assumed that the long-term CSIs among all BSs and users are perfectly known at the CPU through an infrequent feedback with a negligible overhead. Additionally, we assume perfect TDD reciprocity calibration so that the UL channel vector is just a transpose of the DL channel vector (i.e., the UL channel vector from user $k$ to BS $l$ is denoted as $\mathbf{g}^*_{lk}$).

\subsection{Signal Model}
Let $\mathbf{s}=[s_1,\cdots,s_K]^T$ be the $K\times 1$ information symbol vector with $\mathbb{E}[\mathbf{s}\mathbf{s}^H]=\mathbf{I}_K$, where $s_k$ denotes the information symbol for user $k$. Let $ {\mathbf{x}} = [\mathbf{x}_1^T,\cdots,\mathbf{x}_L^T]^T$ be the $LM\times 1$ \emph{global} transmitted signal vector, where $\mathbf{x}_l$ is the \emph{local} transmitted signal vector of BS $l$, given by
\begin{equation}\label{transmit_signal}
 {\mathbf{x}} ={\mathbf{F}}^\Upsilon \left({\mathbf{Q}^\Upsilon}\right)^{\frac{1}{2}}\mathbf{s},
\end{equation}
where $\mathbf{Q}^\Upsilon\triangleq\mathrm{diag}(Q_{1}^\Upsilon,\cdots,Q_{K}^\Upsilon)$ denotes the $K\times K$ power allocation matrix and $Q_k^\Upsilon(\ge 0)$ denotes the power allocated to user $k$, ${\mathbf{F}}^\Upsilon = [\mathbf{f}^\Upsilon_{lk}]_{l=1,k=1}^{L,K}$ denotes the $LM\times K$ precoding matrix and $\mathbf{f}^\Upsilon_{lk}$ denotes the $M\times 1$ precoding vector of user $k$ for BS $l$, and $\Upsilon$ denotes the cooperative transmission operation used in the LS-CRAN. Since the BS serves the associated users only, $\mathbf{f}^\Upsilon_{lk} = \mathbf{0}_{M\times 1}$ for $\forall X_l\notin \mathcal{X}_k$ (or equivalently $\forall U_k \in \mathcal{U}_l$).  Note that the transmitted signal vector of BS $l$ can be written by $\mathbf{x}_l =  \sum\nolimits_{U_j\in\mathcal{U}_l} \mathbf{f}^\Upsilon_{lj}\sqrt{Q_j^\Upsilon}s_j$.
Then, the DL transmit power for user $j$, $P_j^{{\rm{dl}}}$, is given as 
\begin{equation}\label{DL_transmission_power}
P_j^{{\rm{dl}}} = Q_j^\Upsilon\sum\limits_{X_l\in \mathcal{X}_j}\left\|\mathbf{f}^\Upsilon_{lj}\right\|^2,
\end{equation}
and the total DL transmit power is given by
$P^{{\rm{dl}}}_{{\Sigma}} \triangleq \sum\nolimits_{U_j\in\mathcal{U}} P_j^{\rm{dl}}$.

Let ${\mathbf{y}} \triangleq {[{y_1},\cdots,{y_K}]^T}$ be the $K\times 1$ aggregated received signal vector, given by
\begin{equation}\label{received_signal}
\begin{split}
{\mathbf{y}} & =  {\mathbf{G}} {\mathbf{x}} + {\mathbf{n}}\\
&=  {\mathbf{G}}{\mathbf{F}}^\Upsilon \left({\mathbf{Q}^\Upsilon}\right)^{\frac{1}{2}}\mathbf{s}+\mathbf{n},\\
\end{split}
\end{equation}
where ${\mathbf{G}} = \left([\mathbf{g}_{lk}]_{l=1,k=1}^{L,K}\right)^H$ denotes the $K\times LM$ channel matrix among all users and all BSs and ${\mathbf{n}} \triangleq {[{n_1},\cdots,{n_K}]^T}\sim\mathcal{CN}(\mathbf{0},\mathbf{I}_K)$ denotes the $K \times 1$ noise vector.
Then, the received signal at user $k$ can be expressed as
\begin{align}\label{received_signal_from_BS_l}
{y_k} &=  \psi^\Upsilon_{kk} \sqrt{Q^\Upsilon_k} s_k + \sum\limits_{U_j\in\mathcal{U}\backslash\{U_k\}}\psi^\Upsilon_{kj} \sqrt{Q^\Upsilon_j} s_j + n_k,
\end{align}
where $\psi^{\Upsilon}_{kj} = \sum\nolimits_{X_l\in\mathcal{X}_j}\mathbf{g}_{lk}^H\mathbf{f}^\Upsilon_{lj}$ is the effective channel seen at user $k$.

\subsection{Pilot Allocation and Channel Estimation}
The UL channel is estimated during the dedicated UT-phase and then the DL channel is obtained by the TDD channel reciprocity. It is assumed that the length of the UT-phase is $T$ and there are $T$ orthonormal pilot signals, denoted as $ {{\boldsymbol\psi _i} \in {\mathbb{C}^{T \times 1}},i = 1,2,\cdots,T}$, where $\boldsymbol\psi _i^H{\boldsymbol\psi _j} = \delta (i-j)$. Then, user $j$ transmits $\sqrt{P^{\rm{ul}}_j}\boldsymbol\psi_{\pi_j}^T$ during the UT-phase of length $T$, where $P^{\rm{ul}}_j$ is the UL transmit power of user $j$ and $\pi_j$ is the index of the pilot signal allocated to user $j$. The total UL transmit power is denoted as $P^{\rm{ul}}_\Sigma = \sum\nolimits_{U_j\in\mathcal{U}}P_j^{\rm{ul}}$.  Then, the $M\times T$ received signal matrix at BS $l$ during the UT-phase can be written as
\begin{equation}\label{UT_received_signal}
{{\mathbf{Y}}_l} = \sum\limits_{U_j \in\mathcal{U}} {\sqrt {P^{{\rm{ul}}}_j{\beta _{lj}}} {{\mathbf{h}}^*_{lj}}}\boldsymbol\psi _{\pi_j}^T + {{\mathbf{V}}_l},
\end{equation}
where $\mathbf{V}_l$ denotes the $M\times T$ noise matrix whose elements are i.i.d. $\mathcal{CN}(0,1)$. Using the minimum mean-square error (MMSE) channel estimator \cite{LiPilot}, the estimated short-term CSI of user $k\in\mathcal{U}_l$ at BS $l$ can be written as
\begin{equation}\label{estimated_CSI}
\begin{split}
{{\widehat{\mathbf{ h}}}_{lk}} &= {{\vartheta_{lk}}}{{\mathbf{Y}}^*_l}\boldsymbol\psi _{\pi_k} \\ 
&={\phi_{lkk}}\mathbf{h}_{lk}+\sum\limits_{U_j\in\mathcal{U}\backslash\{U_k\}} { {{{{\phi_{ljk}}}}} {{\mathbf{h}}_{lj}}}\delta(\pi_k-\pi_j) + \vartheta_{lk}{{{\mathbf{\widetilde v}}}_{lk}}, 
\end{split}
\end{equation}
where 
$$
 {\vartheta_{lk}} = \frac{{\sqrt {P_k^{{\rm{ul}}}{\beta _{lk}}} }}{{\sum\limits_{U_i \in {{\cal U}}} {P_i^{{\rm{ul}}}{\beta _{li}}}\delta(\pi_k-\pi_i) + {1}}},$$
 $${\phi _{ljk}} = \frac{{\sqrt {P_j^{{\rm{ul}}}P_k^{{\rm{ul}}}{\beta _{lj}}{\beta _{lk}}} }}{{\sum\limits_{U_i \in {{\cal U}}} {P_i^{{\rm{ul}}}{\beta _{li}}}\delta(\pi_k-\pi_i) + {1}}},$$
 ${{\mathbf{\widetilde v}}_{lk}} = {{\mathbf{V}}^*_l}\boldsymbol\psi _{\pi_k}$ with $[\widetilde{\mathbf{v}}_{lk}]_m \sim \mathcal{CN}(0,1)$. In (\ref{estimated_CSI}), the first term is the desired user's channel, the second term is the leakage from the other users' channels, called the \emph{pilot contamination} (PC) effect due to the pilot signal reuse, and the third term is the noise part. Invoking the orthogonality principle of the MMSE estimator \cite{LiPilot}, $\mathbf{h}_{lk}$ can be decomposed as $\mathbf{h}_{lk}=\widehat{\mathbf{h}}_{lk}+\widetilde{\mathbf{h}}_{lk}$, where $\widehat{\mathbf{h}}_{lk}\sim\mathcal{CN}(\mathbf{0},\phi_{lkk}\mathbf{I}_M)$ and $\widetilde{\mathbf{h}}_{lk}\sim\mathcal{CN}(\mathbf{0},(1-\phi_{lkk})\mathbf{I}_M)$ are mutually independent.
 Note that the estimated version of $\mathbf{g}_{lk}$ at BS $l$ is given as $\widehat{\mathbf{g}}_{lk} = \sqrt{\beta_{lk}}\widehat{\mathbf{h}}_{lk}$ for $\forall k\in\mathcal{U}_l$ or $\widehat{\mathbf{g}}_{lk} = \mathbf{0}_{M\times 1}$ for $\forall k\notin\mathcal{U}_l$ and thus  the estimated version of $\mathbf{G}$ is given as $\widehat{ {\mathbf{G}}} = \left([\widehat{\mathbf{g}}_{lk}]_{l=1,k=1}^{L,K}\right)^H$.

%

\subsection{Performance Measure}
When operation $\Upsilon$ is used, the signal-to-noise ratio (SNR), $\mathsf{SNR}_k^\Upsilon$, the signal-to-interference ratio (SIR), $\mathsf{SIR}_k^\Upsilon$, and the signal-to-interference-plus-noise ratio (SINR), $\mathsf{SINR}_k^\Upsilon$, are respectively defined as 
\begin{align*}
&\mathsf{SNR}_k^\Upsilon = {{{Q^\Upsilon_k}{{\left| {{\psi^\Upsilon_{kk}}} \right|}^2}}},\\ 
&\mathsf{SIR}_k^\Upsilon = \frac{{{Q^\Upsilon_k}{{\left| {{\psi^\Upsilon_{kk}}} \right|}^2}}}{{\sum\limits_{U_j \in {\cal U}\backslash \{ U_k\} } {{Q^\Upsilon_j}{{\left| {{\psi^\Upsilon_{kj}}} \right|}^2}} }},\\
&\mathsf{SINR}_k^\Upsilon = \frac{{{Q^\Upsilon_k}{{\left| {{\psi^\Upsilon_{kk}}} \right|}^2}}}{{\sum\limits_{U_j \in {\cal U}\backslash \{ U_k\} } {{Q^\Upsilon_j}{{\left| {{\psi^\Upsilon_{kj}}} \right|}^2}} + {1}}}.
\end{align*}
Obviously, $\mathsf{SNR}_k^\Upsilon$, $\mathsf{SIR}_k^\Upsilon$, and $\mathsf{SINR}_k^\Upsilon$ are random variables depending on the realization of the short-term fading and the long-term fading (i.e., realization of locations of users and BSs).

One of the main objectives of this paper is to characterize the asymptotic behavior of $\mathsf{SINR}^\Upsilon$ when the key network parameters such as the number of BSs $L$, the number of users $K$, and the number of BS antennas $M$ are scaled up. To do this, we define an auxiliary parameter $N=LM$ as the network size (or equivalently, the total number of antennas in the network) and make the following relations: 
\begin{align}\label{eq_16}
L=\Theta(N^{\eta_{\rm{bs}}}),~ M =\Theta(N^{\eta_{\rm{ant}}}),~ K=\Theta(N^{\eta_{\rm{user}}}),
\end{align}
where $\eta_{\rm{bs}}$, $\eta_{\rm{user}}$, and $\eta_{\rm{ant}}$ denote the scaling exponents of the numbers of BSs, users, and BS antennas, respectively. Note that we only consider the case where $0\le \eta_{\rm{bs}},\eta_{\rm{ant}},\eta_{\rm{user}}\le 1$ and $\eta_{\rm{bs}}+\eta_{\rm{ant}} = 1$ by definition. Then, the asymptotic performance of the network can be characterized as follows.

\begin{definition}[Performance Measure] The scaling exponent of the SINR of operation $\Upsilon$, $\mathsf{sinr}^{\Upsilon}$, is the order of growth of the SINR of a randomly selected user as $N$ increases such that, for any $\epsilon>0$, 
\begin{equation}\label{eq_17}
\mathop {\lim }\limits_{N \to \infty } \Pr \left( {\left| {\frac{{\log {\mathsf{SINR}_k^{\Upsilon}}}}{{\log N}} - {\mathop{\mathsf {sinr}}^{\Upsilon}\nolimits} } \right| < \epsilon } \right) = 1.
\end{equation}
\end{definition}

Manipulating (\ref{eq_17}), we can also obtain 
\[
\mathop {\lim }\limits_{N \to \infty } \Pr \left( {{N^{{\mathsf{sinr}^{\Upsilon}} - \epsilon }} < \mathsf{SINR}^{\Upsilon}_k < {N^{{\mathsf{sinr}^{\Upsilon}} + \epsilon }}} \right) = 1,
\]
which implies that the SINR of a randomly selected user is bounded by $[N^{\mathsf{sinr}^{\Upsilon}-\epsilon},N^{\mathsf{sinr}^{\Upsilon}+\epsilon}]$ as $N$ increases. Invoking our order notations, we can simply write $\mathsf{SINR}_k^{\Upsilon} = \Theta(N^{\mathsf{sinr}^{\Upsilon}})$ and the sum-rate served by a network can also be written as $C_{\Sigma} = \Theta(N^{\eta_{\rm{user}}}\log_2(1+N^{\mathsf{sinr}^{\Upsilon}}))$. Our main goal in this paper is to characterize $\mathsf{sinr}^{\Upsilon}$ for various LS-CRAN operations under practical constraints so that how the target SINR and the number of users satisfying the target can simultaneously grow as the network size increases.

Similarly, we define $\mathsf{snr}^{\Upsilon}$ and $\mathsf{sir}^{\Upsilon}$ as the scaling exponents of the SNR and SIR, respectively. Note that it is sufficient to find $\mathsf{snr}^{\Upsilon}$ and $\mathsf{sir}^{\Upsilon}$ and then $\mathsf{sinr}^{\Upsilon} = \min\{\mathsf{snr}^{\Upsilon},\mathsf{sir}^{\Upsilon}\}$ is obtained since $\mathsf{SINR}_k^\Upsilon$ is equal to the harmonic mean of $\mathsf{SNR}_k^\Upsilon$ and $\mathsf{SNR}_k^\Upsilon$, i.e, $(\mathsf{SINR}_k^\Upsilon)^{-1} = (\mathsf{SNR}_k^\Upsilon)^{-1}+( \mathsf{SIR}_k^\Upsilon)^{-1}$.

\section{LS-CRAN Operations Under Practical Limitations}
In this section, the ideal IF operation is introduced as a reference system and practical cooperative operations are reviewed. Note that a comprehensive review on the operations is beyond the scope of this paper so that two well-known practical operations, MRT operation \cite{LoMaximum} and ZF operation \cite{SpencerZero}, are focused with the three practical limitations.

\subsection{Cooperative Transmission Operations}
\subsubsection{Ideal IF Operation}
As a reference, the ideal IF operation is considered, where the interference term in (\ref{received_signal_from_BS_l}) is removed by Genie perfectly without any cost while the desired signal power is maximized by using the IF precoding matrix, 
\begin{equation}
 {\mathbf{F}}^{\mathsf{if}} =  {\widehat{\mathbf{G}}}^H.
\end{equation}
Obviously, the ideal IF operation provides an upper-bound on the performance of any practical operation. Since this operation cannot be realizable, $\mathcal{B}_l^{\mathsf{if}}$ and $\mathcal{F}_l^{\mathsf{if}}$ are not defined.

\subsubsection{MRT Operation}
MRT operation tries to maximize the received signal power of a desired user without considering the effect of interference to undesired users \cite{LoMaximum}. This operation is regarded as a good candidate as the number of BS antennas increases due to its low-computational complexity and low-overhead requirement \cite{Argos}.
The precoder for MRT operation is given by 
\begin{equation}
{\mathbf{F}}^{\mathsf{mrt}} = {\widehat{\mathbf{G}}}^H.
\end{equation}
Since the precoder of MRT operation does not need information exchange among BSs, the backhaul and fronthaul information sets can be expressed as ${\mathcal{B}_l^{\mathsf{mrt}}} = \emptyset$, and $\mathcal{F}_l^\mathsf{mrt}=\left\{\sqrt{Q_j^{\mathsf{mrt}}}s_j|U_j\in{\mathcal{U}_l}\right\}$, respectively.

\subsubsection{ZF Operation}
ZF operation can cancel the interference term in (\ref{received_signal_from_BS_l}) (perfectly, provided that the perfect channel estimation is available at the CPU) at the expense of the desired signal power loss \cite{SpencerZero}. 
The precoder for ZF operation is given by 
\begin{equation}
{\mathbf{F}}^{\mathsf{zf}} =  {\widehat{\mathbf{ G}}}^H{\left( {\widehat{ {\mathbf{G}}}{\widehat{ {\mathbf{ G}}}}^H} \right)^{ - 1}}.
\end{equation}
Note that ZF operation can be used only when the number of antennas in the system is larger than or equal to that of users, i.e., $LM\ge K$ and the estimated channel matrix, $\widehat{\mathbf{G}}$, has full-rank.
To construct the ZF precoder, the CPU requires to know the estimated channel matrix $\widehat{\mathbf{G}}$ so that the short-term CSIs of all associated users at each BS need to be conveyed to the CPU via the dedicated front/backhaul link. Thus, the backhaul and fronthaul information sets can be expressed as ${\mathcal{B}_l^{\mathsf{zf}}} = \left\{ \widehat{\mathbf{h}}_{lj} | U_j\in\mathcal{U}_l\right\}$ and $\mathcal{F}_l^{\mathsf{zf}}=\left\{\mathbf{f}_{lj}^{\mathsf{zf}}\sqrt{Q^{\mathsf{zf}}_j}s_j|U_j\in{\mathcal{U}_l}\right\}$, respectively. 

Note that it is well-known that the performance of ZF operation is better than that of MRT operation if the network has sufficient transmit power. But, if not, ZF operation can be inferior to MRT operation. In this paper, we will show that ZF operation is always superior or identical to MRT operation in the viewpoint of the scaling exponent of the SINR regardless of the operating transmit power. 

\subsection{Practical Limitations}
In a real-world network, there exist various practical limitations. In this paper, we consider  limitations on the total transmit power, the front/backhaul capacity, and the pilot resource. 
\subsubsection{limited total transmit power} 
The major merit of the LS-CRAN is that it can decrease the transmit power consumption so that it is suitable to a green communication. So, the total transmit power is a key constraint in a future cellular system. In order to limit the total transmit power of uplink or downlink, similarly as in (\ref{eq_16}), we make an additional asymptotic relationship as 
\begin{equation}
\begin{split}
P_j^{\rm{ul}} &= \Theta(N^{\rho^{\rm{ul}}})\text{ for UL and }
P_j^{\rm{dl}} =\Theta( N^{\rho^{\rm{dl}}})\text{ for DL,}
\end{split}
\end{equation}
where $\rho^{\rm{ul}}$ and $\rho^{\rm{dl}}$ denote the scaling exponents of the UL and DL transmit powers, respectively. Also, the total transmit power is $P_\Sigma = \sum_{j=1}^K (P_j^{\rm{ul}}+P_j^{\rm{dl}}) = \Theta(N^{\eta_{\rm{user}}+\max\{\rho^{\rm{ul}},\rho^{\rm{dl}}\}})$.

\subsubsection{limited front/backhaul capacity} 
Since the front/backhaul information sets closely depend on the cardinality of the set $\mathcal{X}_j$, i.e., the number of BSs serving user $j$, so that the front/backhaul link overhead can be quantified from the network association state. Define $N^{\rm{PA}}_{j} = |\mathcal{X}_j|$ as the number of BSs associated to user $j$.  Then, similarly as in (\ref{eq_16}), we make an additional asymptotic relationship as
\begin{equation}
N^{\rm{PA}} = \Theta(N^{\upsilon_{\rm{PA}}}),
\end{equation}
where $\upsilon_{\rm{PA}}$ denotes the scaling exponent of the number of associated BSs per user. Since at most $L$ BSs can be associated to each user, $0\le {\upsilon_{\rm{PA}}}\le \eta_{\rm{bs}}$.

\subsubsection{limited pilot resource} 
Due to the natural time-frequency selectivity of wireless channel, it is required to acquire CSI of users for every coherence interval. However, the number of dedicated orthogonal pilots is strictly limited by the number of orthogonal resources in one coherence interval so that the number of pilots (pilot resources), $T$, should be limited appropriately. In this paper, similarly as in (\ref{eq_16}), we make an additional asymptotic relationship as
\begin{equation}
T = \Theta(N^{\upsilon_{\rm{PR}}}),
\end{equation}
where $\upsilon_{\rm{PR}}$ denotes the scaling exponent of the number of pilot resources. Since at most $K$ pilot sequences are sufficient to guarantee no pilot reuse, $ 0 \le \upsilon_{\rm{PR}}\le \eta_{\rm{user}}$. 
Note that $\upsilon_{\rm{PR}}>0$ implies that the available number of pilot sequences increases as the network size increases.

\section{Scaling Exponents of the SINR}
In order to quantify the effect of the CSI accuracy according to the UL transmit power, we consider the following four regimes.
\begin{itemize}
\item Case \textsf{EH} ($\rho^{\rm{ul}}\ge0$): the UL transmit power is sufficiently high so that every BS can acquire the accurate CSI of a randomly selected users. 
\item Case \textsf{H} ($-\frac{\alpha}{2}\eta_{\rm{bs}}\le \rho^{\rm{ul}}<0$): each BS can acquire the accurate CSI of a randomly selected users within a distance of $\Theta(N^{\rho^{\rm{ul}}/\alpha})$. 
\item Case \textsf{M} ($-\frac{\alpha}{2}\eta_{\rm{bs}}-\eta_{\rm{ant}} \le \rho^{\rm{ul}}< -\frac{\alpha}{2}\eta_{\rm{bs}}$): randomly selected user's CSI is erroneous even at the nearest BS but is still meaningful for providing an array gain.
\item Case \textsf{L} ($\rho^{\rm{ul}}< -\frac{\alpha}{2}\eta_{\rm{bs}}-\eta_{\rm{ant}}$): randomly selected user's CSI becomes quite poor even at the nearest BS so that no array gain can be provided. 
\end{itemize}


\begin{figure*}
\centering
\subfigure[$\eta_{\rm{user}}>\eta_{\rm{bs}}$]{\includegraphics[width=8cm]{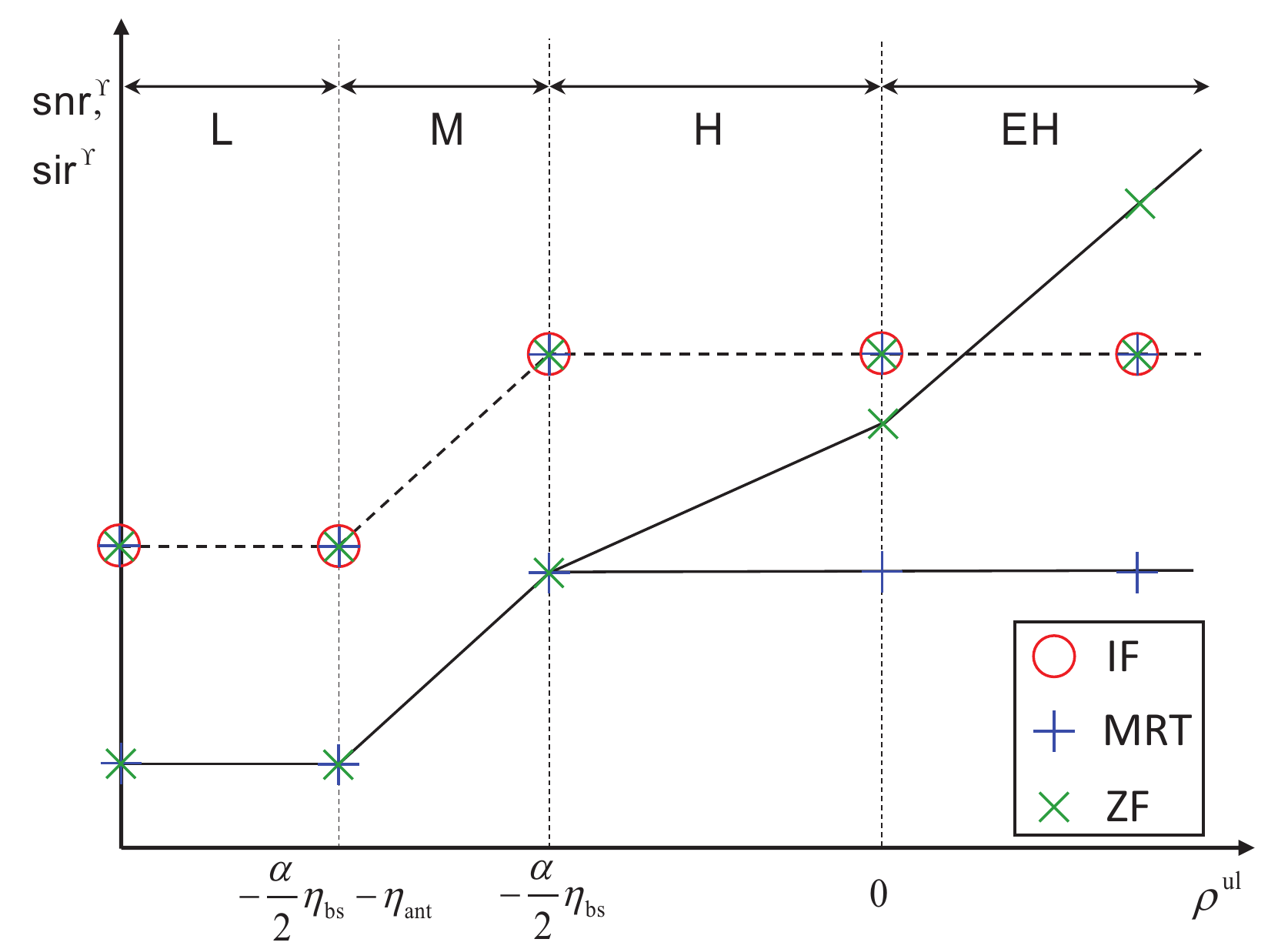}}
\subfigure[$\eta_{\rm{user}}\le\eta_{\rm{bs}}$]{\includegraphics[width=8cm]{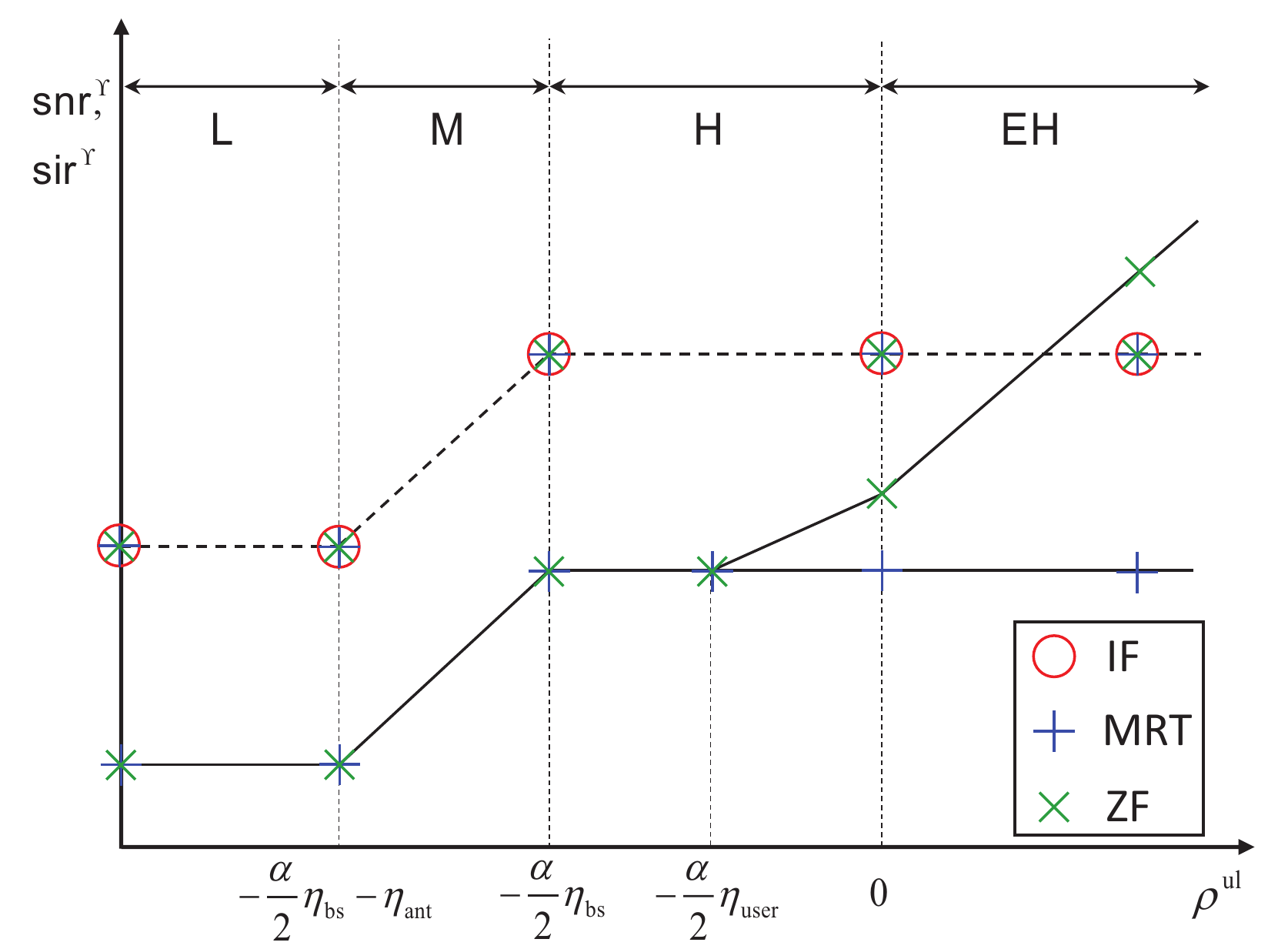}}
\caption{SNR (dashed line) or SIR (solid line) scaling exponents for IF ($\mathsf{o}$-marker), MRT ($\mathsf{+}$-marker), or ZF ($\mathsf{x}$-marker) operations according to $\rho^{\rm{ul}}$.}
\end{figure*}

\begin{figure*} 
\centering
\subfigure[$\rho^{\rm{ul}}<-\frac{\alpha}{2}$]{\includegraphics[width=8cm]{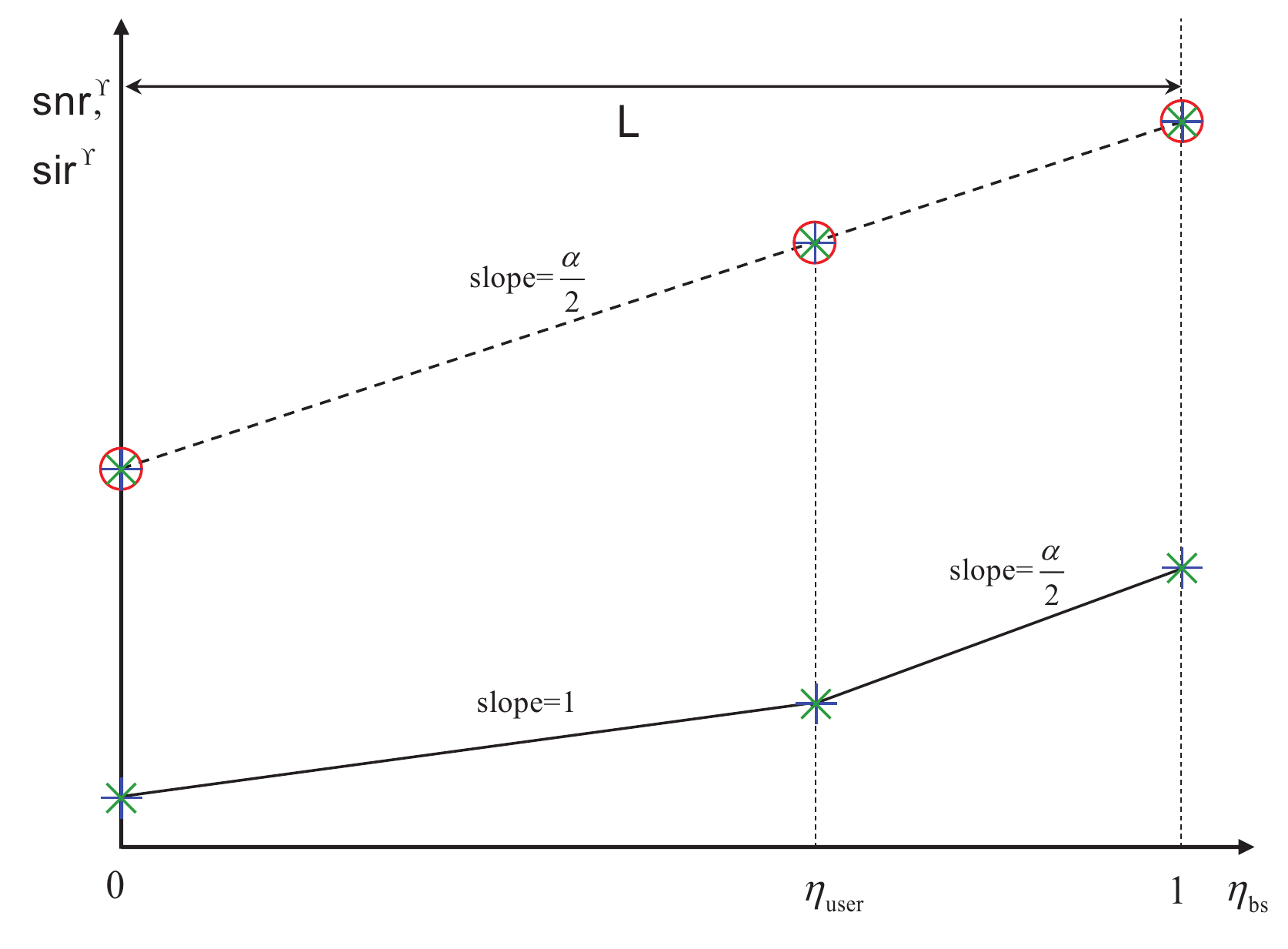}}
\subfigure[$-\frac{\alpha}{2}\le \rho^{\rm{ul}}<-1$]{\includegraphics[width=8cm]{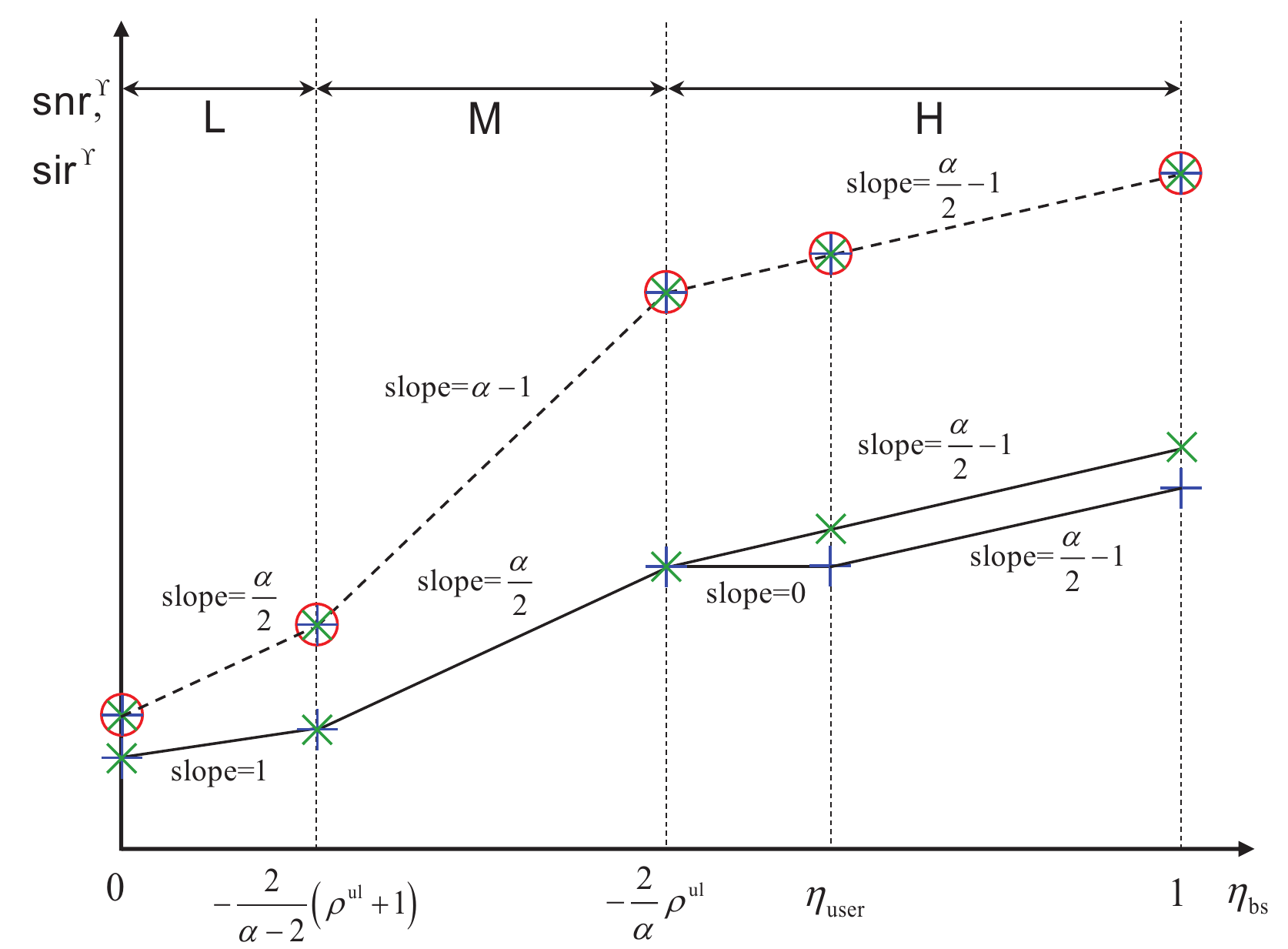}}\\
\subfigure[$-1\le \rho^{\rm{ul}}<0 $]{\includegraphics[width=8cm]{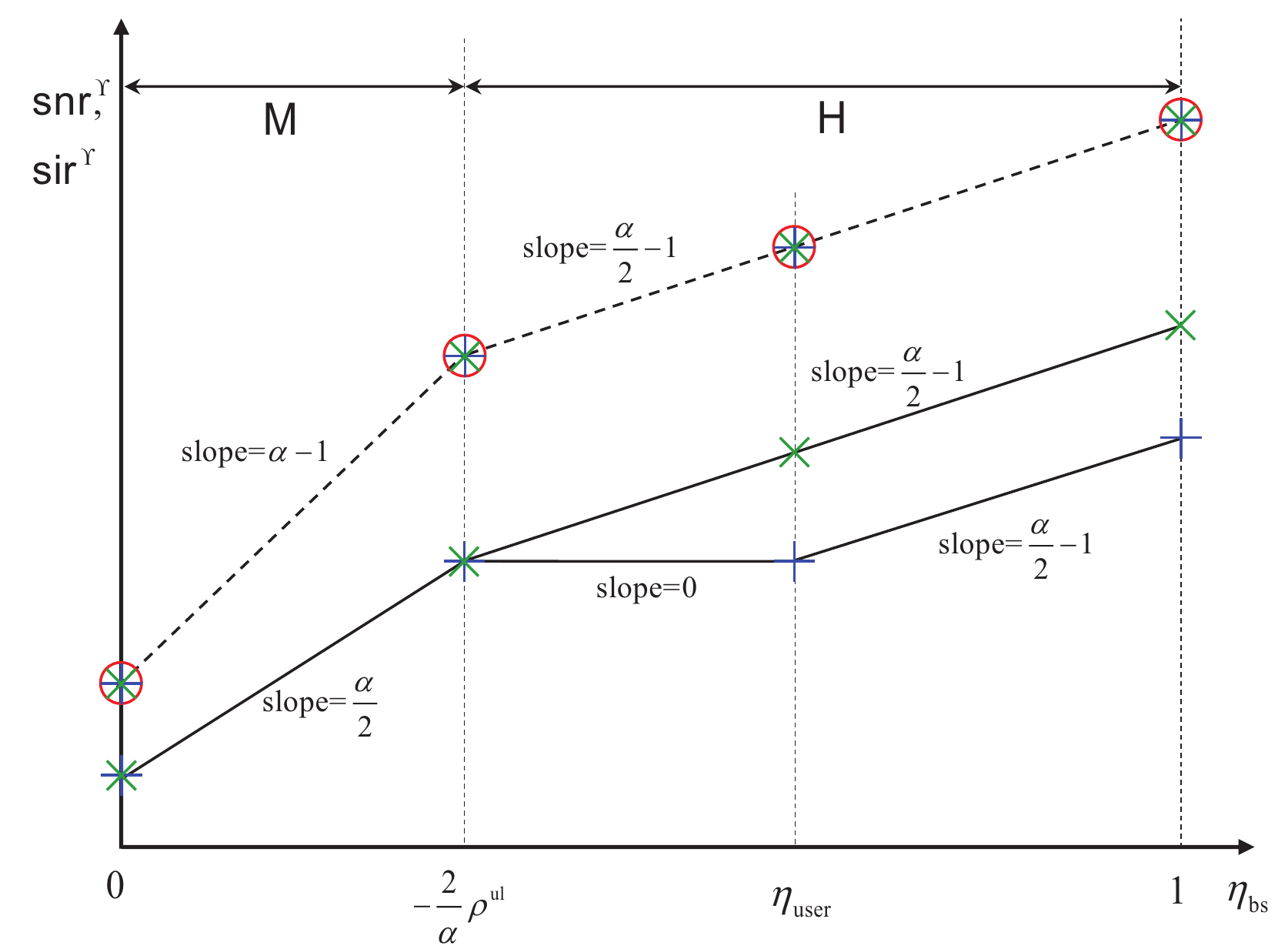}}
\subfigure[$ 0 \le \rho^{\rm{ul}}$]{\includegraphics[width=8cm]{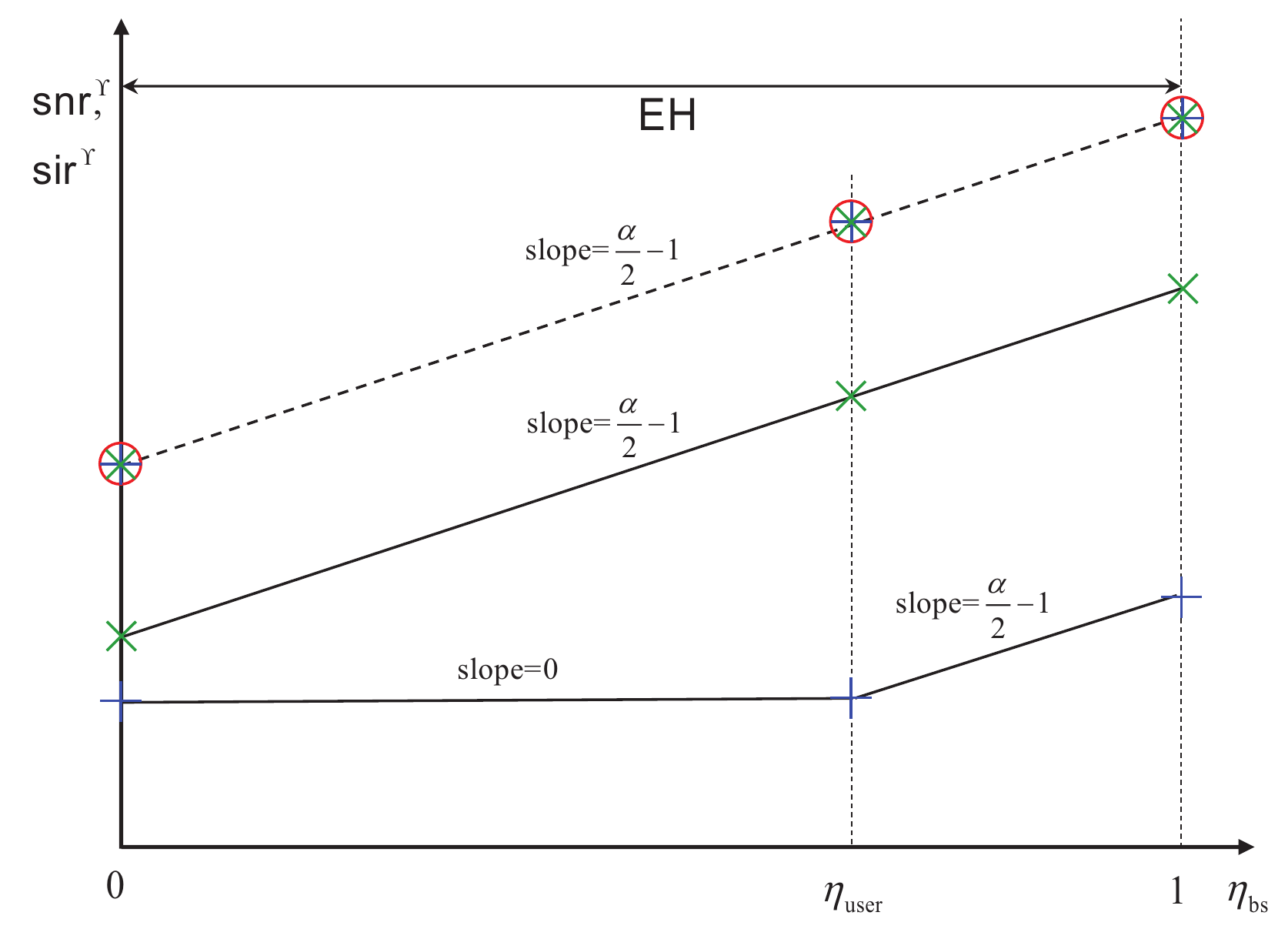}}
\caption{SNR (dashed line) or SIR (solid line) scaling exponents for IF ($\mathsf{o}$-marker), MRT ($\mathsf{+}$-marker), or ZF ($\mathsf{x}$-marker) operations according to $\eta_{\rm{bs}}$.}
\end{figure*}

\subsection{Without the Practical Limitations}
\begin{theorem}  Suppose that IF operation is used with a full association ($\upsilon_{\rm{PA}} = \eta_{\rm{bs}}$) and no pilot reuse ($\upsilon_{\rm{PR}} = \eta_{\rm{user}}$). Then, the scaling exponents are respectively given by 
\begin{align} \label{if_snr_scaling}
\mathsf{snr}^{\mathsf{if}} &= \rho ^{{\rm{dl}}} + \frac{\alpha}{2} {\eta_{\rm{bs}}} + \Xi, \\ 
\label{if_sir_scaling}
\mathsf{sir}^{\mathsf{if}} &= \infty, \\ 
\label{if_sinr_scaling}
\mathsf{sinr}^{\mathsf{if}} &= \mathsf{snr}^{\mathsf{if}},
\end{align}
where $\Xi = \left(  \rho^{\rm{ul}}+{\frac{\alpha }{2}{\eta_{\rm{bs}}} + {\eta_{\rm{ant}}} } \right)^+- \left(  \rho^{\rm{ul}}+{\frac{\alpha }{2}{\eta_{\rm{bs}}} } \right)^+$ denotes the array gain and $(x)^+=\max\{x,0\}$.
\end{theorem}
\begin{IEEEproof}
Please see Appendix A.
\end{IEEEproof}
Theorem 1 is also illustrated in Figs. 2 and 3 according to $\rho^{\rm{ul}}$ and $\eta_{\rm{bs}}$, respectively. 
Intuitively, the SNR of IF operation is composed of the three parts as
\begin{equation}\label{snr_of_if}
\mathsf{SNR}_k^{\mathsf{if}} \asymp \underbrace{N^{\rho^{\rm{dl}}}}_{\text{DL transmit power}} \times  \underbrace{N^{\frac{\alpha}{2}\eta_{bs}}}_{\text{densification gain}}\times\underbrace{N^{\Xi}}_{\text{array gain}},
\end{equation}
where the first part is the DL transmit power, the second part is the densification gain which comes from the decrease of the access distance of  $\Theta(N^{-\frac{1}{2}\eta_{\rm{bs}}})$ and the last part is the array gain of a coherent transmission which depends on the CSI accuracy and thus the UL transmit power.

\begin{remark} [SNR behavior of IF operation]
Fig. 2 reveals how the UL transmit power affects on the SNR behavior. In \textsf{EH} and \textsf{H}, the full array gain ($\Xi=\eta_{\rm{ant}}$) is achieved so that an additional UL transmit power does not improve the quality of DL service in the network, i.e., is wasteful. In \textsf{M}, a partial array gain ($ 0 \le\Xi< \eta_{\rm{ant}}$) depending on the UL transmit power is obtained so that the network total power needs to be consumed by considering both the DL transmit power and the CSI accuracy. In \textsf{L}, no array gain ($\Xi = 0$) is obtained due to poor CSI accuracy so that the quality of DL service becomes irrelevant to the UL transmit power and its performance is identical to the random beamforming without small-scale CSIs in \cite{SharifRBF}.

Fig. 3 shows the SNR behavior according to the BS scaling exponent. It turns out that additional BSs (even with smaller number of BS antennas at each BS) are always beneficial but the slope of $\mathsf{snr^{if}}$ (vs. $\eta_{\rm{bs}}$) varies according to $\rho^{\rm{ul}}$. The slope becomes $\frac{\alpha}{2}$ in \textsf{L}, $\alpha-1$ in \textsf{M}, and $\frac{\alpha}{2}-1$ in \textsf{H} or \textsf{EH}, which implies that additional BSs (while keeping the network size fixed) become the most effective in \textsf{M} because the additional BSs improve not only the densification gain but also the array gain and the least effective in \textsf{H} or \textsf{EH} because only the densification gain is improved.  $\blacksquare$
\end{remark}



\begin{theorem} Suppose that MRT operation is used with a full association ($\upsilon_{\rm{PA}} = \eta_{\rm{bs}}$) and no pilot reuse ($\upsilon_{\rm{PR}} = \eta_{\rm{user}}$). Then, the scaling exponents are respectively given by 
\begin{align} \label{mrt_snr_scaling}
\mathsf{snr}^\mathsf{mrt} &= \mathsf{snr}^\mathsf{if} ,\\\label{mrt_sir_scaling}
\mathsf{sir} ^\mathsf{mrt} &= \mathsf{snr^{mrt}} - \Delta^\mathsf{mrt},\\ \label{mrt_sinr_scaling}
\mathsf{sinr}^\mathsf{mrt}&= \mathsf{snr}^\mathsf{mrt} - \left( \Delta^\mathsf{mrt}\right)^+,
\end{align}
where $\Delta^\mathsf{mrt}= {\rho ^{{\rm{dl}}}}+{{ \frac{\alpha }{2}\min \left\{ {{\eta _{{\rm{bs}}}},{\eta _{{\rm{user}}}}} \right\} + {{\left( {{\eta _{{\rm{user}}}} - {\eta _{{\rm{bs}}}}} \right)}^ + }} }$.
\begin{IEEEproof}
Please see Appendix B.
\end{IEEEproof}
\end{theorem}

Theorem 2 is also illustrated in Figs. 2 and 3 according to $\rho^{\rm{ul}}$ and $\eta_{\rm{bs}}$, respectively.
Interestingly, although $\mathsf{snr^{mrt}}$ is identical to $\mathsf{snr^{if}}$, the gap between $\mathsf{sir^{mrt}}$ and $\mathsf{snr^{mrt}}$, denoted by $\Delta^\mathsf{mrt}$, changes according to the sign of $\eta_{\rm{bs}}-\eta_{\rm{user}}$. 
From the definition of $\Delta^{\mathsf{mrt}}$, the interference caused in MRT operation at a randomly selected user $k$ can be represented as 
\begin{equation}
\begin{split}
I_k &\asymp N^{\Delta^{\mathsf{mrt}}}\\
&=  \underbrace{N^{\rho^{\rm{dl}}}}_{\scriptstyle\text{DL transmit}\atop\scriptstyle\text{power}} \times  \underbrace{N^{\frac{\alpha}{2}\min\{\eta_{\rm{bs}},\eta_{\rm{user}}\}}}_{\scriptstyle \text{received power of} \atop \scriptstyle \text{a dominant interferer}}\times \underbrace{N^{(\eta_{\rm{user}}-\eta_{\rm{bs}})^+}}_{\scriptstyle\text{\# of dominant}\atop\scriptstyle\text{interferers}},
\end{split}
\end{equation}
where the second and third parts depend on the sign of $\eta_{\rm{user}} - \eta_{\rm{bs}}$. 
When $\eta_{\rm{user}}\ge \eta_{\rm{bs}}$, one BS should simultaneously serve $\Theta(N^{\eta_{\rm{user}}-\eta_{\rm{bs}}})$ users apart by $\Theta(N^{-\frac{1}{2}\eta_{\rm{bs}}})$ so that there are $\Theta(N^{\eta_{\rm{user}}-\eta_{\rm{bs}}})$ interferer whose received power is $\Theta(N^{\frac{\alpha}{2}\eta_{\rm{bs}}})$. When $\eta_{\rm{user}}< \eta_{\rm{bs}}$, the dominant interference comes from the  BS apart by $\Theta(N^{-\frac{1}{2}\eta_{\rm{user}}})$ so that the received power of the dominant interference is $\Theta(N^{\frac{\alpha}{2}\eta_{\rm{user}}})$.

\begin{remark} [Asymptotical optimality of MRT operation]
In order for MRT operation to behave as IF operation asymptotically (i.e., $\Delta^{\mathsf{mrt}}\le 0$), the DL transmit power should be limited as 
\begin{equation}\label{eq_23_mrt}
{\rho ^{{\rm{dl}}}} \le \left\{ {\begin{array}{*{20}{l}}
{ - \frac{\alpha }{2}{\eta _{{\rm{user}}}}},&{{\text{if }}{\eta _{{\rm{bs}}}} \ge {\eta _{{\rm{user}}}}},\\
{ - \frac{\alpha }{2} {\eta _{{\rm{bs}}}}-(\eta_{\rm{user}}-\eta_{\rm{bs}})},&{{\text{if }}{\eta _{{\rm{bs}}}} < {\eta _{{\rm{user}}}}},
\end{array}} \right.
\end{equation}
which gives us the following insights.
\begin{itemize}
\item For MRT operation, the IF optimality condition depends only on the numbers of users ($\eta_{\rm{user}}$) and BSs ($\eta_{\rm{bs}}$) and the DL transmit power ($\rho^{{\rm{dl}}}$), but is independent to the number of BS antennas ($\eta_{\rm{ant}}$) and the UL transmit power ($\rho^{\rm{ul}}$). 
\item When the number of antennas in a BS is much larger than the total number of users in a multi-cell network, MRT operation becomes interference-free asymptotically as expected in literature \cite{RusekScaling}. However, this is of little interest because the network size is too large (or the number of users is too small). Instead, Theorem 2 and (\ref{eq_23_mrt}) gives more insightful IF optimality condition for MRT operation being asymptotically interference-free in a multi-cell network. $\blacksquare$
\end{itemize}
\end{remark}

\begin{remark} [SIR behavior of MRT operation]
As can be seen from Fig. 3, additional BSs are beneficial (while fixing $N$) in most cases and the slope of $\mathsf{sir^{mrt}}$ (vs. $\eta_{\rm{bs}}$) can be $0$, $1$, $\frac{\alpha}{2}-1$ and $\frac{\alpha}{2}$. Interestingly, when $-\frac{2}{\alpha}\rho^{\rm{ul}}\le\eta_{\rm{bs}}\le\eta_{\rm{user}}$, the slope of $\mathsf{sir^{mrt}}$ becomes 0, which means that additional BSs are wasteful as long as $\eta_{\rm{bs}}$ is within that interval. On the other hand, the slope becomes $\frac{\alpha}{2}$ in \textsf{M} and in $\textsf{L}$ if $\eta_{\rm{bs}}>\eta_{\rm{user}}$, in which additional BSs are the most effective. $\blacksquare$
\end{remark}


\begin{theorem}
Suppose that ZF operation is used with a full association and no pilot reuse. Then, the scaling exponents are respectively given by 
\begin{align} \label{zf_snr_scaling}
\mathsf{snr^{zf}} &= \mathsf{snr^{if}},\\ \label{zf_sir_scaling}
\mathsf{sir^{zf}} &= \mathsf{snr^{zf}}-\Delta^{\mathsf{zf}} \\ \label{zf_sinr_scaling}
\mathsf{sinr^{zf}} &= \mathsf{snr^{zf}} - \left( \Delta^\mathsf{zf}\right)^+,
\end{align}
where $\Delta ^\mathsf{zf} = \Delta^{\mathsf{mrt}}- \left( {1 - \frac{2}{\alpha }} \right){\left( {\frac{\alpha }{2}\min \left\{ {{\eta _{{\rm{bs}}}},{\eta _{{\rm{user}}}}} \right\} + {\rho ^{{\rm{ul}}}}} \right)^ + } - \frac{2}{\alpha }{\left( {{\rho ^{{\rm{ul}}}}} \right)^ + }$.
\begin{IEEEproof}
Please see Appendix C.
\end{IEEEproof}
\end{theorem}
Theorem 3 is also illustrated in Figs. 2 and 3 according to $\rho^{\rm{ul}}$ and $\eta_{\rm{bs}}$, respectively. Although the SNR exponent is identical to that in MRT operation, the interference is reduced as the CSI accuracy improves so that the gap between $\mathsf{snr^{zf}}$ and $\mathsf{sir^{zf}}$, $\Delta^{\mathsf{zf}}$, varies as $\rho^{\rm{ul}}$ increases. Again, from the definition of $\Delta^{\mathsf{zf}}$, the interference in MRT operation at a randomly selected user $k$ for $\eta_{\rm{user}}\ge \eta_{\rm{bs}}$ can be represented as
\begin{equation}
\begin{split}
I_k &\asymp  N^{\Delta^{\mathsf{zf}}}
\\
&=\left\{ {\begin{array}{*{20}{c}}
{}& \!\!\!\!\times &{{N^{\frac{\alpha }{2}{\eta _{{\rm{bs}}}}}}}& \!\!\!\!\times &{{N^{{\eta _{{\rm{user}}}} - {\eta _{{\rm{bs}}}}}}},&{\textsf{L,~M}},\\
{{N^{{\rho ^{{\rm{dl}}}}}}}& \!\!\!\!\times &{{N^{ - {\rho ^{{\rm{ul}}}}}}}& \!\!\!\!\times &{{N^{\frac{2}{\alpha }{\rho ^{{\rm{ul}}}} + {\eta _{{\rm{user}}}}}}},&\textsf{H},\\
{\underbrace {}_{{\scriptstyle\text{DL transmit}}\atop\scriptstyle{\text{power}}}}& \!\!\!\!\times &{\underbrace {{N^{ - {\rho ^{{\rm{ul}}}}}}}_{{\scriptstyle\text{received power of}\atop\scriptstyle\text{a dominant interferer}}}}& \!\!\!\!\times &{\underbrace {~~~{N^{{\eta _{{\rm{user}}}}}}~~~}_{{\scriptstyle\text{\# of dominant} \atop\scriptstyle\text{interferers}}}},&\textsf{EH}.
\end{array}} \right.
\end{split}
\end{equation}
The CSI of the users whose access distance is $o(N^{\rho^{\rm{ul}}/\alpha})$ is accurately estimated at the nearest BS, while the CSI of the users whose access distance is $\Omega(N^{\rho^{\rm{ul}}/\alpha})$ is poor. So, when $\rho^{\rm{ul}}<-\frac{\alpha}{2}\eta_{\rm{bs}}$, i.e., \textsf{L} or \textsf{M}, the CSIs of all users are poorly estimated at all BSs because the access distance is $\Omega(N^{-\frac{1}{2}\eta_{\rm{bs}}})$. Thus, the interference cancellation is not effective so that the interference in ZF operation is asymptotically the same to that in MRT operation. However, when $\rho^{\rm{ul}}\ge-\frac{\alpha}{2}\eta_{\rm{bs}}$, i.e., \textsf{H} or \textsf{EH},  accurate CSI is available so that the cancellation operation becomes effective. In \textsf{H}, the dominant interference comes from the BS apart by $\Theta(N^{ {\rho^{\rm{ul}}}/{\alpha}})$ so that the received power of the dominant interference is $\Theta(N^{-\rho^{\rm{ul}}})$. The dominant interferers are located in the doughnut of radii of $\Theta(N^{\rho^{\rm{ul}}/\alpha})$ and $\Theta(N^{\rho^{\rm{ul}}/\alpha+\epsilon})$ with an arbitrarily small $\epsilon$ so that the number of interferers is $\Theta(N^{\frac{2}{\alpha}\rho^{\rm{ul}}+\eta_{\rm{user}}})$. In \textsf{EH}, the number of dominant interferers becomes $\Theta(N^{\eta_{\rm{user}}})$, while the received power of them is still $\Theta(N^{-\rho^{\rm{ul}}})$ which can be explained as follows. Once a BS acquires sufficiently accurate CSI of a user, the interference caused from the user becomes proportional to the channel estimation error, i.e., inversely proportional to the UL transmit power. Thus, in \textsf{H} and \textsf{EH}, the slopes of the reduction in ZF operation are $(1-\frac{2}{\alpha})$ and $1$ with respect to $\rho^{\rm{ul}}$, respectively, as shown in Fig. 3. In \textsf{EH}, every BS can acquire all user's CSI with a sufficiently good accuracy so that the slope becomes $1$. In \textsf{H}, however, BSs far from each user (outside the circle with radius of $\Theta(N^{\rho^{\rm{ul}}/{\alpha}})$) cannot obtain its CSI accurately so that ZF operation does not reduce the interference from those BSs. Note that $\mathsf{sir^{zf}}$ does not improve as $\rho^{\rm{ul}}$ increases within $[-\frac{\alpha}{2}\eta_{\rm{bs}},-\frac{\alpha}{2}\eta_{\rm{user}}]$ in \textsf{H} or \textsf{EH} if $\eta_{\rm{user}}\ge \eta_{\rm{bs}}$ because no actual interference reduction happens during that interval due to the low user density.

\begin{remark} [Asymptotical optimality of ZF operation]
 In order for ZF operation to behave as IF operation asymptotically (i.e., $\Delta^{\mathsf{zf}}\le 0$), the DL transmit power should be limited as follows.

If $\eta_{\rm{bs}}\ge\eta_{\rm{user}}$,
\[{\rho ^{{\rm{dl}}}} \le \left\{ {\begin{array}{*{20}{l}}
{\rho^{\rm{ul}}-\eta_{\rm{user}},}&{{\text{if \textsf{EH}}},}\\
{- \frac{\alpha }{2}{\eta _{{\rm{user}}}} + \left( {1 - \frac{2}{\alpha }} \right){\left( {\frac{\alpha }{2}{\eta _{{\rm{user}}}} + {\rho ^{{\rm{ul}}}}} \right)^ + },}&{{\text{if \textsf{H}}},}\\
{ - \frac{\alpha }{2}{\eta _{{\rm{user}}}},}&{{\text{if \textsf{M} or \textsf{L}}},}
\end{array}} \right.\]

or if $\eta_{\rm{bs}} < \eta_{\rm{user}}$
\begin{equation*}
{\rho ^{{\rm{dl}}}} \le \left\{ {\begin{array}{*{20}{l}}
{{\rho ^{{\rm{ul}}}} - {\eta _{{\rm{user}}}},}&{{\text{if \textsf{EH}}},}\\
{- {\eta _{{\rm{user}}}}+\left( {1 - \frac{2}{\alpha }} \right){\rho ^{{\rm{ul}}}} ,}&{{\text{if \textsf{H}}},}\\
{ - \frac{\alpha} {2}{\eta _{{\rm{bs}}}} - ({\eta _{{\rm{user}}}}-\eta_{\rm{bs}}),}&{{\text{if \textsf{M} or \textsf{L}}},}
\end{array}} \right.~~~~~~~~~~~~~
\end{equation*}
which gives us the following insights.
\begin{itemize}
\item Unlike MRT operation, the IF condition of ZF operation depends on the UL transmit power ($\rho^{\rm{ul}}$) as well as $\eta_{\rm{bs}}$, $\eta_{\rm{user}}$, and $\rho^{\rm{dl}}$, but is still independent to the number of BS antennas ($\eta_{\rm{ant}}$). 
\item 
$\Delta^\mathsf{zf}=\Delta^\mathsf{mrt}$ in \textsf{M} and \textsf{L} while $\Delta^\mathsf{zf}\le\Delta^{\mathsf{mrt}}$ in \textsf{EH} or \textsf{H}, which implies that multi-cell cooperation using ZF operation is useless without a sufficient CSI accuracy. As the UL transmit power increases, ZF operation begins to further reduce the interference caused from other BSs' users as shown in Fig. 2.
\end{itemize}
\end{remark}

\begin{remark} [SIR behavior of ZF operation]
As can be seen again from Fig. 3, additional BSs (while fixing $N$) are always beneficial and although the slope of $\mathsf{sir^{zf}}$ vs. $\eta_{\rm{bs}}$ is identical to that of $\mathsf{sir^{mrt}}$ vs. $\eta_{\rm{bs}}$ in  \textsf{L} or \textsf{M}, it remains $\frac{\alpha}{2}-1$ in  \textsf{H} or \textsf{EH}, unlike MRT operation. $\blacksquare$
\end{remark}

%
%

\subsection{Degradation due to the Practical Limitations}

\begin{theorem} Suppose that operation $\Upsilon$ is used with a partial association and no pilot reuse ($\upsilon_{\rm{PR}} = \eta_{\rm{user}}$), in which each user is associated with $\Theta(N^{\upsilon_{\rm{PA}}})$ nearest BSs. Then, the gap $\Delta^\Upsilon$ in Theorems 2 and 3 changes to 
\begin{align} 
\left.\Delta^\Upsilon_{\rm{PA}} = \Delta^\Upsilon\right|_{\rho^{\rm{ul}}\leftarrow \varpi},
\end{align}
where 
\[\varpi  = \left\{ {\begin{array}{*{20}{l}}
{\min \left\{ {{\rho ^{{\rm{ul}}}},\frac{\alpha }{2}\left( {{\upsilon _{{\rm{PA}}}} - {\eta _{{\rm{bs}}}}} \right)} \right\},}&{{\text{if }}0 \le {\upsilon _{{\rm{PA}}}} < {\eta _{{\rm{bs}}}}},\\
{{\rho ^{{\rm{ul}}}},}&{{\text{if }}{\upsilon _{{\rm{PA}}}} = {\eta _{{\rm{bs}}}}}.
\end{array}} \right.\]
\end{theorem}
\begin{IEEEproof} When a BS associates with only a part of users, the CPU does not obtain the CSIs of the non-associated users from the BS, whose effect is identical to the case where the UL transmit power of user $j$ becomes $P^{\rm{ul}}_j = 0$ at far BSs $X_l\in\mathcal{X}\backslash\mathcal{X}_j$, which can directly prove this theorem.
\end{IEEEproof}
\begin{remark} [effect of the limited front/backhaul capacity]
Theorem 4 informs the relation between performance degradation caused from the erroneous CSI due to the limited network total power and that caused from the partial association due to the limited front/backhaul capacity. Since $\Delta^{\mathsf{mrt}}$ is independent to $\rho^{\rm{ul}}$, any partial association does not affect MRT operation, which implies that MRT operation is asymptotically same to a single-cell operation (no cooperation). 

However, since $\Delta^\mathsf{zf}$ does depend on $\rho^{\rm{ul}}$, the limited front/backhaul capacity degrades the asymptotic performance of ZF operation. In  $\textsf{M}$ or $\textsf{L}$, $\rho^{\rm{ul}}\le \frac{\alpha}{2}(\upsilon_{\rm{PA}}-\eta_{\rm{bs}})$ so that any association range with a positive exponent is wasteful, i.e., cooperation must be confined among finite number of nearby BSs. However, in  $\textsf{H}$ or $\textsf{EH}$, larger association range with exponent up to $\frac{1}{\alpha}\rho^{\rm{ul}}+\frac{1}{2}\eta_{\rm{bs}}$ improves the asymptotic performance, which requires the total front/backhaul capacity up to $\Theta(N^{\eta_{\rm{user}}+\eta_{\rm{bs}}+\eta_{\rm{ant}}+\frac{2}{\alpha}\rho^{\rm{ul}}})$ complex values.  $\blacksquare$
\end{remark}

\begin{theorem} Suppose that $ \Theta(N^{\upsilon_{\rm{PR}}})$ orthonormal pilot sequences are available  in a fully associated network. Then, the scaling exponents are as in Theorems 2 and 3 by replacing $\Xi$ and $\Delta^{\Upsilon}$ with $\Xi_{\rm{PR}}$ and $\Delta^{\Upsilon}_{\rm{PR}}$, respectively, where
\begin{align}
\Xi_{\rm{PR}}=&{{\left( {{\rho ^{{\rm{ul}}}} + \frac{\alpha }{2}{\eta _{{\rm{bs}}}} + {\eta _{{\rm{ant}}}}} \right)}^ + }- {{\left( {{\rho ^{{\rm{ul}}}} + \frac{\alpha }{2}{\eta _{{\rm{bs}}}} + {{\left( {{\eta _{{\rm{user}}}} - \upsilon_{\rm{PR}}  - {\eta _{{\rm{bs}}}}} \right)}^ + }} \right)}^ + },\\
\Delta _{{\rm{PR}}}^\Upsilon  =& \left\{ {\begin{array}{*{20}{c}}
{\max \{ {\Delta ^\Upsilon },{\Delta _{{\rm{PR}}}}\} ,}&{\upsilon_{\rm{PR}}  \le {\eta _{{\rm{user}}}},}\\
{{\Delta ^\Upsilon },}&{\upsilon_{\rm{PR}}  = {\eta _{{\rm{user}}}},}
\end{array}} \right.
\end{align} 
with 
\begin{equation}
\begin{split}
\Delta_{\rm{PR}} = & {\rho ^{{\rm{dl}}}} + \frac{\alpha }{2}\min \{ {\eta _{{\rm{bs}}}},{\eta _{{\rm{user}}}} - \upsilon_{\rm{PR}} \}  + {({\eta _{{\rm{user}}}} - \upsilon_{\rm{PR}}  - {\eta _{{\rm{bs}}}})^ + } + {\Xi _{{\rm{PR}}}}.
\end{split}
\end{equation}

\begin{table*}\fontsize{7pt}{7pt}\selectfont
\center\caption{Operating Regimes and Conner Points of Regions in Theorem 6}
\begin{tabular}{ |c||c|c|c|c| }
\specialrule{.2em}{.1em}{.1em} 
{\multirow{2}{*}{Region}}&{{{Interference Free}}}&{{{Maximum Ratio Transmission}}}&{{{Zero Forcing}}}&{Operating}\\
{}&{$\Upsilon=\mathsf{if}$}&{$\Upsilon=\mathsf{mrt}$}&{$\Upsilon=\mathsf{zf}$}&{Regime}\\
\specialrule{.2em}{.1em}{.1em} 
\multirow{2}{*}{$\mathcal{A}^\Upsilon$}&\multicolumn{3}{c|}{$ u^{\Upsilon}(\rho,\tau)=\rho-\tau+\frac{\alpha}{2}\eta_{\rm{bs}}$}&\multirow{1}{*}{$\begin{array}{c}\text{SNR limited}, \\ \textsf{L}\end{array}$}\\ 
\cline{2-4}
{}&{$\tau<-\eta_{\rm{ant}}$}&\multicolumn{2}{c|}{$\tau<-\eta_{\rm{ant}},~\rho<(1-\frac{\alpha}{2})\eta_{\rm{bs}}$}&{}\\ 
\specialrule{.2em}{.1em}{.1em} 
\multirow{2}{*}{$\mathcal{B}^\Upsilon$}&\multicolumn{3}{c|}{$ u^{\Upsilon}(\rho,\tau)=\rho-\frac{1}{2}\tau+\frac{\alpha}{2}\eta_{\rm{bs}}+\frac{1}{2}\eta_{\rm{ant}}$}&\multirow{1}{*}{$\begin{array}{c}\text{SNR limited}, \\ \textsf{M}\end{array}$}\\ 
\cline{2-4}
{}&{$-\eta_{\rm{ant}}\le\tau<\eta_{\rm{ant}}$}&\multicolumn{2}{c|}{$-\eta_{\rm{ant}}\le\tau<\eta_{\rm{ant}}, \rho<(1-\frac{\alpha}{2})\eta_{\rm{bs}}$}&{}\\
\specialrule{.2em}{.1em}{.1em} 
\multirow{2}{*}{$\mathcal{C}^\Upsilon$}&\multicolumn{3}{c|}{$ u^{\Upsilon}(\rho,\tau)=\rho-\tau+\frac{\alpha}{2}\eta_{\rm{bs}}+\eta_{\rm{ant}}$}&\multirow{1}{*}{$\begin{array}{c}\text{SNR limited}, \\ \textsf{H}\text{ or }\textsf{EH}\end{array}$}\\
\cline{2-4}
{}&{$\tau\ge\eta_{\rm{ant}}$}&\multicolumn{2}{c|}{$\tau  \ge {\eta _{{\rm{ant}}}}, \frac{\alpha }{{2 - \alpha }}\rho  + \tau  \ge \frac{\alpha }{2}{\eta _{{\rm{bs}}}} + {\eta _{{\rm{ant}}}}   $}&{}\\
\specialrule{.2em}{.1em}{.1em} 
\multirow{2}{*}{$\mathcal{D}^\Upsilon$}&{\multirow{2}{*}{N.A.}}&\multicolumn{2}{c|}{$ u^{\Upsilon}(\rho,\tau)=\frac{1}{2}\left( \rho-\tau+(\frac{\alpha}{2}+1)\eta_{\rm{bs}}+\frac{1}{2}\eta_{\rm{ant}}\right)$}&\multirow{1}{*}{$\begin{array}{c}\text{SIR limited}, \\ \textsf{L}\end{array}$}\\ 
\cline{3-4}
{}&{}&\multicolumn{2}{c|}{$\rho  \ge \left( {1 - \frac{\alpha }{2}} \right){\eta _{{\rm{bs}}}},~\rho  + \tau  < \left( {1 - \frac{\alpha }{2}} \right){\eta _{{\rm{bs}}}} - {\eta _{{\rm{ant}}}}$}&{}\\
\specialrule{.2em}{.1em}{.1em} 
\multirow{2}{*}{$\mathcal{E}^\Upsilon$}&{\multirow{2}{*}{N.A.}}&\multicolumn{2}{c|}{$ u^{\Upsilon}(\rho,\tau)={\frac{1}{2}\left( {\rho  - \tau  + \left( {\frac{\alpha }{2} + 1} \right){\eta _{{\rm{bs}}}} + {\eta _{{\rm{ant}}}}} \right)}$}&\multirow{1}{*}{$\begin{array}{c}\text{SIR limited}, \\ \textsf{M}\end{array}$}\\ 
\cline{3-4}
{}&{}&\multicolumn{2}{c|}{$\rho  \ge \left( {1 - \frac{\alpha }{2}} \right){\eta _{{\rm{bs}}}},~ \left( {1 - \frac{\alpha }{2}} \right){\eta _{{\rm{bs}}}} - {\eta _{{\rm{ant}}}} \le \rho  + \tau  < \left( {1 - \frac{\alpha }{2}} \right){\eta _{{\rm{bs}}}} + {\eta _{{\rm{ant}}}}$}&{}\\
\specialrule{.2em}{.1em}{.1em} 
\multirow{2}{*}{$\mathcal{F}^\Upsilon$}&{\multirow{2}{*}{N.A.}}&{$ u^{\Upsilon}(\rho,\tau)=\eta_{\rm{bs}}+\eta_{\rm{ant}}-\tau$}&{$ u^{\Upsilon}(\rho,\tau)=\frac{\alpha}{2(\alpha-1)}\left((1-\frac{2}{\alpha})\rho-\tau+\frac{\alpha}{2}\eta_{\rm{bs}}+\eta_{\rm{ant}}\right)$}&\multirow{1}{*}{$\begin{array}{c}\text{SIR limited}, \\ \textsf{H}\end{array}$}\\ 
\cline{3-4}
&{}&{$\left( {1 - \frac{\alpha }{2}} \right){\eta _{{\rm{bs}}}} + {\eta _{{\rm{ant}}}} \le \rho  + \tau,~ \tau  < {\eta _{{\rm{ant}}}}$}&{$\tau  + \rho  < \frac{\alpha }{2}{\eta _{{\rm{bs}}}} + {\eta _{{\rm{ant}}}},~ \tau  + \rho  \ge \left( {1 - \frac{\alpha }{2}} \right){\eta _{{\rm{bs}}}} + {\eta _{{\rm{ant}}}}$}&{}\\
\specialrule{.2em}{.1em}{.1em} 
\multirow{2}{*}{$\mathcal{G}^\Upsilon$}&{\multirow{2}{*}{N.A.}}&{$ u^{\Upsilon}(\rho,\tau)={\frac{2}{\alpha }\left( { - \tau  + \frac{\alpha }{2}{\eta _{{\rm{bs}}}} + {\eta _{{\rm{ant}}}}} \right)}$}&{$ u^{\Upsilon}(\rho,\tau)={\frac{1}{2}\left( {\rho  - \tau  + \frac{\alpha }{2}{\eta _{{\rm{bs}}}} + {\eta _{{\rm{ant}}}}} \right)}$}&\multirow{1}{*}{$\begin{array}{c}\text{SIR limited}, \\ \textsf{EH}\end{array}$}
\\
\cline{3-4}
{}&{}&{$\tau  \ge {\eta _{{\rm{ant}}}},~ \frac{\alpha }{{2 - \alpha }}\rho  + \tau  < \frac{\alpha }{2}{\eta _{{\rm{bs}}}} + {\eta _{{\rm{ant}}}}$}&{$\tau  - \rho  < \frac{\alpha }{2}{\eta _{{\rm{bs}}}} + {\eta _{{\rm{ant}}}},~ \tau  + \rho  > \frac{\alpha }{2}{\eta _{{\rm{bs}}}} + {\eta _{{\rm{ant}}}}$}&{}
\\
\specialrule{.2em}{.1em}{.1em} 
\end{tabular}
\end{table*}

\end{theorem}
\begin{proof}
First, consider how the array gain $\Xi$ changes according to the power leakage due to the pilot reuse.
When $\upsilon_{\rm{PR}} > \eta_{\rm{user}}-\eta_{\rm{bs}}$, no pilot reuse among nearby users (within a typical nearby BS range) so that the transmit power leakage due to the pilot contamination effect is negligible, i.e., $\Xi_{\rm{PR}} = \Xi$. It means that a network requires $\Omega(N^{\eta_{\rm{user}}-\eta_{\rm{bs}}})$ orthonormal pilot sequences to achieve the maximum SNR scaling exponent. However, when $\upsilon_{\rm{PR}} \le \eta_{\rm{user}}-\eta_{\rm{bs}}$, $\Xi_{\rm{PR}}$ is upper-bounded by $\eta_{\rm{ant}}+\upsilon_{\rm{PR}}+\eta_{\rm{bs}}-\eta_{\rm{user}}$ due to the non-negligible power leakage (in the same order) to the $\Theta(N^{\eta_{\rm{user}}-\eta_{\rm{bs}}-\upsilon_{\rm{PR}}})$ users using the common pilot.

Now, consider how the gap $\Delta^\Upsilon$ changes. Let $I_k^{\rm{PR}}$ and $I_k^{\rm{noPR}}$ be the interference at a randomly selected user $k$ caused by the users using the common pilot and caused by the other users, respectively. Then, as similarly discussed in (20), 
\begin{equation}
\begin{split}
\!\!I_k^{\rm{noPR}}&\asymp N^{\Delta^{\mathsf{mrt}}} \\
&=\underbrace{N^{\rho^{\rm{dl}}}}_{\scriptstyle\text{DL transmit}\atop\scriptstyle\text{power}}\times\underbrace{N^{\frac{\alpha}{2}\min\{\eta_{\rm{bs}},\eta_{\rm{user}}\}}}_{\scriptstyle \text{received power of} \atop
\scriptstyle \text{a dominant interferer} }\times\underbrace{N^{(\eta_{\rm{user}}-\eta_{\rm{bs}})^+}}_{\scriptstyle\text{\# of dominant} \atop\scriptstyle \text{interferers}}\times \underbrace{N^0}_{\text{array gain}}.
\end{split}
\end{equation}
Now, consider the $\Theta(N^{\eta_{\rm{user}}-\upsilon_{\rm{PR}}})$ interferers using the same pilot. Since they share the same array gain to the desired user, similarly as in (31), we obtain
\begin{equation}
\begin{split}
I_k^{\rm{PR}}\asymp & N^{\Delta_{\rm{PR}}}\\ 
=& \underbrace{N^{\rho^{\rm{dl}}}}_{\scriptstyle\text{DL transmit}\atop\scriptstyle\text{power}}\times\underbrace{N^{\frac{\alpha}{2}\min\{\eta_{\rm{bs}},\eta_{\rm{user}}-\upsilon_{\rm{PR}}\}}}_{\scriptstyle \text{received power of}\atop
\scriptstyle \text{a dominant interferer} } \times \underbrace{N^{(\eta_{\rm{user}}-\upsilon_{\rm{PR}}-\eta_{\rm{bs}})^+}}_{\scriptstyle\text{\# of dominant} \atop\scriptstyle \text{interferers}}\times \underbrace{N^{\Xi_{\rm{PR}}}}_{\text{array gain}}.
\end{split}
\end{equation}
Then, the final step is to find the  maximum between $\Delta^{\mathsf{mrt}}$ and $\Delta_{\rm{PR}}$.
Thus, $\Delta^{\mathsf{mrt}}_{\rm{PR}}$ can be represented as 
\begin{equation}
\Delta _{{\rm{PR}}}^{{\rm{mrt}}} = {\rho ^{{\rm{dl}}}} + \frac{\alpha }{2}{\eta _{{\rm{bs}}}} + \max \left\{ {{\eta _{{\rm{user}}}} - {\eta _{{\rm{bs}}}},\Xi_{\rm{PR}} } \right\}.
\end{equation}
For ZF operation, $I_k^{\rm{noPR}} \asymp N^{\Delta^{\mathsf{zf}}}$ and $I_k^{\rm{PR}}$ is same as in (32), which completes the proof.
\end{proof}
\begin{remark} [Effect of the limited pilot resource]
Theorem 5 informs the performance degradation due to the limited pilot resource. Although the limited pilot or limited backhaul capacity affects the interference (i.e., SIR), the pilot contamination due to the limited pilot resource affects both on the SNR and SIR, where the SNR loss (array gain reduction) is caused by the transmit power leakage to the contaminated users while the SIR loss comes from the fact that no array gain is available over the contaminated users (i.e., the same array gain is available over the background noise). $\blacksquare$
\end{remark}

%


\section{On the Number of Supportable Users}
For a practical network, it would be of the most interest how the target QoS (in terms of the SINR) and the number of users satisfying the target QoS can simultaneously grow as the network size increases. Also, in the future cellular system called the 5th generation, among the most important key performance indicators are the guaranteed (edge-user) throughput and the corresponding connectivity by which the tradeoff between them is of the most interest.

\begin{definition} The scaling exponent of the number of supportable users for operation $\Upsilon$, $\zeta_{\rm{user}}^{\Upsilon}$, is defined as the growth rate of the maximum number of users while guaranteeing the QoS, which is given by
\[\zeta _{{\rm{user}}}^{\Upsilon } = \sup \{ \eta _{{\rm{user}}} |\mathsf{sinr}^\Upsilon \ge \tau,~0\le \eta_{\rm{user}}\le 1 \}, \]
where $\tau$ is the pre-determined QoS requirement on the SINR exponent for the network.
\end{definition}

The case $\tau=0$ means that there are a lot of users and each user in a network requires a fixed data rate as the network size increases. This case can be considered as an Internet-of-Things (IoT) scenario in which there are a lot of devices requiring small data volumes. However, the case $\tau>0$ means that there are a lot of users and all users in a network require much increased data rate as the network size increases, which can reflect future applications requiring high data volumes such as ultra-high definite (UHD) video streaming. Obviously, higher QoS requirement (i.e., higher $\tau$) results in lower number of supportable users (i.e., lower $\zeta_{\rm{user}}^{\Upsilon}$) and vice versa. However, the tradeoff between them as the network size increases is non-trivial and of interest, especially with the three practical constraints. The following theorem gives the answer on this question under the total transmit power constraint.

\begin{theorem}
Suppose that a fully associated network ($\upsilon_{\rm{PA}} = \eta_{\rm{bs}}$) without pilot reuse ($\upsilon_{\rm{PR}} = \eta_{\rm{user}}$) is constrained by the network total transmit power as $P_\Sigma^{\rm{dl}} + P_\Sigma^{\rm{ul}} = \Theta(N^{\rho})$. Then, the scaling exponent of the number of supportable users is  
\begin{equation}\zeta_{\rm{user}}^{\Upsilon}= \max\{\min\{u^\Upsilon(\rho,\tau),1\},0\},
\end{equation} 
where $u^{\Upsilon}$ is given  in Table II.
\end{theorem}
\begin{proof}
See Appendix E for the sketch of the proof.
\end{proof}

\begin{figure}
\centering
\subfigure[IF operation]{\includegraphics[width=8cm]{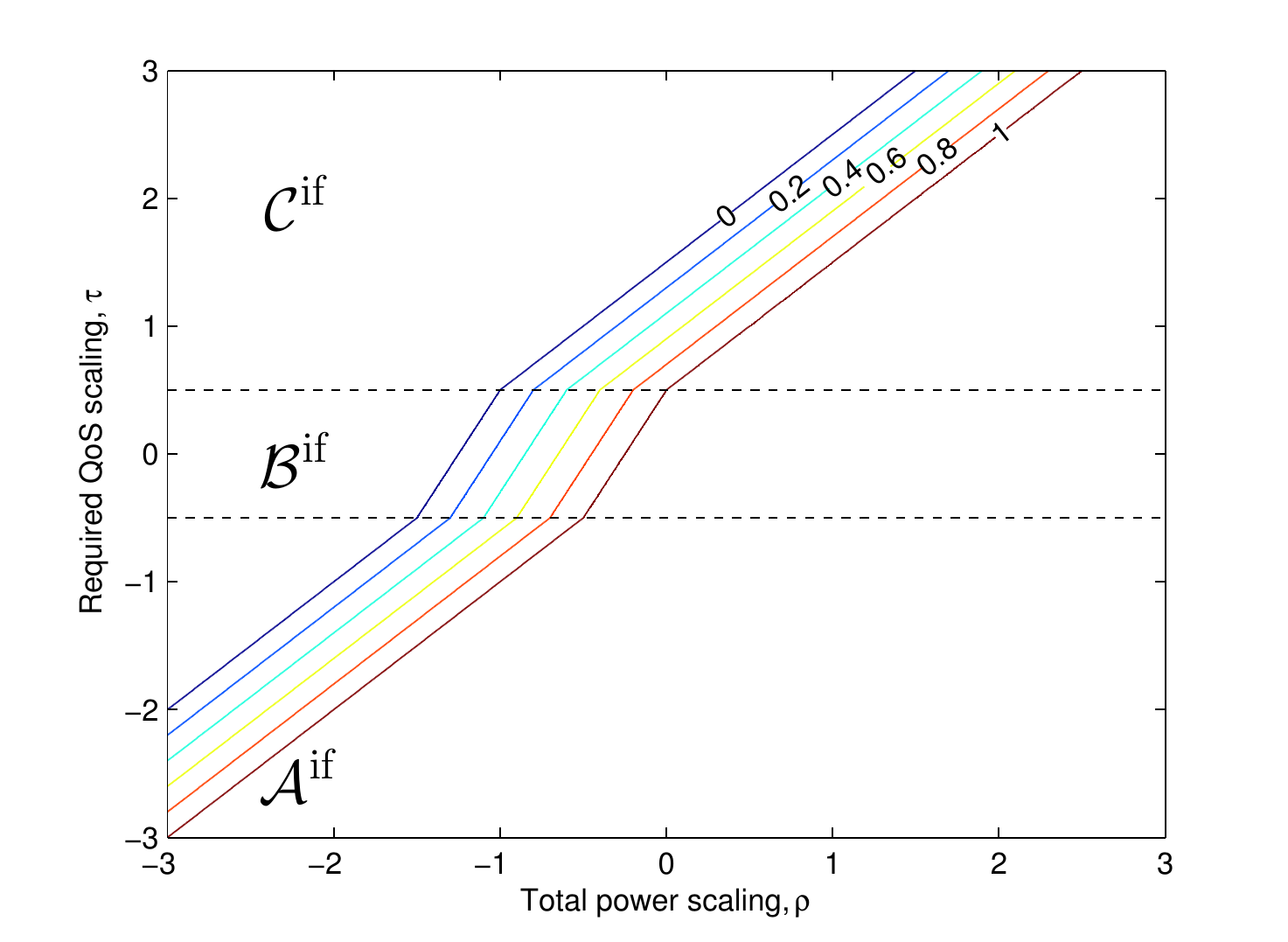}}
\subfigure[MRT operation]{\includegraphics[width=8cm]{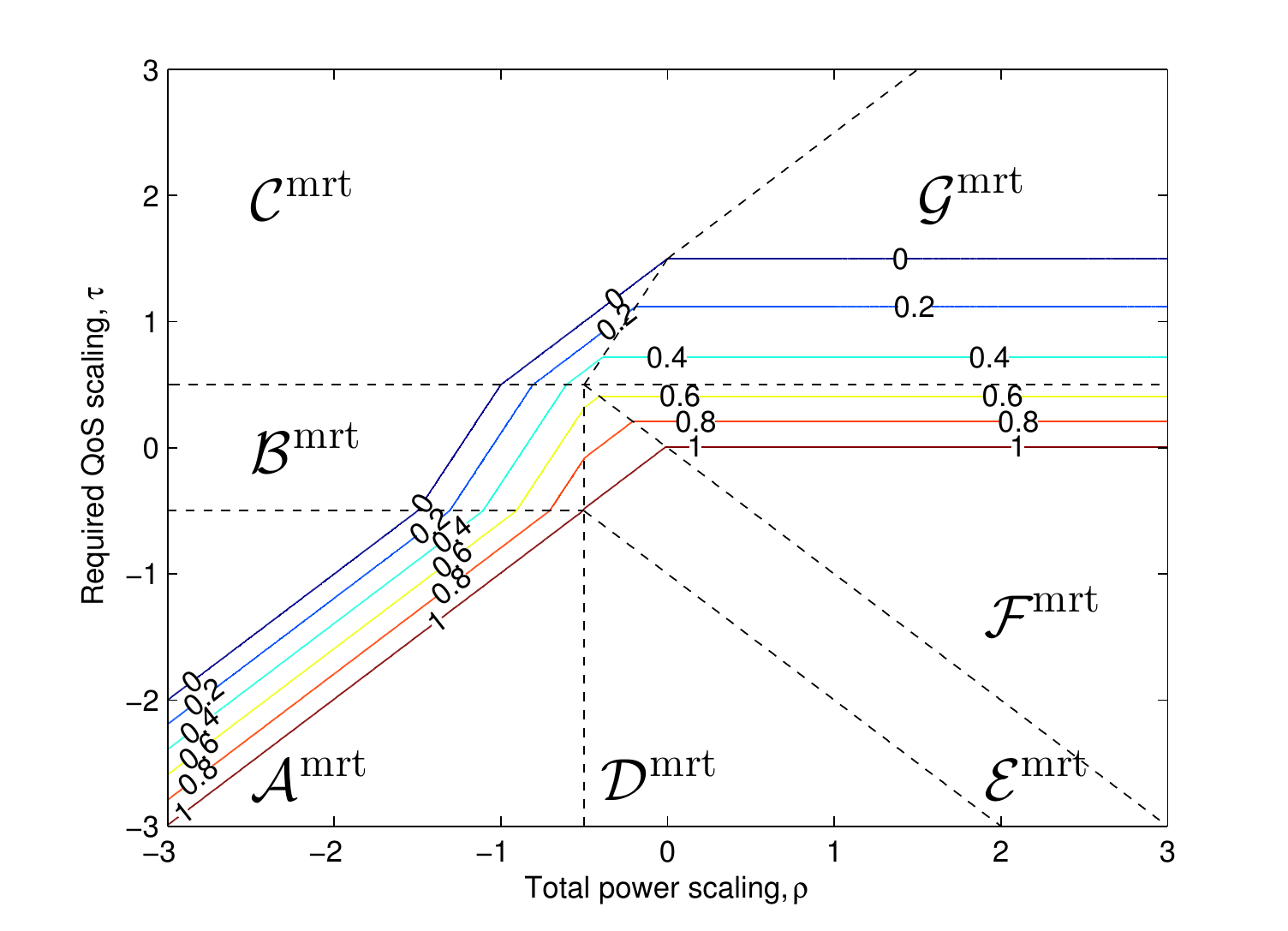}}
\subfigure[ZF operation]{\includegraphics[width=8cm]{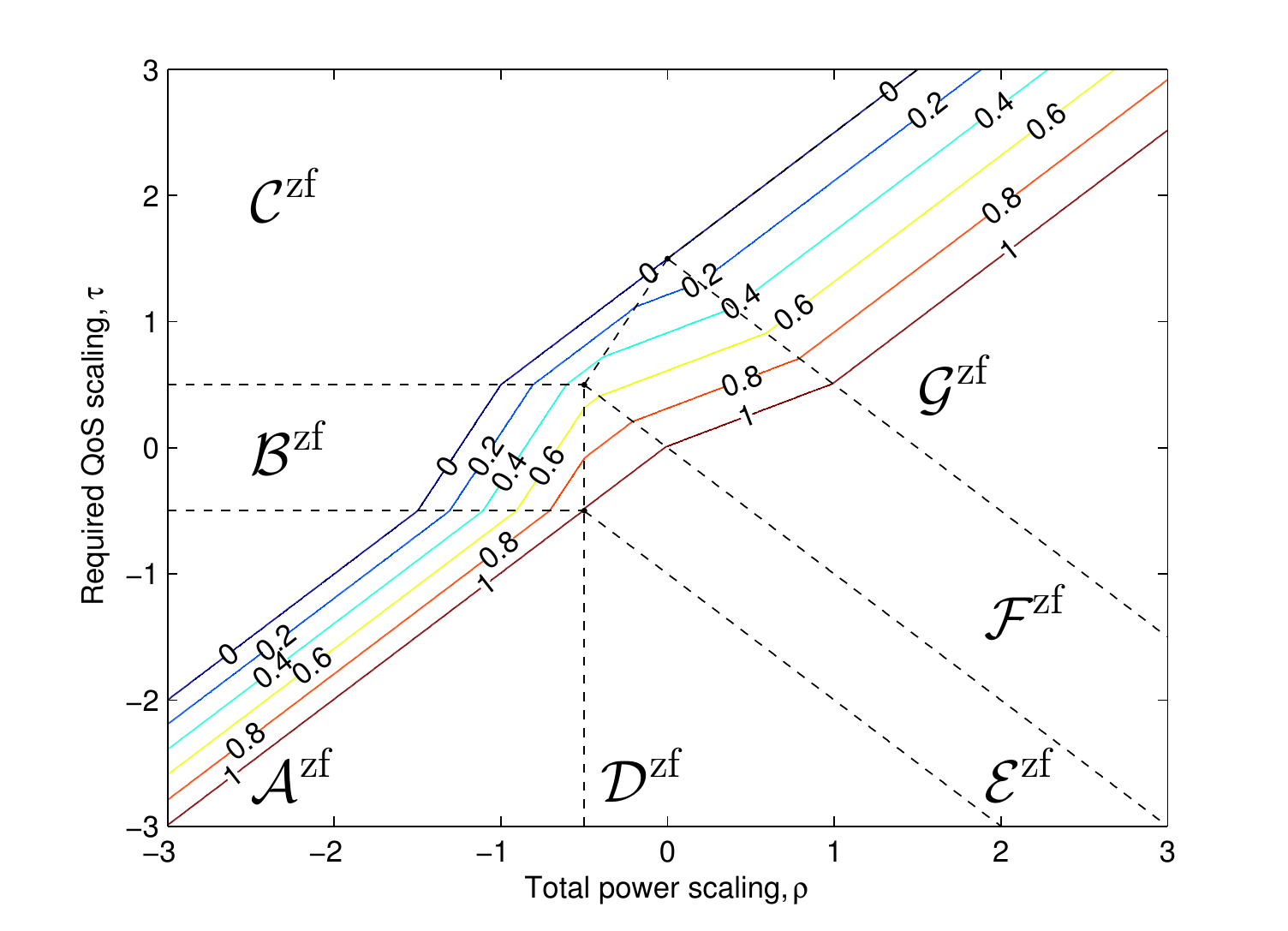} }
\caption{Regions and supportable user scaling exponents of IF, MRT and ZF operations when $\alpha =4$ and $\eta_{\rm{bs}}=\eta_{\rm{ant}}=0.5$.}
\end{figure}

The contour plots of $\zeta_{\rm{user}}^{\Upsilon}$ on $(\rho,\tau)$, $\Upsilon\in\{\mathsf{if, mrt, zf}\}$, are illustrated in Fig. 4, in which higher slope of the contour implies better power efficiency ($\tau$ over $\rho$).
Note that Theorem 6 can be straightforwardly extended to the case of a partially associated network ($\upsilon_{\rm{PA}}<\eta_{\rm{bs}}$), i.e., with the limited front/backhaul capacity constraint and/or with a pilot reuse ($\upsilon_{PR}<\eta_{\rm{user}}$), i.e., the limited pilot resource constraint. However, it is too complex to represent it as in Table II so that its numerical results will be shown in this paper.

\begin{remark}(Operating Regimes)
Regions $\mathcal{A}^{\Upsilon}$, $\mathcal{B}^{\Upsilon}$, and $\mathcal{C}^\Upsilon$ are SNR-limited operating regimes, (i.e., $\tau \le \mathsf{snr}^{\Upsilon}\le\mathsf{sir}^\Upsilon$) so that the same number of users can be supported regardless of operations. With respect to the UL transmit power, regions $\mathcal{A}^{\Upsilon}$, $\mathcal{B}^{\Upsilon}$, and $\mathcal{C}^{\Upsilon}$ are in \textsf{L}, \textsf{M}, and \textsf{H} or \textsf{EH}, respectively, and their slopes in the contour plot are $1$, $2$, and $1$, respectively. This is because $\mathsf{snr}^{\Upsilon}$ increases proportionally to $P^{\rm{dl}}_jP^{\rm{ul}}_j$ in \textsf{M}, while it increases proportionally to $P^{\rm{dl}}_j$ in the other cases.

Regions $\mathcal{D}^{\Upsilon}$, $\mathcal{E}^{\Upsilon}$, $\mathcal{F}^\Upsilon$, and $\mathcal{G}^\Upsilon$ are SIR-limited operating regimes, (i.e., $\tau \le \mathsf{sir}^{\Upsilon}\le\mathsf{snr}^\Upsilon$). Note that IF operation is always SNR-limited due to $\mathsf{sir^{if}} = \infty$ so that $\mathcal{D}^\mathsf{if}$, $\mathcal{E}^\mathsf{if}$, $\mathcal{F}^\mathsf{if}$, and $\mathcal{G}^\Upsilon$ are not defined. With respect to the UL transmit power, 
regions $\mathcal{D}^{\Upsilon}$, $\mathcal{E}^{\Upsilon}$, $\mathcal{F}^\Upsilon$, and $\mathcal{G}^\Upsilon$ are in $\textsf{L}$, $\textsf{M}$, $\textsf{H}$, and $\textsf{EH}$. Since a cooperative operation is meaningless in $\textsf{L}$ and $\textsf{M}$ (in which the CSI of a randomly selected user is inaccurate even in the nearest BS), $\zeta_{\rm{user}}^{\mathsf{mrt}}=\zeta_{\rm{user}}^{\mathsf{zf}}$ in $\mathcal{D}^{\mathsf{mrt}}=\mathcal{D}^{\mathsf{zf}}$ and $\mathcal{E}^{\mathsf{mrt}}=\mathcal{E}^{\mathsf{zf}}$. On the other hand, since ZF operation can cancel more interference in $\textsf{H}$ and $\textsf{EH}$ (in which the CSI of a randomly selected user becomes accurate in some BSs), $\zeta_{\rm{user}}^{\mathsf{zf}}\ge\zeta_{\rm{user}}^{\mathsf{mrt}}$ can be achieved. Interestingly, the slopes in the contour plot of MRT operation are $0$ both in $\mathcal{F}^{\mathsf{mrt}}$ and $\mathcal{G}^{\mathsf{mrt}}$, which implies that even if higher power is available in the network using MRT operation, the tradeoff between the QoS ($\tau$) and the number of users ($\zeta_{user}^{\mathsf{MRT}}$) is not improved. However, ZF operation shows positive slopes (specifically, $1-\frac{2}{\alpha}$ and $1$ in $\mathcal{F}^\mathsf{zf}$ and $\mathcal{G}^\mathsf{zf}$, respectively) so that higher power is always beneficial for improving the tradeoff. 
$\blacksquare$
\end{remark}

\begin{figure*}
\centering
\subfigure[Effect of partial association only, $\upsilon_{\rm{PA}}=0.2$.]{\includegraphics[width=8cm]{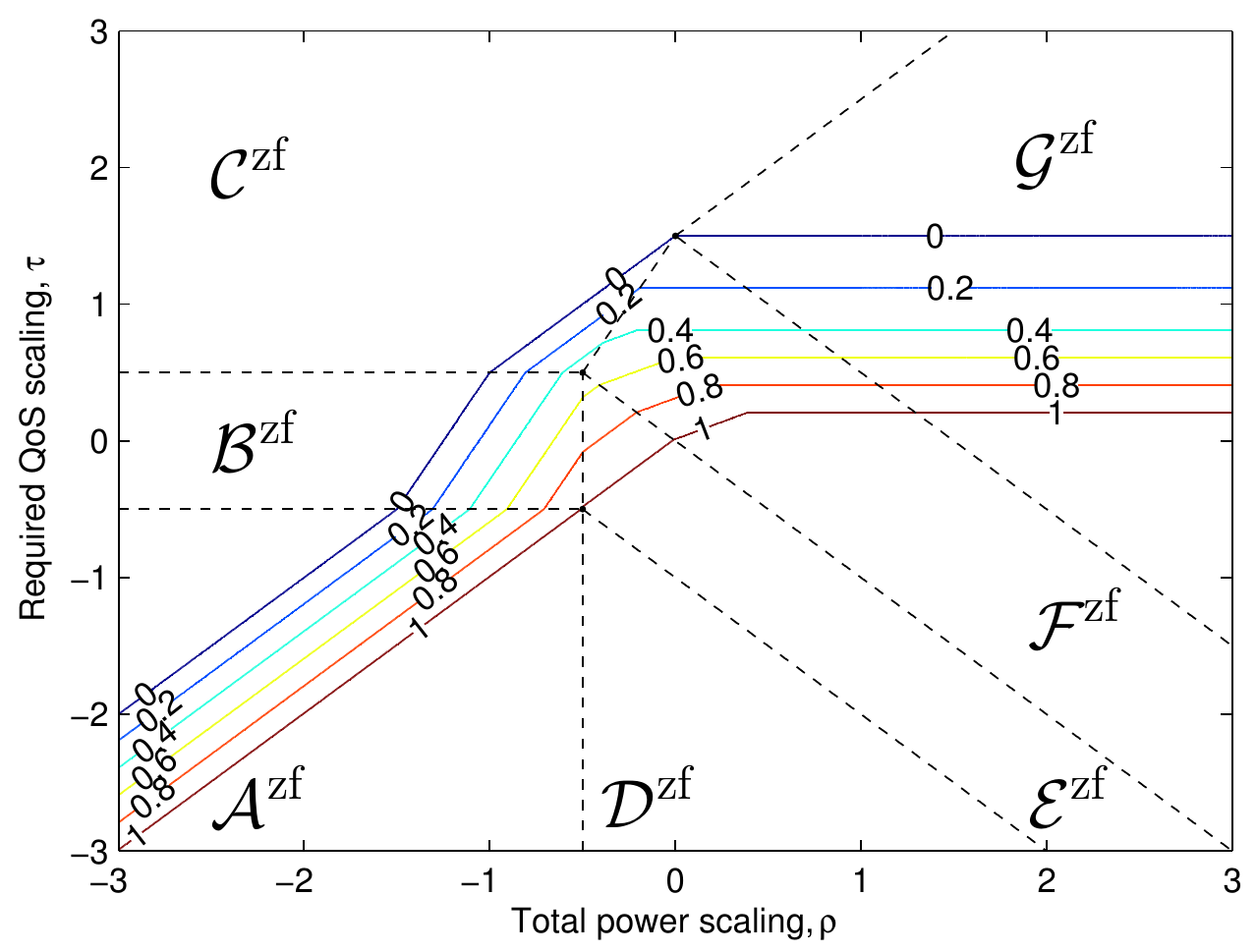}}\quad
\subfigure[Effect of pilot reuse only, $\upsilon_{\rm{PR}} =0.5$.]{\includegraphics[width=8cm]{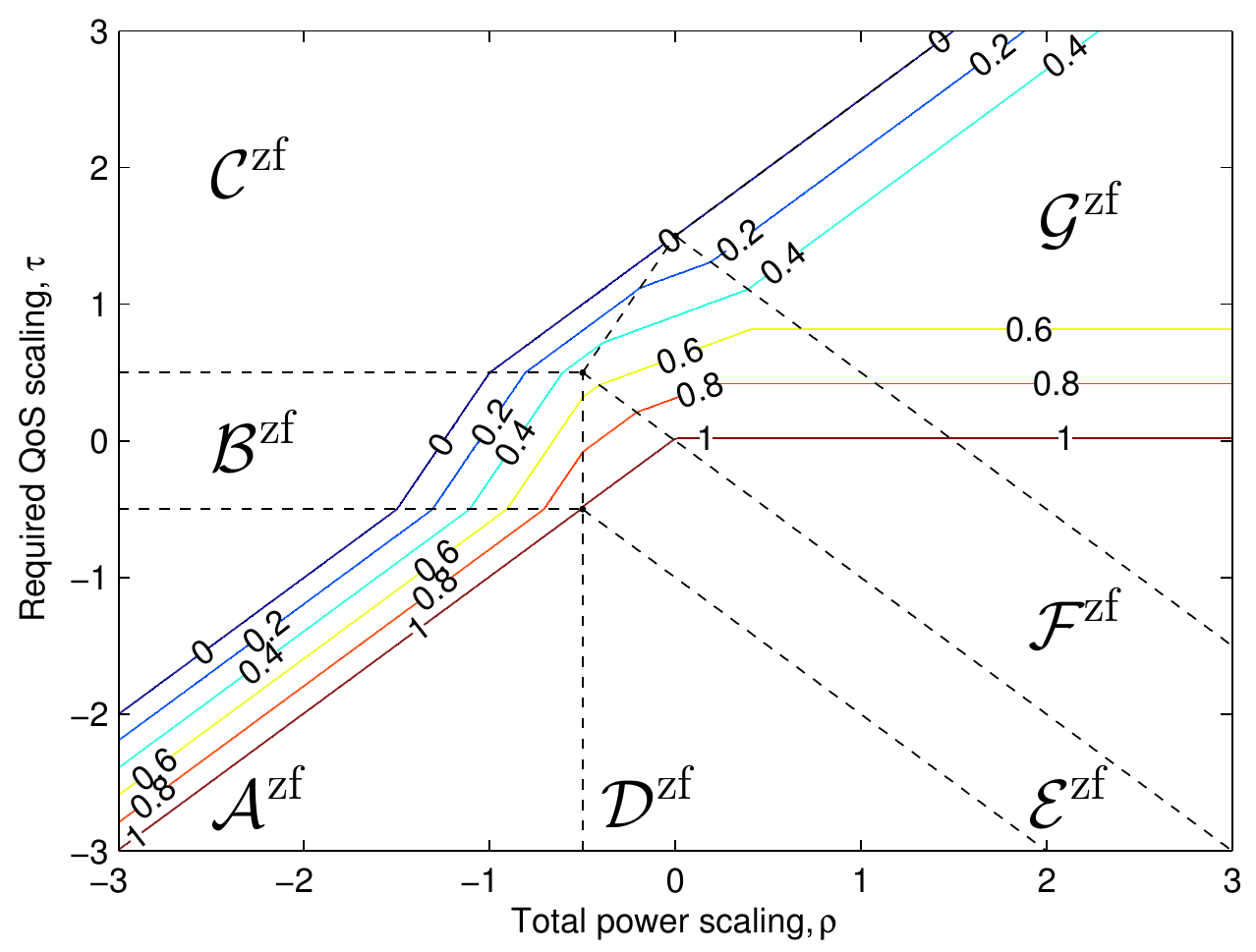}}\\
\subfigure[Effect of both partial association and pilot reuse, $\upsilon_{\rm{PA}}=0.2$ and $\upsilon_{\rm{PR}} =0.5$.]{\includegraphics[width=8cm]{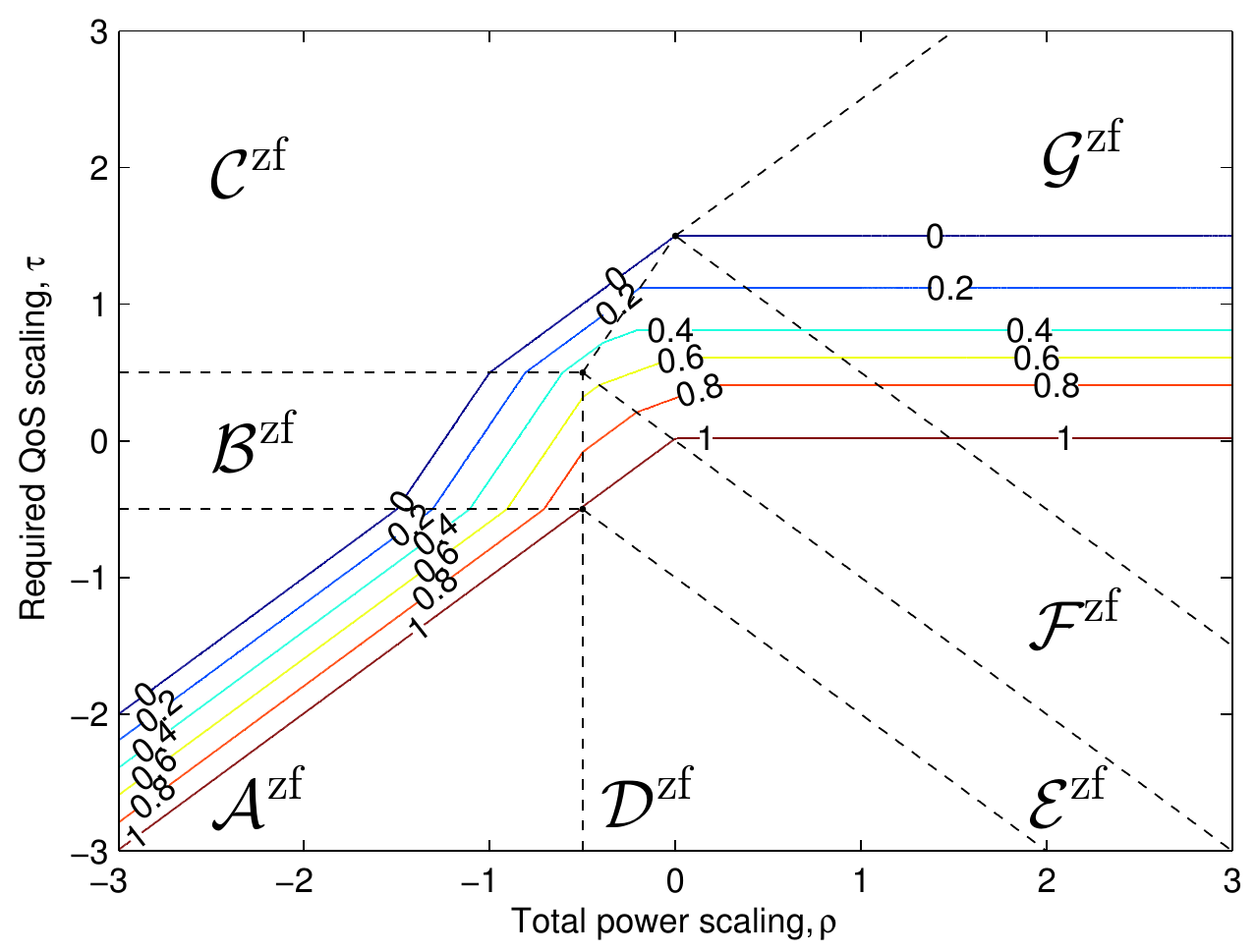}}\\
\caption{The scaling exponent of the number of supportable users of ZF operation with practical limitations, when $\eta_{\rm{bs}} = \eta_{\rm{ant}} = 0.5$.}
\end{figure*}

\begin{remark}(Effect of the partial association and/or the pilot reuse)
The effect of the partial association on the number of supportable users is shown in Fig. 5(a) when ZF operation is employed. Compared to Fig 4(c), it is shown that the contour plots in regions $\mathcal{A}^{\rm{zf}}$, $\mathcal{B}^{\rm{zf}}$, $\mathcal{C}^{\rm{zf}}$, $\mathcal{D}^{\rm{zf}}$ and $\mathcal{E}^{\rm{zf}}$ are unchanged, but those in regions $\mathcal{F}^{\rm{zf}}$ and $\mathcal{G}^{\rm{zf}}$ are considerably degraded. It is shown that a partial association poses an upper-limit on the target QoS scaling exponent $\tau$ even if additional transmit power is consumed, which is caused by the inevitable interference from non-associated BSs. So, the merits of ZF operation are reduced. In fact, if $\upsilon_{\rm{PA}} = 0$, ZF operation becomes MRT operation. However, if $\upsilon_{\rm{PA}}$ is strictly larger than 0, ZF operation can still provide better tradeoff between the target QoS scaling exponent ($\tau$) and the growth rate of the number of supportable users ($\zeta^{\Upsilon}_{\rm{user}}$) compared to MRT operation in regions $\mathcal{F}^{\rm{zf}}$ and $\mathcal{G}^{\rm{zf}}$. 

The effect of the pilot reuse on the number of supportable users is shown in Fig. 5(b). Since the scaling results does not change when $\eta_{\rm{user}} < \upsilon_{\rm{PR}}$, the contour lines marked as 0, 0.2, and 0.4 are not changed. However, the contour lines marked as 0.6, 0.8, and 1 are severely degraded due to the pilot reuse. Similarly as in the case of the partial association, an upper-limit is posed on the target QoS scaling $\tau$ even if additional transmit power is consumed, which comes from the fact that the CSI accuracy cannot be improved due to the interference caused by the users sharing the same pilots.

The effect of both the partial association and the pilot reuse is shown in Fig. 5(c). Evidently, both practical limitations results in lower target QoS scaling. The contour line marked as 1 is the same as in the case with the partial association only, i.e., the effect of partial association is dominant so that higher front/backhaul capacity can improve the network performance, while the other contour lines are the same as in the case with the pilot reuse only, i.e., the effect of pilot reuse is dominant so that more orthogonal pilot resource should be provided for better network performance.

\begin{figure*}
\centering
\subfigure[Required $\rho$ with the full association and no pilot reuse.]{\includegraphics[width=8cm]{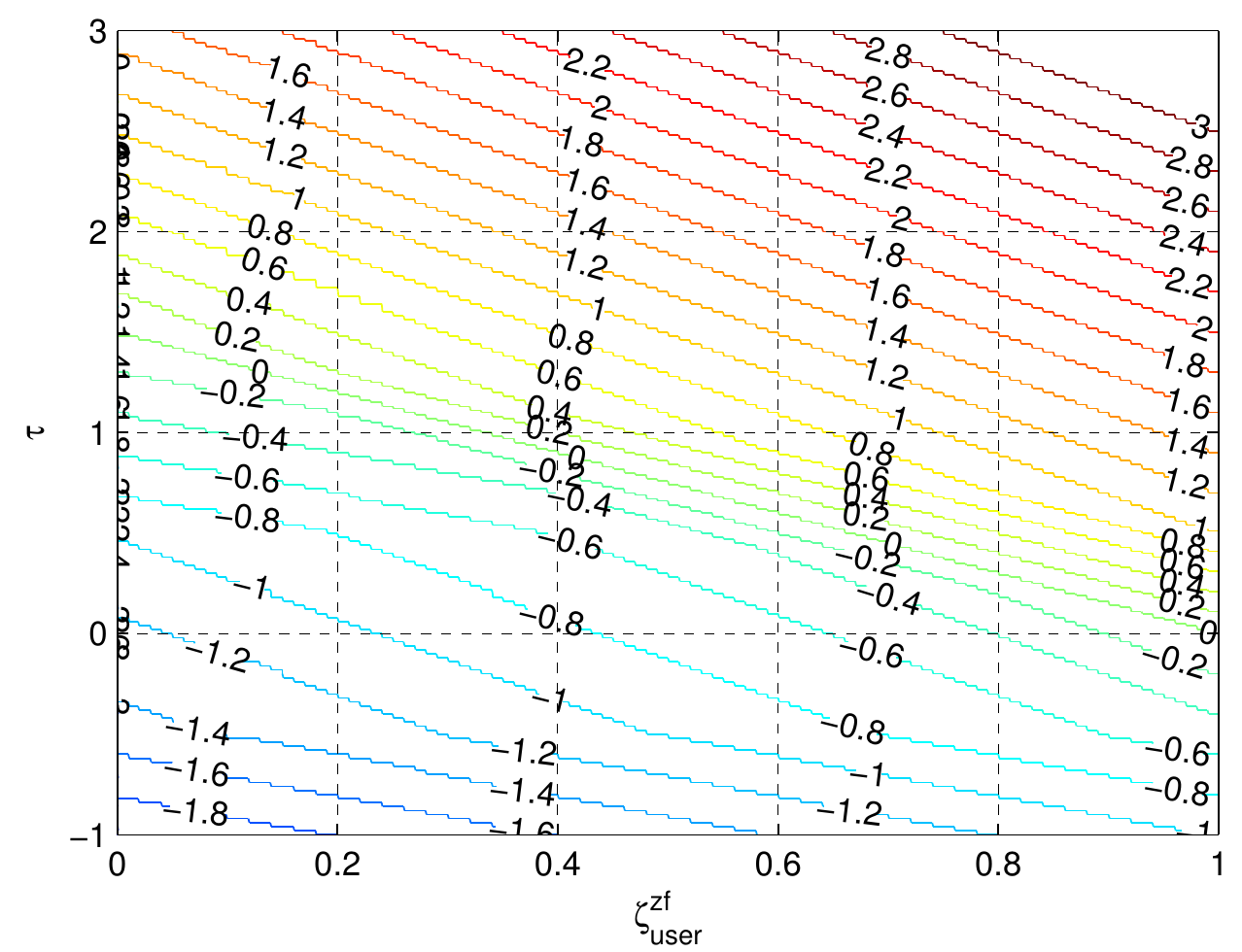}}\quad
\subfigure[Required $\rho$ with the full association and a pilot reuse with $\upsilon_{\rm{PR}}=0.5$.]{\includegraphics[width=8cm]{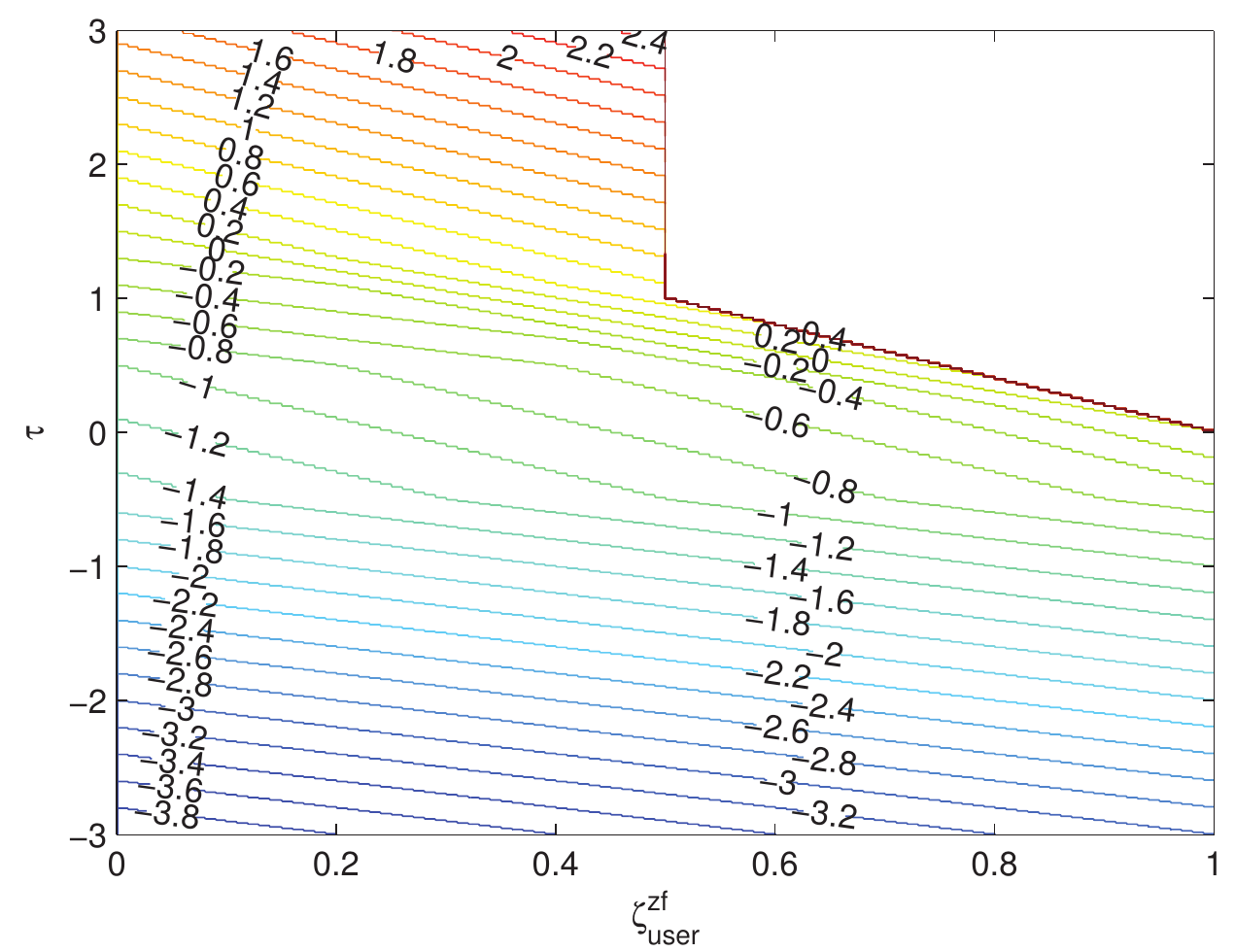}}\\
\subfigure[Required $\rho$ with a partial association with $\upsilon_{\rm{PA}}=0.2$ and no pilot reuse.]{\includegraphics[width=8cm]{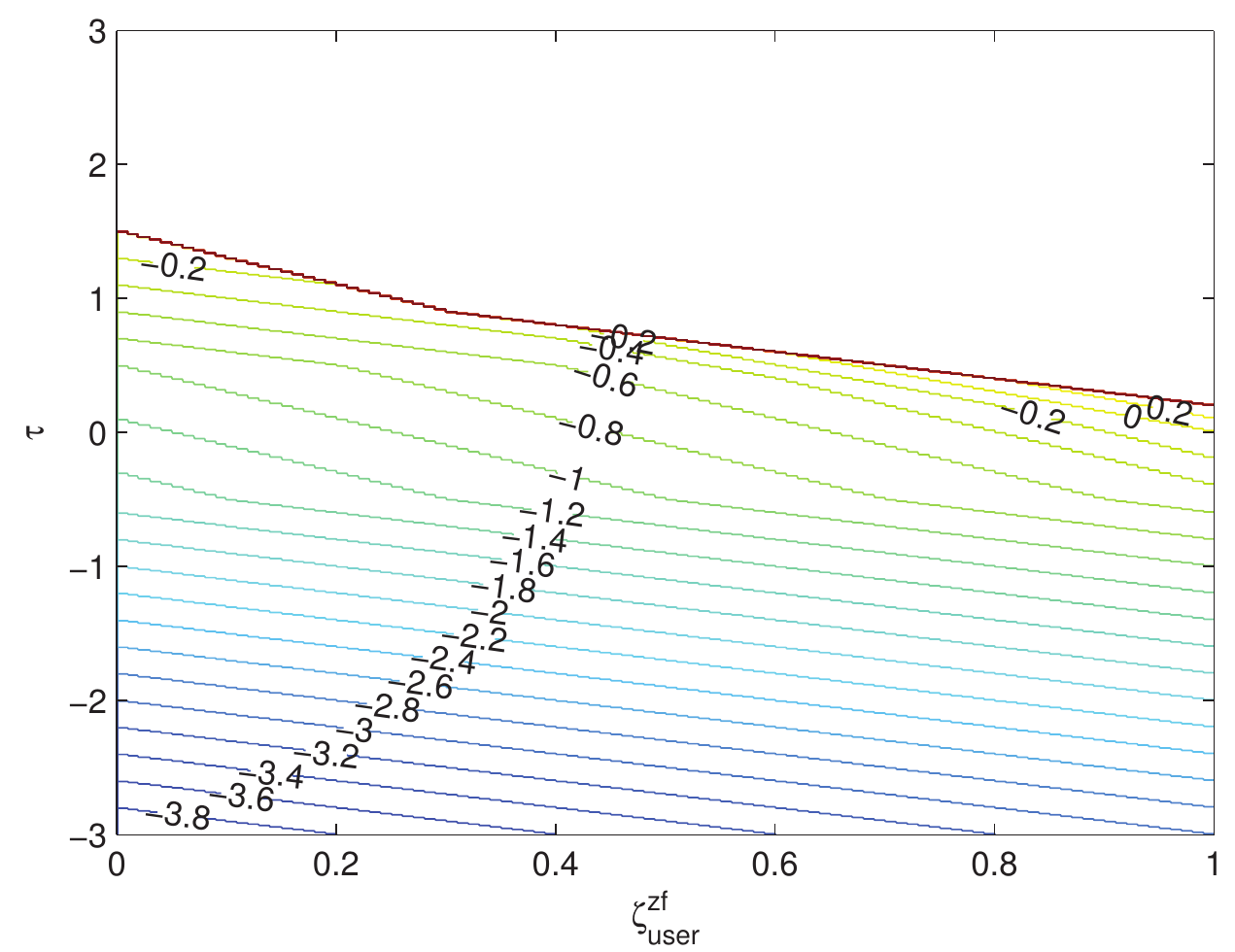}}\\
\caption{Contour plots for the required $\rho$ on ($\tau$,$\zeta_{\rm{user}}^{\mathsf{zf}}$) when $\eta_{\rm{bs}} = \eta_{\rm{ant}} = 0.5$.}
\end{figure*}

In Fig. 6, the contour plots for the required $\rho$ on ($\tau$,$\zeta_{\rm{user}}^{\mathsf{zf}}$) is illustrated, which shows the tradeoff between the target QoS scaling exponent ($\tau$) and the growth rate of the number of supportable users ($\zeta_{\rm{user}^{\mathcal{zf}}}$). By comparing Fig. 6(a) with Fig. 6(b), it is shown that the pilot reuse, $\upsilon_{\rm{PR}}=0.5$, does not change the tradeoff if $\zeta_{\rm{user}}^{\mathsf{zf}}$ is less than $\upsilon_{\rm{PR}}$. However, the pilot reuse degrades the tradeoff if $\zeta_{\rm{user}}^{\mathsf{zf}}$ is larger than $\upsilon_{\rm{PR}}$ so that a non-achievable region is created even if sufficiently high transmit power is available. Thus, in order to improve the tradeoff under the pilot resource constraint, an advanced technique is necessary to avoid or suppress the interference caused from the pilot reuse. Also, by comparing Fig. 6(a) with Fig. 6(c), it is shown that the partial association also degrades the performance and it creates a non-achievable region for high $\tau$ due to the interference mainly caused from non-associated BSs. Thus, a network designer needs to carefully design the association range and the front/backhaul capacity level. 
$\blacksquare$
\end{remark}

\section{Conclusion}
This paper presented a comprehensive and rigorous asymptotic analysis on the performance of the large-scale cloud radio access network (LS-CRAN) under the three practical constraints, 1) limited transmit power, 2) limited front/backhaul capacity, and 3) limited pilot resource. As a main performance measure, the scaling exponent of the signal-to-interference-plus-noise ratio (SINR) was defined and derived as a function of key network parameters, 1) the number of BS, $L = \Theta(N^{\eta_{\rm{bs}}})$, 2) the number of BS antennas $M= \Theta(N^{\eta_{\rm{ant}}})$, 3) the number of single-antenna users $ K= \Theta(N^{\eta_{\rm{user}}})$, 4) the uplink (UL) transmit power, $P_k^{\rm{ul}} = \Theta(N^{\rho^{\rm{ul}}})$, 5) the DL transmit power, $P^{\rm{dl}}= \Theta(N^{\rho^{\rm{dl}}})$, and the distance-dependent pathloss exponent $\alpha$, when interference-free (IF), maximum ratio transmission (MRT), or zero-forcing (ZF) operation is applied.

Then, we show that when MRT or ZF operation becomes interference-free, i.e., order-optimal with the three practical constraints. By considering limited transmit power only, MRT operation is shown to become interference-free only when the DL transmit power is less then a threshold, but the threshold is too low to be meaningful. Also, ZF operation is shown to become interference-free at meaningful DL transmit power so that as higher UL transmit power is provided, more users can be associated and the data rate per user can be increased simultaneously while keeping the order-optimality. Furthermore, when the other two practical constraints are considered together, it is shown that the total front/backhaul overhead of $\Omega(N^{\eta_{\rm{bs}}+\eta_{\rm{ant}}+\eta_{\rm{user}}+\frac{2}{\alpha}\rho^{\rm{ul}}})$ and $\Omega(N^{\eta_{\rm{user}}-\eta_{\rm{bs}}})$ pilot resources are required to keep the order-optimality. 

Last, in order to characterize the network-wise performance of the LS-CRAN, the growth rate of the number of supportable users, $\zeta_{\rm{user}}^{\Upsilon}$, is defined when a requirement on the target QoS scaling is given as $\tau$. Then, the tradeoff between ($\tau,\zeta_{\rm{user}}^{\Upsilon}$) is derived under the three practical constraints. The results quantify the achievable tradeoff between them at a given level of total transmit power, which shows that the target QoS and the number of users satisfying the QoS can grow simultaneously with a nice tradeoff as the network size increases. It is also shown and quantified that the other two practical constraints, limited front/backhaul capacity and limited pilot resources, may pose an upper-bound on the target QoS so that non-achievable tradeoff region is created, which would be a key ingredient providing a guideline for next-generation LS-CRANs.

\appendices

\section{Preliminary Results}

Here, we provide key preliminary results to prove theorems in this paper. First, we state the method to derive the scaling exponent of a random measure given by a sum of infinitely many i.i.d. random variables. 

\textbf{\emph{Lemma 1}}. Let $x_1,x_2,\cdots$ be a sequence of i.i.d. real-valued random variables with common distribution function of $F(x)$ and $h(x)$ be a non-negative integrable function. Then, as $n\to \infty$, we have
\begin{equation}\label{eq_18}
\begin{split}
\sum\limits_{k=1}^{n}h(x_k) &=\Theta \left( {\log n\int_0^1 {{n^t}} h\left(F^{-1}\left( {{n^{ - 1 + t}}} \right)\right)dt} \right)\\
&= \Theta \left( {n\int_{{F^{ - 1}}(\frac{1}{n})}^{{F^{ - 1}}(1)} h (u)f(u)du} \right).
\end{split}
\end{equation}
where $F^{-1}(x)$ is an inverse (or quantile) function of $F(x)$ and $f(x) = \frac{d}{dx}F(x)$ is the probability density function.
\begin{IEEEproof}
Suppose that $y_1,y_2,\cdots y_n$ are i.i.d. uniform random variables on $[0,1]$. 
Then, for any positive small $\epsilon$, the following holds:
\begin{equation}\label{eq_19}
\begin{split}
\Pr\left(  \min_{1 \le k\le n} y_{k} > \frac{1}{n ^{1+\epsilon}} \right) &={\left( {1 - \frac{1}{{{n^{1 + \epsilon}}}}} \right)^n}\\
&\to 1,~\text{as}~ n\to\infty,
\end{split}
\end{equation}
which implies that all of $y_1,\cdots ,y_n$ are included in the interval $(n^{-1-\epsilon}, 1]$, but not included in the interval $[0, n^{-1-\epsilon}]$ in probability. By dividing the interval $(n^{-1-\epsilon}, 1]$ into $\lceil g(n) \rceil$ intervals, ${{\cal S}_1} = ({n^{ - 1 - \epsilon }},{n^{{u_1}}}],{{\cal S}_2} = ({n^{{u_1}}},{n^{{u_2}}}], \cdots ,{{\cal S}_{ \lceil g(n) \rceil }} = ({n^{{u_{ \lceil g(n) \rceil  - 1}}}},1]$, where $u_i = -1+ \frac{i}{\left\lceil g(n)\right\rceil}$ for $i=1,2,\cdots,\lceil g(n) \rceil$, $\left\lceil \cdot \right\rceil$ denotes the ceiling operation, and $g(n)$ is a positive function of $n$. Let 
$\mathcal{Y}_i = \{y_j|y_j\in\mathcal{S}_i\}$ for $i=1,2,\cdots,\lceil g(n)\rceil$. Then, for $g(n) = o(n)$, 
\begin{equation}\label{eq_20}
\begin{split}
|{\mathcal{Y}}_i| &= \Theta\left(n^{1+u_{i}}\right)= \Theta\left(n^{\frac{i}{\lceil g(n) \rceil}}\right).
\end{split}
\end{equation} 
Since $y_1,\cdots,y_n$ are included in the interval $(n^{-1-\epsilon}, 1]$ in probability, we have 
\begin{equation}\label{eq_21}
\begin{split}
 \frac{1}{{\lceil g(n)\rceil}}\sum\limits_{i=1}^{{\lceil g(n)\rceil}}|{\mathcal{Y}}_i| &=  \Theta\left(\frac{1}{{\lceil g(n)\rceil}}\sum\limits_{i=1}^{{\lceil g(n)\rceil}}n^{\frac{i}{{\lceil g(n)\rceil}}}\right) \\
 &= \Theta\left(\int_0^1 {{n^t}dt}\right) \\
 & = \Theta \left( {\frac{n}{{\log n}}} \right),
\end{split}
\end{equation} 
where $g(n)$ is assumed to be properly chosen so that the second equality holds as a Riemann summation. Then, it is seen that $g(n)$ should increase logarithmically, i.e., $g(n) = \Theta(\log n)$.

Now, assume that $h(x)$ is a monotonically decreasing non-negative integrable function. First, we consider the upper-bound:
\begin{equation}\label{eq_22}
\begin{split}
\frac{1}{{\left\lceil {g(n)} \right\rceil }}\sum\limits_{k = 1}^n {h({y_k})} &= \frac{1}{{\left\lceil {g(n)} \right\rceil }}\sum\limits_{i = 1}^{\left\lceil {g(n)} \right\rceil } {\sum\limits_{{y_k} \in {{\cal Y}_i}} {h({y_k})} } \\
& \le \frac{1}{{\left\lceil {g(n)} \right\rceil }}\sum\limits_{i = 1}^{\left\lceil {g(n)} \right\rceil } {\left| {{{\cal Y}_i}} \right|h\left({n^{{u_{i - 1}}}}\right)} \\
 &= \Theta\left(\frac{{{n^{\frac{1}{{\left\lceil {g(n)} \right\rceil }}}}}}{{\left\lceil {g(n)} \right\rceil }}\sum\limits_{i = 1}^{\left\lceil {g(n)} \right\rceil } {{n^{\frac{{i - 1}}{{\left\lceil {g(n)} \right\rceil }}}}h\left({n^{ - 1 + \frac{{i - 1}}{{\left\lceil {g(n)} \right\rceil }}}}\right)}\right) \\
 &= \Theta \left( {{n^{\frac{1}{{\left\lceil {g(n)} \right\rceil }}}}\int_0^1 {{n^t}} h\left( {{n^{ - 1 + t}}} \right)dt} \right),
\end{split}
\end{equation}
where the last equality comes from the Riemann summation. Similarly, we obtain the lower-bound:
\begin{equation}\label{eq_23}
\begin{split}
\frac{1}{{\left\lceil {g(n)} \right\rceil }}\sum\limits_{k = 1}^n {h({y_k})} & \ge \Theta \left( {\int_0^1 {{n^t}} h\left( {{n^{ - 1 + t}}} \right)dt} \right)
\end{split}
\end{equation}
so that the gap between the lower-bound and the upper-bound is $\Theta(n^{\frac{1}{\lceil g(n)\rceil}}) = \Theta(1)$ due to $g(n) = \Theta(\log n)$. Since the gap is bounded by a constant, the upper-bound and the lower-bound is asymptotically tight. 
Additionally, $h(x)$ is set to a monotonically decreasing (or increasing) function so that such a gap is maximized because  the gap is still $\Theta(1)$ for any non-negative integrable function $h(x)$. Finally, by using the continuous mapping theorem \cite{BillingsleyConvergence} and changing the variable $u=F^{-1}(n^{-1+t})$, we obtain (\ref{eq_18}).
\end{IEEEproof}
Recall that our order notation ignores a logarithm term, but the proof of Lemma 1 does not ignore this term so that the Lemma 1 still holds when we use the order notation in \cite{KnuthBig}. 

This result is closely related to the (weak) law of large number. Obviously, ${\int_{{F^{ - 1}}({1}/{n})}^{{F^{ - 1}}(1)} h (u)f(u)du}$ converges to $\mathbb{E}[h(x_k)]$ as $n\to\infty$ if the function $h(\cdot)$ is independent to $n$. If $\mathbb{E}[h(x_k)]$ is finite and non-zero, $\sum_{k=1}^n h(x_k) = \Theta(n)$, which is consistent with the weak law of large number. In the case where $\mathbb{E}[h(x_k)]$ increases unboundedly or approaches zero as $n\to\infty$, Lemma 1 gives the way to quantify its asymptotic behavior via a simple integral form. Usually, it is too hard to derive the integral directly so a simpler way is necessary. By using Lemma 1 and our order notation, we can easily derive the scaling exponent by solving an optimization problem. 

\textbf{\emph{Corollary 1}}. Let $x_1,x_2,\cdots$ be a sequence of i.i.d. real-valued random variables with common distribution function of $F(x)$ and $h(x)$ be a non-negative integrable function with $h(F^{-1}(x)) = \Theta(x^{-p})$. Then, as $n\to\infty$, we have
\begin{equation}
\begin{split}
\sum\limits_{k=1}^n h(x_k) &= \Theta \left( {\mathop {\sup }\limits_{t \in [{t_{\min }},{t_{\max }}]} \left| {{{\cal X}_t}} \right|h\left( {{n^t}} \right)} \right),
\end{split}
\end{equation}
where $\mathcal{X}_t = \{x_k|x_k \in (c_1n^{t}, c_2n^{t}]\}$ for two constants $c_1<c_2$,  $$t_{\min} = \mathop{\arg\sup}_t \lim_{n\to\infty} \Pr(\mathcal{X}_{t}) = 0,$$
$$t_{\max} = \mathop{\arg\inf}_t \lim_{n\to\infty} \Pr(\mathcal{X}_{t}) = 1.$$
\begin{IEEEproof}
Using the similar argument in the proof of Lemma 1, we can easily show that if $h(F^{-1}(x)) = \Theta(x^{-p})$ for $p>1$,
\begin{equation}
\sum\limits_{k=1}^n h(x_k) = \Theta \left( {\mathop {\sup }\limits_{t \in [{t_{\min }},{t_{\max }}]} \sum\limits_{{x_k} \in {{\cal X}_t}} h ({x_k})} \right).
\end{equation} 
Since all elements in $\mathcal{X}_t$ scales like $\Theta(n^t)$, $x_k$ can be replaced by $n^t$, which concludes the proof.
\end{IEEEproof}

Note that, as $n\to\infty$,  all of $x_1,...,x_n$ are included in at least one set of $\mathcal{X}_{t}$ for $t_{\min}\le t\le t_{\max}$ in the probabilistic sense. For example, if $x_k$ are a sequence of i.i.d. uniform random variables on $[0,1]$, we can take $t_{\min}=  -1 $ and $t_{\max} = 0$.

This corollary informs that the scaling exponent of $\sum_{k=1}^n h(x_k)$ is only determined by the supremum of the partial sums, $\sum_{x_k\in\mathcal{X}_t}h(x_k)$, for all possible $t$ under some mild condition on $h(F^{-1}(x)) = \Theta(x^{-p})$ for $p>1$. Since we use the distance-based decay pathloss model, $d^{-\alpha}$, with $\alpha>2$, most cases hold this condition  
so that the scaling exponents obtained in this paper can be derived only by constructing the sets $\mathcal{X}_t$ and checking which set results in highest scaling exponent. The merits of Lemma 1 and Corollary 1 are shown in the next examples.

\emph{Example 1.} The aggregated interference at the origin of a circular-shaped network of radius $R$  can be written as  $$I = \sum\limits_{k=1}^{n} \max\left\{b,|x_k|\right\}^{-\alpha}, $$ 
where $b\ge 0$ is a bound and  $\{x_k\}_{k=1}^n$ follows a binomial point process on the circle of radius $R$, i.e., $F(x) = \frac{x^2}{R^2}$ for  $ 0\le x \le R$. 
If $b>0$, the pathloss model is bounded and $\mathbb{E}[\max\{b,|x_k|\}^{-\alpha}]$ converges to a finite constant so that $I = \Theta(n)$. However, if $b=0$, the pathloss model is unbounded at the origin and $\mathbb{E}[|x_k|^{-\alpha}] $ diverges for $\alpha\ge2$. Since ${F^{ - 1}}(x) = R\sqrt x$, the function $h(F^{-1}(x)) = \Theta(x^{-\frac{\alpha}{2}})$ if $b=0$ so Corollary 1 can be applied. 
After some manipulations, $|\mathcal{X}_t| = \Theta(n^{2t+1})$ is obtained for $-1/2\le t\le 0$. Then, applying Corollary 1, we have 
\[I = \Theta \left( {\mathop {\sup }\limits_{t \in [ - 1/2,0]} {n^{ - \alpha t}}} \right) = \Theta \left( {{n^{\frac{\alpha }{2}}}} \right).\]

\emph{Example 2.} In \cite{LevequeInformation}, the upper-bound of the capacity of a wireless ad-hoc network on $[0, \sqrt{n}]\times[0, \sqrt{n}]$ is shown as 
\begin{align*}
  {C}& = O\left(\frac{1}{n} {\sum\limits_{k = 1}^n {\log } \left( {1 + \sqrt n l({x_k})} \right)} \right) ,
\end{align*}
where ${l({x_k}) = {e^{ - \frac{\gamma }{2}|{x_k}|}}}$ for an exponential-decaying pathloss model with a positive $\gamma$ or ${l({x_k}) = \left|{x_k}\right|^{ - \frac{\alpha }{2}}}$ for a power-decaying pathloss model and $x_1,x_2...,x_n$ are i.i.d. uniform random variables on $[0,\sqrt{n}]$, i.e., $F(x) = x/\sqrt{n} $ for $x\in[0,\sqrt{n}]$. From Lemma 1, we obtain 
\begin{equation}\label{eq_48_a}
\begin{split}
C &= O\left( { \frac{\log n}{n} \int_0^1 {{n^t}\log \left( {1 + {n^{\frac{1}{2}}}l\left( {{n^{ - \frac{1}{2} + t}}} \right)} \right)} dt} \right)\\
 &= O\left( {\frac{\log n}{n}{n^{{t^\star}}}\log \left( {1 + {n^{\frac{1}{2}}}l\left( {{n^{ - \frac{1}{2} + {t^\star}}}} \right)} \right)} \right),
\end{split}
\end{equation}
where $t^\star = \arg\max_{0\le t\le 1} {n^{{t}}}\log ( {1 + {n^{\frac{1}{2}}}l( {{n^{ - \frac{1}{2} + {t}}}} )})$. Note that the last equality in (\ref{eq_48_a}) comes from that fact $\int_{0}^1 f(t)dt \le f(t^\star) $ where $t^\star = \max_{0\le t\le1} f(t)$. 
For the exponential-decaying pathloss model, 
$
{C} = O\left({{n^{-\frac{1}{2}}}{{\left( {\log n} \right)}^2}}\right)
$
by taking $t^\star= \frac{1}{2}$ as $n \to \infty$, which is consistent with Theorem 2.10 in \cite{LevequeInformation} and for the power-decaying pathloss model, 
$
{C}  = O \left({n^{\left( {\frac{1}{\alpha } - \frac{1}{2}} \right)}}\log n\right),
$
by taking $t^\star= {\frac{1}{\alpha } + \frac{1}{2}}$, as $n \to \infty$, which is consistent with Theorem 2.5 in \cite{LevequeInformation}.

\begin{figure}[t]
\centering
\includegraphics[width=8cm]{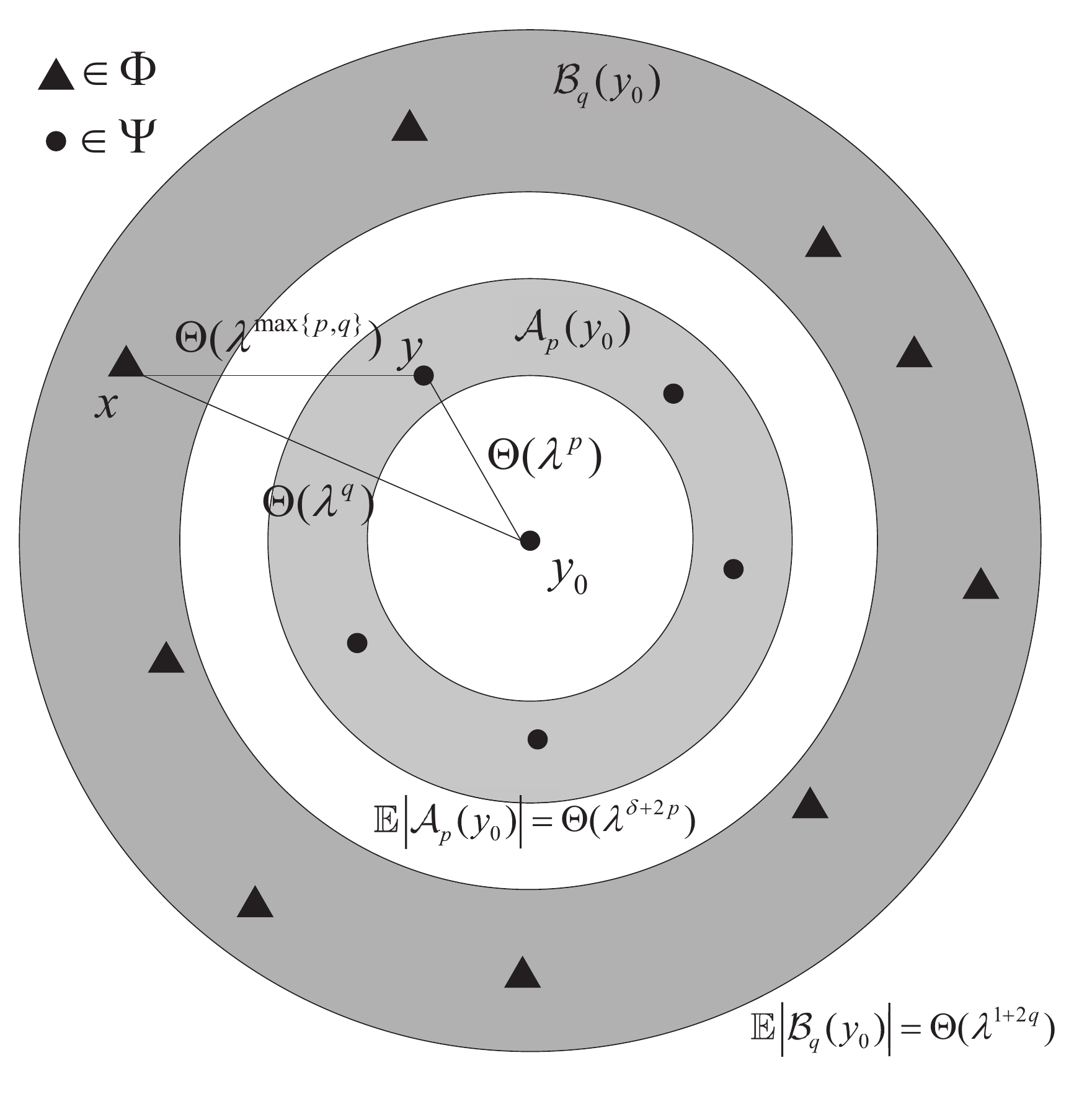}
\caption{Conceptual illustration for the proof of Lemma 2.}
\end{figure}

\textbf{\emph{Lemma 2}}. Let $\Phi$ and $\Psi$ be two independent homogeneous PPPs over a finite region $\mathcal{R}\subset \mathbb{R}^2$ with densities $\lambda$ and $\mu = \Theta(\lambda^\delta)$ for $\delta>0$, respectively. Assume that an arbitrarily chosen point $Y_0 \in\Psi$ is given. Then, the following holds for $\Phi_p(Z)\triangleq\{X\in\Phi \left| |X -Z| =  O(\lambda^p)\right. \}$,

1)  For $\alpha> 2$ and $-1/2<p\le0$, as $\lambda\to\infty$,
\begin{align}\label{lemma1_1}
\sum\limits_{X \in  {\Phi_p}({Y_0})} {{{\left| {X - {Y_0}} \right|}^{ - \alpha }}}  &= \Theta \left( {{\lambda ^{\frac{\alpha }{2}}}} \right).
\end{align}

2) For $\alpha_0,\alpha_1>2$, as $\lambda\rightarrow\infty$,

\begin{equation}\label{lemma1_2}
\begin{split}
\sum\limits_{Y \in \Psi \backslash \{ {Y_0}\} }& {\sum\limits_{ X \in \Phi \backslash \left( {{\Phi_p}({Y_0}) \bigcup {\Phi_p}(Y)} \right)} {{{\left| {X - {Y_0}} \right|}^{ - {\alpha _0}}}{{\left| {X - Y} \right|}^{ - {\alpha _1}}}} } = \Theta(\lambda^s),
\end{split}
\end{equation}
where $s$ is given by 
\[
s = \left\{ \begin{array}{*{20}{l}}
\frac{{{\alpha _{\max }}}}{2} + \frac{{{\alpha _{\min }}}}{2}\delta,&{\text{ if }}z <  - \frac{1}{2},\\
1 + \left( {2 - {\alpha _{\max }}} \right)z + \frac{{{\alpha _{\min }}}}{2}\delta,&{\text{ if }} - \frac{1}{2} \le z <  - \frac{1}{2}\delta ,\\
1 + \delta  + \left( {4 - {\alpha _{\max }} - {\alpha _{\min }}} \right)z,&{\text{ if }}z \ge  - \frac{1}{2}\delta ,
\end{array} \right.
\]
for $\delta  < 1$, or 
\[
s = \left\{ \begin{array}{*{20}{l}}
\frac{{{\alpha _{\max }} + {\alpha _{\min }}}}{2} + \delta  - 1,&{\text{ if }}z <  - \frac{1}{2},\\
1 + \delta  + \left( {4 - {\alpha _{\max }} - {\alpha _{\min }}} \right)z,&{\text{ if  }}z \ge  - \frac{1}{2},
\end{array} \right.
\]
for $\delta  \ge 1$, and
\begin{equation}\label{lemma1_3}
\begin{split}
\sum\limits_{Y \in \Psi \backslash \{ {Y_0}\} }& {\sum\limits_{\scriptstyle X \in  {{\Phi_p}({Y_0}) \bigcap {\Phi_p}(Y)} } {{{\left| {X - {Y_0}} \right|}^{ - {\alpha _0}}}{{\left| {X - Y} \right|}^{ - {\alpha _1}}}} } = \Theta(\lambda^t),
\end{split}
\end{equation}
where $t$ is given by 
\[t = \left\{ \begin{array}{*{20}{l}}
\frac{{{\alpha _{\max }}}}{2} + \frac{{{\alpha _{\min }}}}{2}\delta,&{\text{if }}z >  - \frac{1}{2}\delta  >  - \frac{1}{2},\\
\frac{{{\alpha _{\max }} + {\alpha _{\min }}}}{2} + \delta  - 1,&{\text{if }}z >  - \frac{1}{2} >  - \frac{1}{2}\delta ,\\
\frac{{{\alpha _{\max }} + {\alpha _{\min }}}}{2} + \delta  + 2z,&{\text{if }} - \frac{1}{2}\delta  < z <  - \frac{1}{2},\\
-\infty,&{\text{otherwise },}
\end{array} \right.\]
for $\alpha_{\max} = \max\{\alpha_0,\alpha_1\}$ and $\alpha_{\min} = \min\{\alpha_0,\alpha_1\}$.
\begin{IEEEproof}
1) To prove the first part of Lemma 2, we construct the lower-bound and upper-bound as follows.
$$\max\limits_{X \in  {{\cal A}_p}({Y_0})} {{{\left| {X - {Y_0}} \right|}^{ - \alpha }}}\le \sum\limits_{X \in  {{\cal A}_p}({Y_0})} {{{\left| {X - {Y_0}} \right|}^{ - \alpha }}}\le \sum\limits_{X \in  \Phi} {{{\left| {X - {Y_0}} \right|}^{ - \alpha }}}.$$
Here, the lower-bound is obviously $\Theta(\lambda^{\alpha/2})$ because the minimum distance is $\min_{X\in\mathcal{A}_p(Y_0)}|X-Y_0|=\Theta(\lambda^{-1/2})$. To complete the proof of the first part, it is sufficient to prove $\sum\limits_{X \in  \Phi} {{{\left| {X - {Y_0}} \right|}^{ - \alpha }}}=\Theta(\lambda^{\alpha/2})$. Let $h(x) = x^{-\alpha}$. For a randomly selected $X\in\Phi$, its contact distribution function is given by $$F(x;Y_0) = \Pr(|{X}-Y_0|<x) = \frac{v\left(b(Y_0,X) \bigcap \mathcal{R}\right)}{v(\mathcal{R})},$$
where $b(a,r) = \{b\in\mathbb{R}^2 | |a-b|\le r \}$ and $v(\cdot)$ denotes the Lebesgue measure (i.e. area measure). Since $v(\mathcal{R})$ is finite and independent to $x$, $F(x) = c'x^2+ o(x^2)$ and thus $h(F^{-1})(x) = cx^{-\frac{\alpha}{2}}+ o(x^{-\frac{\alpha}{2}})$, where $c$ is independent to $x$. From Lemma 1, we have 
\begin{equation}
\begin{split}
\sum\limits_{k = 1}^n {h({x_k})} & = \Theta \left( {\log n\int_0^1 {{n^{\left( {1 - \frac{\alpha }{2}} \right)t + \frac{\alpha }{2}}}dt} } \right) \\
&= \Theta \left( n \right) + \Theta \left( {{n^{\frac{\alpha }{2}}}} \right) \\
&= \Theta \left( {{n^{\frac{\alpha }{2}}}} \right),
\end{split}
\end{equation}
 where the last equality comes from $\alpha>2$ and the fact $n = \Theta(\lambda)$ completes the proof of the first part.

2)
Define two sets $\mathcal{A}_p(Z)\triangleq\{X\in\Phi | |X -Z| \mathop{=}\limits^{p} \Theta(\lambda^p) \}$ and $\mathcal{B}_q(Z)\triangleq\{y\in\Psi | |Y -Z| \mathop{=}\limits^{p} \Theta(\lambda^q) \}$ for a given point $Z$. Since $\Phi$ and $\Psi$ are PPPs with density $\lambda$ and $\mu$, respectively, $\mathbb{E}|\mathcal{A}_p(Z)| = \Theta(\lambda^{1+2p})$ and $\mathbb{E}|\mathcal{B}_q(Z)| = \Theta(\lambda^{\delta+2q})$. Note that $\Pr\left\{|\mathcal{A}_p(Z)|>0\right\}\rightarrow 1$ if and only if $-\frac{1}{2}<p\le 0$ and $\Pr\left\{|\mathcal{B}_q(Z)|>0\right\}\rightarrow 1$ if and only if $-\frac{\delta}{2}<q\le 0$. 
Assume that $Y_1 \in \Psi$ is given with $|Y_0-Y_1|=\Theta(\lambda^d)$ and $\alpha_0\ge \alpha_1$.  For any $X\in\mathcal{A}_p(Y_0)$, $|X-Y_1|= \Theta(\lambda^{\max\{p,d \}})$ so that ${{{\left| {X - {Y_0}} \right|}^{ - \alpha_0 }}{{\left| {X - {Y_1}} \right|}^{ - \alpha_1 }}} = \Theta \left( {{\lambda ^{- \alpha_0 p - \alpha_1 \max \left\{ {p,d} \right\}}}} \right)$. Then, we have
\begin{equation}\label{eq_51}
\begin{split}
&\sum\limits_{Y\in \mathcal{B}_d(Y_1)}\sum\limits_{ X\in\mathcal{A}_p(Y_0)\backslash{\left( {{\Phi _p}({Y_0})\bigcup {{\Phi _p}} (Y)} \right)}}{{{\left| {X - {Y_0}} \right|}^{ - \alpha_0 }}{{\left| {X - {Y_1}} \right|}^{ - \alpha_1 }}}
\\&=\Theta\left( {\lambda ^{\delta  + 2d + 1 + 2p - \alpha_0 p - {\alpha _1}\max \left\{ {d,p} \right\}}} \right),
\end{split}
\end{equation}
with the constraints of $\max \left\{ {z, - \frac{1}{2}} \right\} < p \le 0$ and $ - \frac{\delta}{2}  \le d \le 0$.

Similarly, we have 
\begin{equation}\label{eq_52}
\begin{split}
&\sum\limits_{Y\in \mathcal{B}_d(Y_1)}\sum\limits_{X\in\mathcal{A}_p(Y_0)\bigcap{\left( {{\Phi _p}({Y_0})\bigcap {{\Phi _p}} (Y)} \right)}}{{{\left| {X - {Y_0}} \right|}^{ - \alpha_0 }}{{\left| {X - {Y_1}} \right|}^{ - \alpha_1 }}}
\\&=\Theta\left( {\lambda ^{\delta  + 2d + 1 + 2p - \alpha_0 p - {\alpha _1}\max \left\{ {d,p} \right\}}} \right),
\end{split}
\end{equation}
with the constraints of $- \frac{1}{2} < p \le \min \left\{ {z,0} \right\}$ and $ - \frac{\delta}{2}  \le d \le 0$.
Then, the second part is easily given by obtaining the supremums of (\ref{eq_51}) and (\ref{eq_52}) under the above constraints.

\section{Proof of Theorem 1}
Since the interference is removed by the Genie without any cost, $\mathsf{sir}^{\mathsf{if}} = \infty$. Inserting $\mathbf{F}^\mathsf{if} = \widehat{\mathbf{G}}^H$ into (\ref{received_signal_from_BS_l}), ${\psi ^{\mathsf{mrt}}_{kk}} = \sum\nolimits_{X_l\in\mathcal{X}_k} \beta_{lk} \mathbf{h}^H_{lk}\widehat{\mathbf{h}}_{lk}$ and the square of which is given by 
\begin{equation}
\begin{split}
{\left| {{\psi^{\mathsf{if}} _{kk}}} \right|^2} \asymp& \sum\limits_{X_l \in\mathcal{X}_k} {{\beta _{lk}^2}} {\left| {{\bf{h}}_{lk}^H{{{\bf{\widehat h}}}_{lk}}} \right|^2} \\ 
 =&{\sum\limits_{X_l \in\mathcal{X}_k} {{{\left( {\frac{{{P_k^{\rm{ul}}}}}{{{P_k^{\rm{ul}}}{\beta _{lk}} + {1}}}} \right)}^2}\beta _{lk}^{4}{{\left| {{\bf{h}}_{lk}^H{{\bf{h}}_{lk}}} \right|}^2}}}
 +{ \sum\limits_{X_l \in\mathcal{X}_k} {{{\left( {\frac{{\sqrt {{P_k^{\rm{ul}}}} }}{{{P_k^{\rm{ul}}}{\beta _{lk}} + {1}}}} \right)}^2}\beta _{lk}^3{{\left| {{\bf{h}}_{lk}^H{{{\bf{\widetilde v}}}_{lk}}} \right|}^2}}},
\end{split}
\end{equation}
where the first asymptotic equality comes from the fact that the cross term is asymptotically ignorable and the last equality comes from inserting (\ref{estimated_CSI}).

Define $\mathcal{L}_{k}$ as the set of BSs sufficiently near user $k$, given by 
\begin{equation}\label{CSI_accuracy_set}
\mathcal{L}_{k} = \{X_l\in\mathcal{X}_k | P_k^{\rm{ul}}\beta_{lk} = \Omega(1)\}.
\end{equation}
Since 
$P_k^{\rm{ul}}\beta_{lk} = \Omega(1) \Leftrightarrow |X_l-U_k| = O\left(N^{{\rho^{{\rm{ul}}}/\alpha}}\right)$, $|\mathcal{X}_k| = O\left({N^{\eta_{\rm{bs}}+ \frac{2}{\alpha }{\rho ^{{\rm{ul}}}}}}\right)$, which means that $\mathcal{L}_{k}$ is a non-empty set if and only if $\rho^{\rm{ul}}\ge -\frac{\alpha}{2}\eta_{\rm{bs}}$ and $\mathcal{L}_{k} = \mathcal{X}_{k}$ if $\rho^{\rm{ul}}\ge 0$. 
Then, we have 
\begin{equation}
\begin{split}
{\left| {{\psi^{\mathsf{if}} _{kk}}} \right|^2} =& \sum\limits_{X_l \in {{\cal L}_{k}}} {\beta _{lk}^2{{\left| {{\bf{h}}_{lk}^H{{\bf{h}}_{lk}}} \right|}^2}} + \frac{1}{{P_k^{{\rm{ul}}}}}\sum\limits_{X_l \in {{\cal L}_{k}}} {\beta _{lk} {{\left| {{\bf{h}}_{lk}^H{{\widetilde {\bf{v}}}_{lk}}} \right|}^2}}  \\
&+{\left( {{{P_k^{{\rm{ul}}}}}} \right)^2}\sum\limits_{X_l \in {\cal X}_k\backslash {{\cal L}_{k}}} {\beta _{lk}^4{{\left| {{\bf{h}}_{lk}^H{{\bf{h}}_{lk}}} \right|}^2}} + {{P_k^{{\rm{ul}}}}}\sum\limits_{X_l \in {\cal X}_k\backslash {{\cal L}_{k}}} {\beta _{lk}^3{{\left| {{\bf{h}}_{lk}^H{{\widetilde {\bf{v}}}_{lk}}} \right|}^2}}.
\end{split}
\end{equation}

\emph{Case \textsf{EH}} ($\rho^{\rm{ul}}\ge 0$): in this case, $\mathcal{L}_{k} = \mathcal{X}_k$ and $\mathcal{X}_k\backslash\mathcal{L}_{k} = \emptyset$. Then, we have 
\begin{equation}\label{eq_38}
\begin{split}
\!\!\!{\left| {{\psi^{\mathsf{if}} _{kk}}} \right|^2}& = \sum\limits_{X_l \in {{\cal X}_k}} {\beta _{lk}^2{{\left| {{\bf{h}}_{lk}^H{{\bf{h}}_{lk}}} \right|}^2}} + \frac{1}{{P_k^{{\rm{ul}}}}}\sum\limits_{X_l \in {{\cal X}_k}} {\beta _{lk} {{\left| {{\bf{h}}_{lk}^H{{\widetilde {\bf{v}}}_{lk}}} \right|}^2}} \\
 &\asymp  {{N^{ \alpha{\eta_{\rm{bs}}} + 2{\eta_{\rm{ant}}}}}} + N^{-\rho^{\rm{ul}}+\frac{\alpha}{2}\eta_{\rm{bs}}+\eta_{\rm{ant}}}\\
&\asymp {{N^{ \alpha{\eta_{\rm{bs}}} + 2{\eta_{\rm{ant}}}}}} ,
\end{split}
\end{equation}
where the second asymptotic equality comes from Lemma 2 and the last asymptotic equality comes from the fact that the first term is always dominant for $\rho^{\rm{ul}}\ge 0$.

\emph{Cases \textsf{M} and \textsf{L}} ($\rho^{\rm{ul}} < -\frac{\alpha}{2}\eta_{\rm{bs}}$): in this case, $\mathcal{L}_{k} = \emptyset$ and $\mathcal{X}_k\backslash\mathcal{L}_{k} = \mathcal{X}_k$, we have 
\begin{equation}\label{eq_39}
\begin{split}
|\psi_{kk}^{\mathsf{if}}|^2 &= {\left( {{{P_k^{{\rm{ul}}}}}} \right)^2}\sum\limits_{X_l \in {\cal X}_k} {\beta _{lk}^4{{\left| {{\bf{h}}_{lk}^H{{\bf{h}}_{lk}}} \right|}^2}} + {{P_k^{{\rm{ul}}}}}\sum\limits_{X_l \in {\cal X}_k} {\beta _{lk}^3{{\left| {{\bf{h}}_{lk}^H{{\widetilde {\bf{v}}}_{lk}}} \right|}^2}}\\
&\asymp {{N^{ 2{\rho ^{\rm{ul}}} + 2\alpha {\eta_{\rm{bs}}} + 2{\eta_{\rm{ant}}} }}}+ {{N^{ {\rho ^{\rm{ul}}} + \frac{3\alpha}{2} {\eta_{\rm{bs}}} + {\eta_{\rm{ant}}}}}} \\
&\asymp \left\{ {\begin{array}{*{20}{l}}
 { {{N^{ 2{\rho ^{\rm{ul}}} + 2\alpha {\eta_{\rm{bs}}} + 2{\eta_{\rm{ant}}}}}},}~{{\text{if }} \text{ \textsf{M}},} \\ 
 { N^{ {\rho ^{\rm{ul}}} + \frac{{3\alpha }}{2}{\eta_{\rm{bs}}} + {\eta_{\rm{ant}}}},}~~~{{\text{if }} \text{ \textsf{L}},}
\end{array}} \right.
\end{split}
\end{equation}
where the second asymptotic equality comes from Lemma 2 and the third asymptotic equality comes from the fact that the first term is only dominant if $\rho^{\rm{ul}}>-\frac{\alpha}{2}\eta_{\rm{bs}}-\eta_{\rm{ant}}$.

\emph{Case \textsf{H}} ($-\frac{\alpha}{2}\eta_{\rm{bs}} \le \rho^{\rm{ul}} < 0$): 
The scaling is simply derived by combining the results in (\ref{eq_38}) and (\ref{eq_39}) so that it is given by 
\begin{equation}\label{eq_40}
\begin{split}
{\left| {{\psi^{\mathsf{if}} _{kk}}} \right|^2} &\asymp  {{N^{\alpha {\eta _{{\rm{bs}}}} + 2{\eta _{{\rm{ant}}}}}}} .
\end{split}
\end{equation}

Define an auxiliary variable $\nu$ such that $Q_k^{\mathsf{if}} = \Theta\left(N^{\nu}\right)$ for all $k$. 
The final step is to find the scaling exponent of the DL transmit power for user $k$, which is given by
\begin{equation}\label{eq_42}
\begin{split}
P_k^{{\rm{dl}}} &= Q^{\mathsf{if}}_k\sum\limits_{X_l \in {{\cal X}}_k} \beta_{lk} {\| {\widehat {\bf{h}}_{lk}} \|^2}\\
 &\asymp MQ^{\mathsf{if}}_k\sum\limits_{X_l \in {\mathcal{X}_k}} {\frac{{P_k^{\rm{ul}}\beta_{lk}^2}}{{P_k^{\rm{ul}}\beta_{lk} + {1}}}} \\
 &\asymp\left\{ {\begin{array}{*{20}{l}}
{ {{N^{\nu + \frac{\alpha }{2}{\eta_{\rm{bs}}} + {\eta_{\rm{ant}}}}}},}&\text{if \textsf{EH} or \textsf{H}},\\
{{{N^{\nu + {\rho ^{\rm{ul}}} + \alpha {\eta_{\rm{bs}}} + {\eta_{\rm{ant}}}}}},}&\text{if \textsf{M} or \textsf{L}.}
\end{array}} \right.
\end{split}
\end{equation}
By using  $P_k^{{\rm{dl}}} = \Theta( N^{\rho^{\rm{dl}}})$ for all $k$, and the results (\ref{eq_38}), (\ref{eq_39}), and (\ref{eq_40}), we obtain (\ref{if_snr_scaling}).

\section{Proof of Theorem 2}
Inserting $\mathbf{F}^\mathsf{mrt}= \widehat {\mathbf{G}}^H$ into (\ref{received_signal_from_BS_l}), ${\psi ^{\mathsf{mrt}}_{kj}} = \sum\limits_{X_l\in\mathcal{X}_j} \sqrt{\beta_{lk}\beta_{lj}} \mathbf{h}^H_{lk}\widehat{\mathbf{h}}_{lj}$, and the square of which is given by 
\begin{equation}\label{eq_45}
\begin{split}
{\left| {{\psi^{\mathsf{mrt}} _{kj}}} \right|^2}  \asymp& \sum\limits_{X_l \in\mathcal{X}_j} {{\beta _{lk}}{\beta _{lj}}} {\left| {{\bf{h}}_{lk}^H{{{\bf{\widehat h}}}_{lj}}} \right|^2} \\ 
= & \sum\limits_{X_l \in\mathcal{X}_j} {{{\left( {\frac{{{P_j^{\rm{ul}}}}}{{{P_j^{\rm{ul}}}{\beta _{lj}} + {1}}}} \right)}^2}\beta _{lk}\beta _{lj}^3{{\left| {{\bf{h}}_{lk}^H{{\bf{h}}_{lj}}} \right|}^2}} + \sum\limits_{X_l \in\mathcal{X}_j} {{{\left( {\frac{{\sqrt {{P_j^{\rm{ul}}}} }}{{{P_j^{\rm{ul}}}{\beta _{lj}} + {1}}}} \right)}^2}\beta _{lk} \beta _{lj}^2{{\left| {{\bf{h}}_{lk}^H{{{\bf{\widetilde v}}}_{lj}}} \right|}^2}},
\end{split}
\end{equation}
where the first asymptotic equality comes from the fact that the cross term is asymptotically ignorable and the third equality is obtained by inserting (\ref{estimated_CSI}). Note that $\psi^{\mathsf{mrt}}_{kk} = \psi^{\mathsf{if}}_{kk}$ so that the SNR scaling exponent of MRT operation is the same as that of IF operation. 

The interference power, denoted as $I_k$, can be written as 
\begin{equation}\label{eq_46}
\begin{split}
{I_k} = &{\sum _{U_j \in {\cal U}\backslash \{ U_k\} }}{Q^{\mathsf{mrt}}_j}|\psi _{kj}^\mathsf{mrt}{|^2}\\
 = &{\sum _{U_j \in {\cal U}\backslash \{ U_k\} }}{Q^{\mathsf{mrt}}_j}\sum\limits_{X_l \in \mathcal{L}_{j}} {\beta _{lk} \beta _{lj} {{\left| {{\bf{h}}_{lk}^H{{\bf{h}}_{lj}}} \right|}^2}}+ \frac{1}{{P_k^{{\rm{ul}}}}}{\sum _{U_j \in {\cal U}\backslash \{ U_k\} }}{Q^{\mathsf{mrt}}_j}\sum\limits_{X_l \in\mathcal{L}_{j}} {\beta _{lk} {{\left| {{\bf{h}}_{lk}^H{{\widetilde {\bf{v}}}_{lj}}} \right|}^2}} \\
& + {\left( {{{P_k^{{\rm{ul}}}}}} \right)^2}{\sum _{U_j \in {\cal U}\backslash \{ U_k\} }}{Q^{\mathsf{mrt}}_j}\sum\limits_{X_l \in {{\cal X}_j\backslash\mathcal{L}_j}} {\beta _{lk} \beta _{lj}^3{{\left| {{\bf{h}}_{lk}^H{{\bf{h}}_{lj}}} \right|}^2}}+ {{P_k^{{\rm{ul}}}}}{\sum _{U_j \in {\cal U}\backslash \{ U_k\} }}{Q^{\mathsf{mrt}}_j}\sum\limits_{X_l \in\mathcal{X}_{j}\backslash\mathcal{L}_{j}} {\beta _{lk} \beta _{lj}^2{{\left| {{\bf{h}}_{lk}^H{{\widetilde {\bf{v}}}_{lj}}} \right|}^2}},
\end{split}
\end{equation}
where $\mathcal{L}_{j}$ is defined in (\ref{CSI_accuracy_set}).

Define an auxiliary variable $\nu$ such that $Q_k^{\mathsf{mrt}} = \Theta\left(N^{\nu}\right)$ for all $k$
 and consider the four regimes \textsf{EH}, \textsf{H}, \textsf{M}, and \textsf{L}.

\emph{Case \textsf{EH}} ($\rho^{\rm{ul}}\ge 0$): in this case, $\mathcal{L}_{j} = \mathcal{X}_j$ and $\mathcal{X}_j\backslash\mathcal{L}_j = \emptyset$ for all $j$. Then, we have 
\begin{equation}\label{eq_47}
\begin{split}
I_k =& \sum\limits_{U_j \in\mathcal{U}\backslash\{U_k\}}Q^{\mathsf{mrt}}_j{\sum\limits_{X_l \in\mathcal{X}_j} {\beta _{lk}\beta _{lj}{{\left| {{\mathbf{h}}_{lk}^H{{\mathbf{h}}_{lj}}} \right|}^2}} } + \frac{1}{{{P_k^{\rm{ul}}}}}\sum\limits_{U_j\in\mathcal{U}\backslash\{U_k\}} Q^{\mathsf{mrt}}_j{\sum\limits_{X_l \in\mathcal{X}_j} {\beta _{lk} {{\left| {{\mathbf{h}}_{lk}^H{{{\mathbf{\widetilde v}}}_{lj}}} \right|}^2}} } \\
\asymp& N^{\nu+\frac{\alpha}{2}({\eta_{\rm{bs}}}+\eta_{\rm{user}})+(1-\frac{\alpha}{2})({\eta_{\rm{user}}}-{\eta_{\rm{bs}}})^++\eta_{\rm{ant}}}  +N^{\nu-\rho^{\rm{ul}}+\frac{\alpha}{2}{\eta_{\rm{bs}}}+\eta_{\rm{user}}+\eta_{\rm{ant}}}\\
 \asymp& N^{\nu+\frac{\alpha}{2}({\eta_{\rm{bs}}}+\eta_{\rm{user}})+(1-\frac{\alpha}{2})({\eta_{\rm{user}}}-{\eta_{\rm{bs}}})^++\eta_{\rm{ant}}} , 
\end{split}
\end{equation}
where the second asymptotic equality comes from Lemma 2 and the third asymptotic equality comes from the fact that the first term is always dominant due to 
$\left( {\frac{\alpha }{2} - 1} \right)\left( {{{({\eta_{\rm{user}}} - {\eta_{\rm{bs}}})}^ + } - {\eta_{\rm{user}}}} \right)\le0$.

\emph{Cases \textsf{M} and \textsf{L}} ($\rho^{\rm{ul}} < -\frac{\alpha}{2}\eta_{\rm{bs}}$): in this case, $\mathcal{L}_{j} = \emptyset$ and $\mathcal{X}_j\backslash\mathcal{L}_{j} = \mathcal{X}_j$. Then, we have 
\begin{equation}\label{eq_48}
\begin{split}
I_k = &{\left( {{{{P_k^{\rm{ul}}}}}} \right)^2}\sum\limits_{U_j \in\mathcal{U}\backslash\{U_k\}} Q^{\mathsf{mrt}}_j{\sum\limits_{X_l \in\mathcal{X}_j} {\beta _{lk} \beta _{lj}^3{{\left| {{\mathbf{h}}_{lk}^H{{\mathbf{h}}_{lj}}} \right|}^2}} } +{{{P_k^{\rm{ul}}}}}\sum\limits_{U_j\in\mathcal{U}\backslash\{U_k\}} Q^{\mathsf{mrt}}_j{\sum\limits_{X_l \in\mathcal{X}_j} {\beta _{lk} \beta _{lj}^2{{\left| {{\mathbf{h}}_{lk}^H{{{\mathbf{\widetilde v}}}_{lj}}} \right|}^2}} } \\
\asymp&  {{N^{\nu + 2{\rho ^{\rm{ul}}} +  {\frac{{3\alpha }}{2}{\eta_{\rm{bs}}} + \frac{\alpha }{2}{\eta_{\rm{user}}} + \left( {1 - \frac{\alpha }{2}} \right){{\left( {{\eta_{\rm{user}}} - {\eta_{\rm{bs}}}} \right)}^ + }} + {\eta_{\rm{ant}}}}}}   +{{N^{\nu + {\rho ^{\rm{ul}}} +  {\alpha {\eta_{\rm{bs}}} + \frac{\alpha }{2}{\eta_{\rm{user}}} + \left( {1 - \frac{\alpha }{2}} \right){{\left( {{\eta_{\rm{user}}} - {\eta_{\rm{bs}}}} \right)}^ + }}  + {\eta_{\rm{ant}}}}}} 
\\
\asymp&  {{N^{\nu + {\rho ^{\rm{ul}}} +  {\alpha {\eta_{\rm{bs}}} + \frac{\alpha }{2}{\eta_{\rm{user}}} + \left( {1 - \frac{\alpha }{2}} \right){{\left( {{\eta_{\rm{user}}} - {\eta_{\rm{bs}}}} \right)}^ + }} + {\eta_{\rm{ant}}}}}},
\end{split}
\end{equation}
where the second equality comes from Lemma 2 and the third asymptotic equality comes from the fact that the first term is always dominant due to 
$\rho^{\rm{ul}}<-\frac{\alpha}{2}\eta_{\rm{bs}}$. 

\emph{Case \textsf{H}} ($-\frac{\alpha}{2}\eta_{\rm{bs}} \le \rho^{\rm{ul}} < 0$): the scaling exponent can be simply derived as 
\begin{equation}\label{eq_49}
I_k \asymp N^{\nu+\frac{\alpha}{2}({\eta_{\rm{bs}}}+\eta_{\rm{user}})+(1-\frac{\alpha}{2})({\eta_{\rm{user}}}-{\eta_{\rm{bs}}})^++\eta_{\rm{ant}}},
\end{equation}
because the result in (\ref{eq_48}) at $\rho^{\rm{ul}} = -\frac{\alpha}{2}\eta_{\rm{bs}}$ is the same to the result in (\ref{eq_47}). Thus, by combining (\ref{eq_47})-(\ref{eq_49}), the scaling exponent of the interference power is given by 
\begin{equation}\label{eq_50}
\!\!\!\!\!\!\!\!\!\!\!{I_k} \asymp\left\{ \begin{array}{l}
{N^{\nu + \frac{\alpha }{2}({\eta_{\rm{bs}}} + {\eta_{\rm{user}}}) + (1 - \frac{\alpha }{2}){{({\eta_{\rm{user}}} - {\eta_{\rm{bs}}})}^ + } + {\eta_{\rm{ant}}}}},~~~~~{\text{ if \textsf{EH} or \textsf{H}}},\\
{N^{\nu + {\rho ^{\rm{ul}}} + \alpha {\eta_{\rm{bs}}} + \frac{\alpha }{2}{\eta_{\rm{user}}} + \left( {1 - \frac{\alpha }{2}} \right){{\left( {{\eta_{\rm{user}}} - {\eta_{\rm{bs}}}} \right)}^ + } + {\eta_{\rm{ant}}}}},{\text{ if \textsf{M} or \textsf{L}}},
\end{array} \right.
\end{equation}
and combining it with (\ref{if_snr_scaling}) and (\ref{eq_42}) results in (\ref{mrt_sir_scaling}), which completes the proof.

\section{Proof of Theorem 3}
Inserting $ {\mathbf{F}}^{\mathsf{zf}} =  {\widehat {\bf{G}}^H}{\left( {\widehat {\bf{G}}{{\widehat {\bf{G}}}^H}} \right)^{ - 1}}$  into (\ref{received_signal_from_BS_l}), $\psi _{kj}^\mathsf{zf} = \delta (k - j) + {[\widetilde {\bf{G}}]_{k,:}}{{\bf{F}}^{zf}}{{\bf{e}}_j}$, where $\mathbf{e}_j$ denotes the $j$th column of the identity matrix $\mathbf{I}_K$, and the square of which is given by 
\begin{equation}\label{ZF_psi}
\begin{split}
{\left| {\psi _{kj}^\mathsf{zf}} \right|^2}& = \delta (k - j) + 2{\rm{Re}}\left\{ {{{[\widetilde {\bf{G}}]}_{k,:}}{{\bf{F}}^\mathsf{zf}}{{\bf{e}}_j}} \right\} + {\rm{Tr}}\left( {{{\left( {{{\bf{F}}^\mathsf{zf}}} \right)}^H}{{\bf{R}}_k}{{\bf{F}}^{zf}}{{\bf{e}}_j}{\bf{e}}_j^H} \right) \\
& \asymp \delta (k - j) + {\rm{Tr}}\left( {{{\left( {{{\bf{F}}^\mathsf{zf}}} \right)}^H}{{\bf{R}}_k}{{\bf{F}}^\mathsf{zf}}{{\bf{e}}_j}{\bf{e}}_j^H} \right),
\end{split}
\end{equation}
where $\mathbf{R}_k = [ {\widetilde {\bf{G}}} ]_{k,: }^H{{[ {\widetilde {\bf{G}}} ]}_{k,: }}$ and  the last asymptotic equality comes from the fact that $(a+b)^2\asymp~a^2+b^2$. Also, the DL transmit power consumed for user $k$ is given by 
\begin{equation}
\begin{split}
P_k^{{\rm{dl}}} = Q_k^\mathsf{zf}{\rm{Tr}}\left( {{{\left( {{{\bf{F}}^{{\mathsf{zf}}}}} \right)}^H}{{\bf{F}}^{{\mathsf{zf}}}}{{\bf{e}}_k}{\bf{e}}_k^H} \right).
\end{split}
\end{equation}

Note that we have
\begin{equation}
\begin{split}
{\left[ {\widehat {\bf{G}}{{\widehat {\bf{G}}}^H}} \right]_{i,j}} &= \sum\limits_{{X_l} \in {\cal X}} {\sqrt {{\beta _{li}}{\beta _{lj}}} {\bf{\widehat {h}}}_{li}^H{{{\bf{\widehat {h}}}}_{lj}}}\\
& = \left\{ {\begin{array}{*{20}{l}}
{M\sum\limits_{{X_l} \in {\cal X}} {\frac{{P_j^{{\rm{ul}}}\beta _{lj}^2}}{{P_j^{{\rm{ul}}}{\beta _{lj}} + 1}}} ,}&{{\text{if }}i = j,}\\
{\sqrt M \sum\limits_{{X_l} \in {\cal X}} {\sqrt {\frac{{P_i^{{\rm{ul}}}\beta _{li}^2}}{{P_i^{{\rm{ul}}}{\beta _{li}} + 1}}\frac{{P_j^{{\rm{ul}}}\beta _{lj}^2}}{{P_j^{{\rm{ul}}}{\beta _{lj}} + 1}}} } ,}&{{\text{if }}i \ne j,}
\end{array}} \right.
\end{split}
\end{equation}
and consider the four regimes.

\emph{Case \textsf{EH}} ($\rho^{\rm{ul}}\ge 0$):  in this case, $\mathcal{L}_k =\mathcal{X}_k$ so that $\mathbf{R}_k \asymp \frac{1}{P_k^{\rm{ul}}} \mathbf{I}_{N}$
and 
\begin{equation}\label{ZF_EH2}
\begin{split}
{\left[ {\widehat {\bf{G}}{{\widehat {\bf{G}}}^H}} \right]_{i,j}} \asymp \left\{ {\begin{array}{*{20}{c}}
{M\sum\limits_{{X_l} \in {\cal X}} {{\beta _{lj}}} ,}&{{\text{if }}i = j,}\\
{\sqrt M \sum\limits_{{X_l} \in {\cal X}} {\sqrt {{\beta _{li}}{\beta _{lj}}} } ,}&{{\text{if }}i \ne j,}
\end{array}} \right.
\end{split}
\end{equation}
Inserting (\ref{ZF_EH2}) into (\ref{ZF_psi}), we have 
\begin{equation}
\begin{split}
{\left| {{\psi^\mathsf{zf} _{kk}}} \right|^2} \asymp&  1 +\frac{{{1}}}{{MP_k^{\rm{ul}}}}\left(\sum\limits_{X_l\in\mathcal{X}_k} \beta_{lk}\right)^{-1}\\
\asymp& 1+N^{-\rho^{\rm{ul}}-\frac{\alpha}{2}\eta_{\rm{bs}}-\eta_{\rm{ant}}}\\
\asymp& 1
\end{split}
\end{equation}
and 
\begin{equation}
\begin{split}
\sum\limits_{U_j\in\mathcal{U}\backslash\{U_k\}}{\left| {{\psi ^\mathsf{zf}_{kj}}} \right|^2} &\asymp \frac{{{1}}}{{MP_k^{\rm{ul}}}}\sum\limits_{U_j\in\mathcal{U}\backslash\{U_k\}}  \left(\sum\limits_{X_l\in\mathcal{X}_k}\beta_{lk}\right) \left( {{{\sum\limits_{X_l\in\mathcal{X}_j} {{\beta _{lj}}} }}} \right)^{-2} \\
&\asymp   {{N^{ - {\rho ^{\rm{ul}}} - \frac{\alpha }{2}{\eta _{\rm{bs}}} - {\eta _{\rm{ant}}}+{\eta _{\rm{user}}}}}},
\end{split}
\end{equation}
where the last asymptotic equality comes from Lemma 2 and $\rho^{\rm{ul}}\ge 0$. Defining an auxiliary variable $\nu$ such that $Q_k=\Theta(N^{\nu})$, the DL transmit power can be written as
\begin{equation}
\begin{split}
P_k^{{\rm{dl}}} =& \frac{{{Q^\mathsf{zf}_k}}}{M}\left({{\sum\limits_{X_l \in \mathcal{X}_k} {{\beta _{lk}}} }}\right)^{-1}\\
 \asymp&  {{N^{\nu  - \frac{\alpha }{2}{\eta _{{\rm{bs}}}} - {\eta _{{\rm{ant}}}}}}}.
\end{split}
\end{equation}
For a reasonable DL transmit power $P_k^{{\rm{dl}}}=\Theta(N^{\rho^{\rm{ul}}})$, $\nu = \rho^{\rm{dl}}+\frac{\alpha}{2}\eta_{\rm{bs}}+\eta_{\rm{ant}}$ and thus
\begin{equation}
\begin{split}
\mathsf{snr^{zf}} =& \rho^{\rm{dl}}+\frac{\alpha}{2}\eta_{\rm{bs}}+\eta_{\rm{ant}},\\
\mathsf{sir^{zf}} =& \rho^{\rm{ul}}+\frac{\alpha}{2}\eta_{\rm{bs}}+\eta_{\rm{ant}}-\eta_{\rm{user}}.
\end{split}
\end{equation}

\emph{Case \textsf{H}} ($-\frac{\alpha}{2}\eta_{\rm{bs}} \le \rho^{\rm{ul}}<0$):
define $\mathcal{U}^{(d)}_k = \left\{U_j \in \mathcal{U} |  \left|U_k-U_j\right| = \Theta\left(N^d\right)\right\}$. Since
$\mathbf{R}_k = {\rm{diag}}\left\{ { {{\theta _{1k}}} {{\bf{I}}_M},..., {{\theta _{Lk}}}{{\bf{I}}_M}} \right\}$,
 where  $\theta_{lk} \asymp 1/P_k^{\rm{ul}}$, if $X_l\in\mathcal{X}_k$, or $\theta_{lk} \asymp\beta_{lk}$, if $X_l\in\mathcal{X}\backslash\mathcal{X}_k$, 
and 
\begin{equation}
\begin{split}
&{\left[ {{\bf{\widehat G}}  {{{\bf{\widehat G}}}^H}} \right]_{i,j}}\asymp  \left\{ {\begin{array}{*{20}{l}}
{M\left( {\sum\limits_{{X_l} \in {{\cal L}_j}} {{\beta _{lj}}}  + {{P_j^{{\rm{ul}}}}}\sum\limits_{{X_l} \in {{\cal X}_j}\backslash {{\cal L}_j}} {\beta _{lj}^2} } \right),}&{{\text{if }}i = j,}\\
{\sqrt M \left( \begin{array}{l}
\sum\limits_{{X_l} \in {{\cal L}_i}\bigcap {{{\cal L}_j}} } {\sqrt {{\beta _{li}}{\beta _{lj}}} } 
 + \sqrt {{{P_j^{{\rm{ul}}}}}} \sum\limits_{{X_l} \in {{\cal L}_i}\backslash {{\cal L}_j}} {\sqrt {{\beta _{li}}\beta _{lj}^2} } \\
 + \sqrt {{{P_i^{{\rm{ul}}}}}} \sum\limits_{{X_l} \in {{\cal L}_j}\backslash {{\cal L}_i}} {\sqrt {\beta _{li}^2{\beta _{lj}}} } 
 + {{\sqrt {P_i^{{\rm{ul}}}P_j^{{\rm{ul}}}}}}\sum\limits_{{X_l} \in {{\cal X}_i} \cap {{\cal X}_j}\backslash {{\cal L}_i} \cup {{\cal L}_j}} {\sqrt {\beta _{li}^2\beta _{lj}^2} } 
\end{array} \right),}&{{\text{if }}i \ne j.}
\end{array}} \right.
\end{split}
\end{equation}
Then, ${\left| {{\psi^\mathsf{zf} _{kj}}} \right|^2}$ is given as in (\ref{eq_90}), where the last asymptotic equality comes from the fact that
\begin{figure*}[!t]
\small
\begin{equation}\label{eq_90}
\begin{split}
\sum\limits_{U_j\in\mathcal{U}\backslash\{U_k\}}{\left| {{\psi^\mathsf{zf} _{kj}}} \right|^2} \asymp&  \sum\limits_{U_j\in\mathcal{U}\backslash\{U_k\}}\frac{{\frac{{{1}}}{{P ^{{\rm{ul}}}}}\sum\limits_{X_l \in {{\cal X}_k} \cap {{\cal X}_j}}  {{\beta _{lj}}}  + \sum\limits_{X_l \in {{\cal X}_k}\backslash {{\cal X}_j}}  {\beta _{lj}^2}  + \sum\limits_{X_l \in {{\cal X}_j}\backslash {{\cal X}_k}}  {{\beta _{lk}}{\beta _{lj}}}  + {{P ^{{\rm{ul}}}}}\sum\limits_{X_l \in {\cal X}\backslash \left( {{{\cal X}_k} \cup {{\cal X}_j}} \right)}  {{\beta _{lk}}\beta _{lj}^2} }}{{M{{\left( {\sum\limits_{l \in {{\cal X}_j}}  {{\beta _{lj}}}  + {{P ^{{\rm{ul}}}}}\sum\limits_{X_l \in {\cal X}\backslash {{\cal X}_j}}  {\beta _{lj}^2} } \right)}^2}}}\\
\asymp& \frac{1}{{{N^{\alpha {\eta _{{\rm{bs}}}} + {\eta _{{\rm{ant}}}}}}}}\left( {\underbrace {\frac{{{1}}}{{{P^{{\rm{ul}}}}}}{\sum _{{U_j} \in {\cal U}\backslash \{ {U_k}\} }}\sum\limits_{{X_l} \in {{\cal X}_k} \cap {{\cal X}_j}} {{\beta _{lj}}} }_{{t_1}} + \underbrace {{\sum _{{U_j} \in {\cal U}\backslash \{ {U_k}\} }}{\sum _{{X_l} \in {{\cal X}_k}\backslash {{\cal X}_j}}}\beta _{lj}^2}_{{t_2}}}\right.\\
&~~~~~~~~~~~~~~~~~~~~~~~~~~~~~\left.{+ \underbrace {{\sum _{{U_j} \in {\cal U}\backslash \{ {U_k}\} }}{\sum _{{X_l} \in {{\cal X}_j}\backslash {{\cal X}_k}}}{\beta _{lk}}{\beta _{lj}}}_{{t_3}} + \underbrace {{{{P^{{\rm{ul}}}}}}{\sum _{{U_j} \in {\cal U}\backslash \{ {U_k}\} }}{\sum _{{X_l} \in {\cal X}\backslash \left( {{{\cal X}_k} \cup {{\cal X}_j}} \right)}}{\beta _{lk}}\beta _{lj}^2}_{{t_4}}} \right),
\end{split}
\end{equation}
\hrulefill
\vspace*{4pt}
\end{figure*}
\begin{equation}
\begin{split}
\sum\limits_{X_l \in {{\cal L}_k}} {{\beta _{lk}}}  + {{P_k^{{\rm{ul}}}}}\sum\limits_{X_l \in {\cal X}_k\backslash {{\cal L}_k}} {\beta _{lk}^2}\asymp & {N^{\frac{\alpha }{2}{\eta _{{\rm{bs}}}}}} + {N^{{\eta _{{\rm{bs}}}} + \left( {\frac{2}{\alpha } - 1} \right){\rho ^{{\rm{ul}}}}}}
\\ \asymp&  {N^{\frac{\alpha }{2}{\eta _{{\rm{bs}}}}}}.
 \end{split}
\end{equation}
Here, $t_1$, $t_2$, $t_3$, and $t_4$ can be further derived from Lemma 2 as follows.
Suppose that $\eta_{\rm{bs}}\ge \eta_{\rm{user}}$ and $- \frac{\alpha }{2}{\eta _{\rm{bs}}} \le {\rho ^{{\rm{ul}}}} <  - \frac{\alpha }{2}{\eta _{{\rm{user}}}}$. In this case, ${{\cal L}_k}\bigcap {{{\cal L}_j}}  = \emptyset $ for all $U_j\in\mathcal{U}$ in probability so that, by using Lemma 2, we obtain
\[\begin{array}{l}
{t_1} = 0,\\
{t_2}\asymp{N^{{\eta _{bs}} + \alpha {\eta _{{\rm{user}}}} + \frac{2}{\alpha }{\rho ^{{\rm{ul}}}}}},\\
{t_3}\asymp{N^{\frac{\alpha }{2}{\eta _{{\rm{bs}}}} + \frac{\alpha }{2}{\eta _{{\rm{user}}}}}},\\
{t_4}\asymp{N^{{\eta _{\rm{bs}}} + \frac{\alpha }{2}{\eta _{\rm{user}}} + \left( {\frac{2}{\alpha } - 1} \right){\rho ^{\rm{ul}}}}}.
\end{array}\]
Since $\frac{\alpha }{2}{\eta _{{\rm{bs}}}} + \frac{\alpha }{2}{\eta _{{\rm{user}}}} - \left( {{\eta _{bs}}+\alpha {\eta _{{\rm{user}}}} + \frac{2}{\alpha }{\rho ^{{\rm{ul}}}}} \right) = \left( {\frac{\alpha }{2} - 1} \right)\left( {{\eta _{{\rm{bs}}}} - {\eta _{{\rm{user}}}}} \right) - \frac{2}{\alpha }\left( {{\rho ^{{\rm{ul}}}} + \frac{\alpha }{2}{\eta _{{\rm{user}}}}} \right) \ge 0$ and $
\frac{\alpha }{2}{\eta _{{\rm{bs}}}} + \frac{\alpha }{2}{\eta _{{\rm{user}}}} - \left( {{\eta _{bs}} + \frac{\alpha }{2}{\eta _{{\rm{user}}}} + \left( {\frac{2}{\alpha } - 1} \right){\rho ^{{\rm{ul}}}}} \right) = \left( {1 - \frac{2}{\alpha }} \right)\left( {\frac{\alpha }{2}{\eta _{{\rm{bs}}}} + {\rho ^{{\rm{ul}}}}} \right) \ge 0$, $t_3$ becomes always dominant and thus
\begin{equation}
\sum\limits_{U_j\in\mathcal{U}\backslash\{U_k\}}{\left| {\psi _{kj}^\mathsf{zf}} \right|^2} \asymp {N^{{ - \frac{\alpha }{2}{\eta _{{\rm{bs}}}} + \frac{\alpha }{2}{\eta _{{\rm{user}}}} - {\eta _{{\rm{ant}}}}}}}.
\end{equation}

 Now, suppose that  $\eta_{\rm{bs}}\ge \eta_{\rm{user}}$ and $\rho^{\rm{ul}}\ge -\frac{\alpha}{2}\eta_{\rm{user}}$. In this case, $\mathcal{L}_j=\mathcal{X}_j$ if $d\le {\rho^{\rm{ul}}}/{\alpha}$ and $\mathcal{X}_j\bigcap\mathcal{L}_j =\emptyset$ if $d>\rho^{\rm{ul}}/\alpha$ for all $U_j\in\mathcal{U}_k^{(d)}$ in probability so that, by using Lemma 2, we obtain
\[\begin{array}{l}
{t_1}\asymp{N^{\frac{\alpha }{2}{\eta _{{\rm{bs}}}} + {\eta _{{\rm{user}}}} + \left( {\frac{2}{\alpha } - 1} \right){\rho ^{{\rm{ul}}}}}},\\
{t_2}\asymp{N^{{\eta _{{\rm{bs}}}} + {\eta _{{\rm{user}}}} + 2\left( {\frac{2}{\alpha } - 1} \right){\rho ^{{\rm{ul}}}}}},\\
{t_3}\asymp{N^{\frac{\alpha }{2}{\eta _{{\rm{bs}}}} + {\eta _{{\rm{user}}}} + \left( {\frac{2}{\alpha } - 1} \right){\rho ^{{\rm{ul}}}}}},\\
{t_4}\asymp{N^{{\eta _{{\rm{bs}}}} + {\eta _{{\rm{user}}}} + 2\left( {\frac{2}{\alpha } - 1} \right){\rho ^{{\rm{ul}}}}}}.
\end{array}\]
Since $\frac{\alpha }{2}{\eta _{{\rm{bs}}}} + {\eta _{{\rm{user}}}} + \left( {\frac{2}{\alpha } - 1} \right){\rho ^{{\rm{ul}}}} - \left( {{\eta _{{\rm{bs}}}} + {\eta _{{\rm{user}}}} + 2\left( {\frac{2}{\alpha } - 1} \right){\rho ^{{\rm{ul}}}}} \right) = \left( {1 - \frac{2}{\alpha }} \right)\left( {\frac{\alpha }{2}{\eta _{{\rm{bs}}}} + {\rho ^{{\rm{ul}}}}} \right) \ge 0$, $t_1$ or $t_3$ becomes dominant and thus
\begin{equation}
\sum\limits_{U_j\in\mathcal{U}\backslash\{U_k\}}{\left| {\psi _{kj}^\mathsf{zf}} \right|^2} \asymp {N^{\left( {\frac{2}{\alpha } - 1} \right){\rho ^{{\rm{ul}}}} - \frac{\alpha }{2}{\eta _{{\rm{bs}}}} + {\eta _{{\rm{user}}}} - {\eta _{{\rm{ant}}}}}}.
\end{equation}
When $\eta_{\rm{bs}}\le \eta_{\rm{user}}$, the similar derivation can be done and we obtain
\begin{equation}
\sum\limits_{U_j\in\mathcal{U}\backslash\{U_k\}}{\left| {\psi _{kj}^\mathsf{zf}} \right|^2} \asymp {N^{\left( {\frac{2}{\alpha } - 1} \right){\rho ^{{\rm{ul}}}} - \frac{\alpha }{2}{\eta _{{\rm{bs}}}} + {\eta _{{\rm{user}}}} - {\eta _{{\rm{ant}}}}}}.
\end{equation}

Also,
\begin{equation}
\begin{split}
{\left| {\psi _{kk}^\mathsf{zf}} \right|^2} &\asymp 1+\frac{{\frac{{{1}}}{{{P^{{\rm{ul}}}}}}\sum\limits_{X_l \in {{\cal L}_k}} {{\beta _{lk}}}  + {{{P^{{\rm{ul}}}}}}\sum\limits_{X_l \in {\cal X}_k\backslash {{\cal L}_k}} {\beta _{lk}^3} }}{{M{{\left( {\sum\limits_{X_l \in {{\cal L}_k}} {{\beta _{lk}}}  + {{P_k^{{\rm{ul}}}}}\sum\limits_{X_l \in {\cal X}_k\backslash {{\cal L}_k}} {\beta _{lk}^2} } \right)}^2}}} \\
&\asymp 1+ {N^{ - {\rho ^{{\rm{ul}}}} - \frac{\alpha }{2}{\eta _{{\rm{bs}}}} - {\eta _{{\rm{ant}}}}}}\\
&\asymp 1,
\end{split}
\end{equation}
where the second and last asymptotic equalities come from Lemma 2 and the fact $-\frac{\alpha}{2}\eta_{\rm{bs}}\le \rho^{\rm{ul}}<0$, respectively. Since the DL transmit power can be written as
\begin{equation}
\begin{split}
P_k^{{\rm{dl}}} & = \frac{{{Q^\mathsf{zf}_k}}}{M}\left({{\sum\limits_{X_l \in {\mathcal{L}_k}}  {{\beta _{lk}}}  + {{{P^{{\rm{ul}}}}}}\sum\limits_{X_l \in \mathcal{X}_k\backslash{\mathcal{L}_k}} {\beta _{lk}^2} }}\right)^{-1}\\
& \asymp {N^{\nu - \frac{\alpha }{2}{\eta _{{\rm{bs}}}} - {\eta _{{\rm{ant}}}}}}
\end{split}
\end{equation}
so that $\nu = \rho^{\rm{dl}}+\frac{\alpha}{2}\eta_{\rm{bs}}+\eta_{\rm{ant}}$ for a reasonable DL transmit power,
\begin{equation}
\mathsf{snr^{zf}} =  \rho^{\rm{dl}}+\frac{\alpha}{2}\eta_{\rm{bs}}+\eta_{\rm{ant}},
\end{equation}
and 
\begin{equation}
\begin{split}
\mathsf{si{r^{zf}}} = \left\{ {\begin{array}{*{20}{l}}
{\left( {1 - \frac{2}{\alpha }} \right){{\left( {{\rho ^{{\rm{ul}}}} + \frac{\alpha }{2}{\eta _{{\rm{user}}}}} \right)}^ + } + \frac{\alpha }{2}\left( {{\eta _{{\rm{bs}}}} - {\eta _{{\rm{user}}}}} \right) + {\eta _{{\rm{ant}}}},}&{{\text{if }}{\eta _{{\rm{bs}}}} \ge {\eta _{{\rm{user}}}},}\\
{\left( {1 - \frac{2}{\alpha }} \right){\rho ^{{\rm{ul}}}} + \frac{\alpha }{2}{\eta _{{\rm{bs}}}} - {\eta _{{\rm{user}}}} + {\eta _{{\rm{ant}}}},}&{{\text{if }}{\eta _{{\rm{bs}}}} < {\eta _{{\rm{user}}}}.}
\end{array}} \right.
\end{split}
\end{equation}

\emph{Cases \textsf{M} and \textsf{L}} ($ \rho^{\rm{ul}} <-\frac{\alpha}{2}\eta_{\rm{bs}}$):
in these cases,   
$\mathbf{R}_k \asymp {\rm{diag}}\{{\beta_{1k}}\mathbf{I}_{M},{\beta_{2k}}\mathbf{I}_{M},\cdots,{\beta_{Lk}}\mathbf{I}_{M}\}$,
and 
\begin{equation}\label{ML_2}
\begin{split}
&{\left[ {{\bf{\widehat G}} {{{\bf{\widehat G}}}^H}} \right]_{i,j}} \asymp
\left\{ {\begin{array}{*{20}{l}}
{{{MP_j^{{\rm{ul}}}}}\sum\limits_{{X_l} \in {{\cal X}_j}} {\beta _{lj}^2}, }&{{\text{if }}i = j,}\\
{{{\sqrt {MP_i^{{\rm{ul}}}P_j^{{\rm{ul}}}} }}\sum\limits_{{X_l} \in {{\cal X}_i}\bigcap {{{\cal X}_j}} } {\sqrt {\beta _{li}^2\beta _{lj}^2} }, }&{{\text{if }}i \ne j,}
\end{array}} \right.
\end{split}
\end{equation}
Inserting (\ref{ML_2}) into (\ref{ZF_psi}), we have 
\begin{equation}
\begin{split}
{\left| {{\psi ^\mathsf{zf}_{kj}}} \right|^2} \asymp& \mathbbm{1}_{kj}+ \frac{{1}}{{MP_j^{\rm{ul}}}}{{{{\left( {\sum\limits_{X_l\in\mathcal{X}_j} {\beta _{lj}^2} } \right)}^{-2}}}}{{\sum\limits_{X_l\in\mathcal{X}_j} {\beta _{lj}^2{\beta _{lk}}} }},
\end{split}
\end{equation}
which is obtained in a similar way as before as 
\begin{equation}
\begin{split}
{\left| {{\psi^\mathsf{zf} _{kk}}} \right|^2}  &\asymp 1+  {{N^{ - {\rho ^{\rm{ul}}} - \frac{\alpha }{2}{\eta_{\rm{bs}}} - {\eta _{\rm{ant}}}}}}\\
 &\asymp  \left\{ {\begin{array}{*{20}{l}}
1,&{{\text{if }}\mathsf{M},}\\
{{N^{ - {\rho ^{{\rm{ul}}}} - \frac{\alpha }{2}{\eta _{{\rm{bs}}}} - {\eta _{{\rm{ant}}}}}},}&{{\text{if }}\mathsf{L},}
\end{array}} \right.
\end{split}
\end{equation}
and
\begin{equation}
\begin{split}
\sum\limits_{U_j\in\mathcal{U}\backslash\{U_k\}} {{{\left| {{\psi ^\mathsf{zf}_{kj}}} \right|}^2}}  &\asymp \frac{{1}}{{MP^{\rm{ul}}}}\sum\limits_{U_j\in\mathcal{U}\backslash\{U_k\}} {{{{\left( {\sum\limits_{X_l\in\mathcal{X}_j} {\beta _{lj}^2} } \right)}^{-2}}}}{{{\sum\limits_{X_l\in\mathcal{X}_j} {\beta _{lj}^2{\beta _{lk}}} }}} \\
& \asymp {N^{ - {\rho ^{{\rm{ul}}}} - \left( {\frac{\alpha }{2} + 1} \right){\eta _{{\rm{bs}}}} - {\eta _{{\rm{ant}}}} + {\eta _{{\rm{user}}}} - \left( {\frac{\alpha }{2} - 1} \right){{\left( {{\eta _{{\rm{bs}}}} - {\eta _{{\rm{user}}}}} \right)}^ + }}} .
\end{split}
\end{equation}
Since the DL transmit power can be written as
\begin{equation}
\begin{split}
P_k^{\rm{dl}} &= \frac{Q^\mathsf{zf}_k}{MP_k^{\rm{ul}}}\left(\sum\limits_{X_l\in{\mathcal{X}_k}}\beta_{lk}^2\right)^{-1}
\asymp N^{\nu-\rho^{\rm{ul}}-{\alpha}\eta_{\rm{bs}} -\eta_{\rm{ant}}}
\end{split}
\end{equation}
so that $\nu = \rho^{\rm{dl}}+\rho^{\rm{ul}}+\alpha\eta_{\rm{bs}}+\eta_{\rm{ant}}$ for a reasonable DL transmit power, we can obtain
\begin{equation}
\begin{split}
\mathsf{sn{r^{zf}}} &=  \left\{ \begin{array}{*{20}{l}}
{\rho ^{{\rm{dl}}}} + {\rho ^{{\rm{ul}}}} + \alpha {\eta _{{\rm{bs}}}}{\rm{ + }}{\eta _{{\rm{ant}}}},&{\text{if \textsf{M},}}\\
{\rho ^{{\rm{dl}}}} + \frac{\alpha }{2}{\eta _{{\rm{bs}}}},&{\text{if \textsf{L},}}
\end{array} \right.
\end{split}
\end{equation}
and
\begin{equation}
\begin{split}
\mathsf{si{r^{zf}}} =&  \left\{ \begin{array}{l}
{\rho ^{{\rm{ul}}}} + \frac{\alpha }{2}{\eta _{{\rm{bs}}}} + {\eta _{{\rm{ant}}}} \\+ ({\eta _{{\rm{bs}}}} - {\eta _{{\rm{user}}}}) + \left( {\frac{\alpha }{2} - 1} \right){\left( {{\eta _{{\rm{bs}}}} - {\eta _{{\rm{user}}}}} \right)^ + },{\text{if \textsf{M},}}\\
({\eta _{{\rm{bs}}}} - {\eta _{{\rm{user}}}}) + \left( {\frac{\alpha }{2} - 1} \right){\left( {{\eta _{{\rm{bs}}}} - {\eta _{{\rm{user}}}}} \right)^ + },{~~\text{if \textsf{L},}}
\end{array} \right.
\end{split}
\end{equation}
which concludes the proof.

\section{Proof of Theorem 6}
Since the total transmit power is given by $P_\Sigma^{\rm{dl}} + P_\Sigma^{\rm{ul}}= \Theta(N^{\max\{\rho^{{\rm{dl}}},\rho^{\rm{ul}}\}+\eta_{\rm{user}}})$ by definition, in order to guarantee $\mathsf{sinr}^\Upsilon \ge \tau$, $\mathsf{snr}^{\Upsilon}  \ge\tau$ and $\mathsf{sir}^{\Upsilon}\ge \tau$ at the power constraints of $\rho^{{\rm{dl}}}\le \rho-\eta_{\rm{user}}$ and $\rho^{\rm{ul}} \le \rho -\eta_{\rm{user}}$. Formally, in order to find the supportable user scaling exponent, we should solve the following optimization problem (P1):
\begin{equation*}
\begin{aligned}
& {\text{maximize}} & & {\eta_{{\rm{user}}}} \\
& \text{subject to} & & {\mathsf{snr}}^\Upsilon \ge \tau ,\mathsf{sir}^\Upsilon \ge \tau ,\\
& & &{\rho ^{{\rm{ul}}}} \le \rho- {\eta_{{\rm{user}}}},{\rho ^{{\rm{dl}}}} \le \rho - {\eta_{{\rm{user}}}},0\le\eta_{\rm{user}}\le1.
\end{aligned}
\end{equation*}
One important and straightforward observation is that more UL or DL transmit power results in higher SNR and SIR scaling exponents. So, the maximum of $\eta_{\rm{user}}$ is obtained when 
\begin{equation}\label{power_condition}
\rho^{\rm{ul}}=\rho^{\rm{dl}} =\rho-\eta_{\rm{user}}. 
\end{equation}
Inserting (\ref{power_condition}) into the SNR and SIR scaling exponents of each operation, we can readily derive the conditions ${\mathsf{snr}}^\Upsilon \ge \tau$ and $\mathsf{sir}^\Upsilon \ge \tau$, and then the supportable user scaling exponent is determined within the intersection of the two conditions.

Since the proof for the results for one region is similar to those for other regions, the proof for region  $\mathcal{D}^{\rm{zf}}$ is shown only for brevity.  From Theorems 1 and 3, the SNR scaling exponent of ZF operation is identical that of IF operation, and by inserting (\ref{power_condition}) into (\ref{zf_sir_scaling}), the conditions for the SIR scaling exponents are given by 
\begin{align}\label{ZF_condition1} 
\mathsf{snr^{zf}} = &  \rho- {\zeta^{\mathsf{zf}} _{{\rm{user}}}} + \frac{\alpha }{2}{\eta _{{\rm{bs}}}} + {\left( {\frac{\alpha }{2}{\eta _{{\rm{bs}}}} + {\eta _{{\rm{ant}}}} +\rho- {\zeta^{\mathsf{zf}} _{{\rm{user}}}}} \right)^ + } - {\left( {\frac{\alpha }{2}{\eta _{{\rm{bs}}}} +\rho - {\zeta^{\mathsf{zf}} _{{\rm{user}}}}} \right)^ + } 
\ge \tau,
\\ \nonumber
\mathsf{sir^{zf}} = & {\eta _{{\rm{bs}}}} - \zeta _{{\rm{user}}}^\mathsf{zf} + \left( {\frac{\alpha }{2} - 1} \right){\left( {{\eta _{{\rm{bs}}}} - \zeta _{{\rm{user}}}^\mathsf{zf}} \right)^ + }
 + {\left( {\frac{\alpha }{2}{\eta _{{\rm{bs}}}} + {\eta _{{\rm{ant}}}} + \rho  - \zeta _{{\rm{user}}}^\mathsf{zf}} \right)^ + } - {\left( {\frac{\alpha }{2}{\eta _{{\rm{bs}}}} + \rho  - \zeta _{{\rm{user}}}^\mathsf{zf}} \right)^ + } 
\\ \label{ZF_condition2} 
&+ \left( {1 - \frac{2}{\alpha }} \right){\left( {\frac{\alpha }{2}\min \left\{ {{\eta _{{\rm{bs}}}},\zeta _{{\rm{user}}}^\mathsf{zf}} \right\} + {\rho -\zeta^\mathsf{zf}_{\rm{user}}}} \right)^ + }
+ \frac{2}{\alpha }{\left( {{\rho - \zeta^\mathsf{zf}_{\rm{user}} }} \right)^ + } \ge \tau.
\end{align}
Then, it is sufficient to show that the above conditions hold when inserting $\zeta_{\rm{user}}^\mathsf{zf}$ into  (\ref{ZF_condition1}) and (\ref{ZF_condition2}). For the case $(\rho,\tau)\in\mathcal{D}^\mathsf{zf}$, inserting $\zeta_{\rm{user}}^{\mathsf{zf}} = \frac{1}{2}\left(\rho-\tau+\left(1+\frac{\alpha}{2}\right)\eta_{\rm{bs}}+\eta_{\rm{ant}}\right)$ into  (\ref{ZF_condition1}) and (\ref{ZF_condition2}) yields
\begin{equation}\nonumber
\begin{split}
\mathsf{snr^{zf}} &= \rho  + \tau  + \left( {\frac{\alpha }{2} - 1} \right){\eta _{{\rm{bs}}}} \ge \tau,\\ 
\mathsf{sir^{zf}}&={\eta _{{\rm{bs}}}} - \zeta _{{\rm{user}}}^\mathsf{zf} + \left( {\frac{\alpha }{2}{\eta _{{\rm{bs}}}} + {\eta _{{\rm{ant}}}} + \rho  - \zeta _{{\rm{user}}}^\mathsf{zf}} \right) 
+ \left( {1 - \frac{2}{\alpha }} \right){\left( {\frac{\alpha }{2}{\eta _{{\rm{bs}}}} + \rho  - \zeta _{{\rm{user}}}^\mathsf{zf}} \right)^ + } + \frac{2}{\alpha }{\left( {\rho  - \zeta _{{\rm{user}}}^\mathsf{zf}} \right)^ + }\\
&={\eta _{{\rm{bs}}}} - \zeta _{{\rm{user}}}^\mathsf{zf} + \left( {\frac{\alpha }{2}{\eta _{{\rm{bs}}}} + {\eta _{{\rm{ant}}}} + \rho  - \zeta _{{\rm{user}}}^\mathsf{zf}} \right)\\
& = \tau,
\end{split}
\end{equation}
where the first equality comes from the fact that $\eta_{\rm{bs}}\ge\zeta_{\rm{user}}^\mathsf{zf}$ and the second equality comes from the fact that ${\frac{\alpha }{2}{\eta _{{\rm{bs}}}} + \rho  - \zeta _{{\rm{user}}}^\mathsf{zf}}$.
Since it is easily seen that the conditions are violated for larger $\zeta_{\rm{user}}^\mathsf{zf}$, the proof for region $\mathcal{D}^\mathsf{zf}$ is completed.
\end{IEEEproof}


\begin{thebibliography}{1}
\bibitem{AndrewsWhat}
J. G. Andrews et al., ``What will 5G be?,'' \emph{IEEE J. Sel. Areas Commun.}, vol. 32, pp. 1065--1082, June 2014.
\bibitem{BhushanNetwork}
N. Bhushan, J. Li, D. Malladi, R. Gilmore, D. Brenner, A. Damnjanovic, R. T. Sukhavasi, C. Patel, and S. Geirhofer, ``Network densification: the dominant theme for wireless evolution into 5G,'' \emph{ IEEE Commun. Mag.}, vol. 52, no. 2, pp. 82--89, Feb. 2014.
\bibitem{RusekScaling}
F. Rusek, D. Persson, B. K. Lau, E. G. Larsson, T. L. Marzetta, O. Edfors, and F. Tufvesson, ``Scaling up MIMO: Opportunities and challenges with very large arrays,'' \emph{IEEE Signal Process. Mag.}, vol. 30, no. 1, pp. 40--46, Jan. 2013.
\bibitem{GotsisUDN}
A. Gotsis, S. Stefanatos, and A. Alexiou, ``Ultra dense networks: The new wireless frontier for enabling 5G access,'' \emph{IEEE Veh. Technol. Mag.}, vol. 11, no. 2, pp. 71--78, June 2016.
\bibitem{NgoEnergy}
N. Q. Ngo, E. G. Larsson, and T. L. Marzetta, ``Energy and spectral efficiency of very large multiuser MIMO systems,'' \emph{IEEE Trans. Commun.}, vol. 61, no. 4, pp. 1436--1449, Apr. 2012.
\bibitem{HwangHolistic}
 I. Hwang, B. Song, and S. Soliman, ``A holistic view on hyper-dense heterogeneous and small cell networks,'' \emph{IEEE Commun. Mag.}, vol. 51, pp. 20--27, June 2013.
\bibitem{HoydisGreen}
J. Hoydis, M. Kobayashi, and M. Debbah, ``Green small-cell networks,'' \emph{IEEE Veh. Technol. Mag.}, vol. 6, pp. 37--43, Mar. 2011.
\bibitem{HuangIncreasing}
H. Huang, M. Trivellato, A. Hottinen, M. Sha, P. Smith, and R. Valenzuela, ``Increasing DL cellular throughput with limited network MIMO coordination,'' \emph{IEEE Trans. Wireless Commun.}, vol. 8, pp. 2983--2989, June 2009.
\bibitem{IrmerCoordinated}
R. Irmer, H. Droste, P. Marsch, M. Grieger, G. Fettweis, S. Brueck, H.-P. Mayer, L. Thiele, and V. Jungnickel, ``Coordinated multipoint: Concepts, performance, and field trial results,'' \emph{IEEE Commun. Mag.}, vol. 49, no. 2, pp. 102--111, Feb. 2011.
\bibitem{CheckoCloud}
 A. Checko et al., ``Cloud RAN for Mobile Networks -- A Technology Overview,'' \emph{IEEE Commun. Surv. Tut.}, vol. 17, no. 1, 1 Quarter 2015, pp. 405--426
\bibitem{CaireAchievable}
G. Caire and S. Shamai, ``On the achievable throughput of a multiantenna Gaussian broadcast channel,'' \emph{IEEE Trans. Inf. Theory}, vol. 49, no. 7, pp. 1691--1706, July 2003.
\bibitem{LinToward}
I. Chih-Lin et al., ``Toward green and soft: A 5g perspective,'' \emph{IEEE Commun. Mag.}, vol. 52, no. 2, pp. 66--73, Feb. 2014.


\bibitem{BjornsonMassive}
E. Bj\"{o}rnson, J. Hoydis, M. Kountouris, and M. Debbah, ``Massive MIMO systems with non-ideal hardware: Energy efficiency, estimation, and capacity limits,'' \emph{IEEE Trans. Inf. Theory}, vol. 60, no. 11, pp. 7112--7139, Nov. 2014.


\bibitem{BiermannHow}
T. Biermann, L. Scalia, C. Choi, W. Kellerer and H. Karl, ``How front/backhaul networks influence the feasibility of coordinated multipoint in cellular networks,'' \emph{IEEE Commun. Mag.}, vol. 51, no. 8, pp. 168--176, Aug. 2013.

\bibitem{ChanclouOptical}
P. Chanclou et al,, ``Optical fiber solution for mobile fronthaul to achieve cloud radio access network,'' \emph{Proc. Future Netw. Mobile Summit}, pp. 1--11,


\bibitem{LoMaximum} 
T. K. Y. Lo, ``Maximum ratio transmission,'' \emph{IEEE Trans. Commun.}, vol.
47, pp. 1458--1461, Oct. 1999.

\bibitem{SpencerZero}
Q. H. Spencer, A. L. Swindlehurst, and M. Haardt, ``Zero-forcing methods for DL spatial multiplexing in multiuser MIMO channels,'' \emph{IEEE Trans. Signal Process.}, vol. 52, no. 2, pp. 461--471, Feb. 2004.





\bibitem{MarschLarge}
P. Marsch, M. Grieger, and G. Fettweis, ``Large scale field trial results on different UL coordinated multi-point (CoMP) concepts in an urban environment,'' \emph{in Proc. 2011 IEEE WCNC}, pp. 1858--1863.


\bibitem{BasnayakaPerformance1}
D. A. Basnayaka, P. J. Smith, and P. A. Martin, ``Performance analysis of dual-user macrodiversity MIMO systems with linear receivers in flat rayleigh fading,'' \emph{IEEE Trans. Wireless Commun.}, vol. 11, no. 12, pp. 4394--4404, Dec. 2012.


\bibitem{BasnayakaPerformance2}
D. A. Basnayaka, P. J. Smith, and P. A. Martin, ``Performance analysis of macrodiversity MIMO systems with MMSE and ZF receivers in flat Rayleigh fading,'' \emph{IEEE Trans. Wireless Commun.}, vol. 12, no. 5, May 2013. 

\bibitem{CheikhAnalytical}
D. B. Cheikh, J.-M. Kelif, M. Coupechoux, and P. Godlewski, ``Analytical joint processing multi-point cooperation performance in rayleigh fading,'' \emph{IEEE Wireless Commun. Lett.}, vol. 1, no. 4, Aug. 2012.


\bibitem{YangPerformance}
H. Yang and T. L. Marzetta, ``Performance of conjugate and zero-forcing beamforming in large-scale antenna systems,'' \emph{IEEE J. Sel. Areas Commun.}, vol. 31, no. 2, pp. 172--179, Feb. 2013.


\bibitem{WangAsymptotic}
J. Wang and L. Dai, ``Asymptotic rate analysis of DL multi-User systems With co-Located and distributed Antennas,'' \emph{IEEE Trans. Wireless Commun.}, vol. 14, no. 6, pp. 3046--3058, June 2015.


\bibitem{TanbougiTractable}
R. Tanbougi, S. Singh, J. G. Andrew, and F. K. Jondral ``A tractable model for noncoherent joint-transmission base station cooperation,'' \emph{IEEE Trans. Wireless Commun.}, vol. 13, no. 9, pp. 4959--4973,  Sep. 2014. 
\bibitem{NigamCoordinated}
G. Nigam, P. Minero, and M. Haenggi, ``Coordinated multipoint joint transmission in heterogeneous networks,'' \emph{IEEE Trans. Commun.}, vol. 62, no. 11, pp. 4134--4146, Nov. 2014.
\bibitem{GiovanidisStochastic}
A. Giovanidis and F. Baccelli, ``A stochastic geometry framework for analyzing pairwise-cooperative cellular networks,'' \emph{IEEE Trans. Wireless Commun.}, vol. 14, no. 2, pp. 794--808, Feb. 2015.
\bibitem{HuangAnalytical}
K. Huang, and J. G. Andrews, ``An analytical framework for multicell cooperation via stochastic geometry and large deviations,''  \emph{IEEE Trans. Inf. Theory}, vol. 59, no. 4, pp. 2501--2516, Apr. 2013.
\bibitem{LeeSpectral}
N. Lee, F. Baccelli, and R.W. Heath, ``Spectral efficiency scaling laws in dense random wireless networks with multiple receive antennas,'' \emph{IEEE Trans. Inf. Theory}, vol. 62, no. 3, pp. 1344--1359, Mar. 2016.
\bibitem{MarschUL}
P. Marsch and G. Fettweis, ``UL CoMP under a constrained backhaul and imperfect channel knowledge,'' \emph{IEEE Trans. Wireless Commun}. vol. 10, no. 6, pp. 1730--1742, June 2011.
\bibitem{WangPerformance}
Q. Wang, and Y. Jing, ``Performance analysis and scaling law of MRC/MRT relaying with CSI error in massive MIMO systems,'' \emph{submitted to IEEE Trans. Wireless Commun.}, June 2016. [Online]. Available: http://arxiv.org/pdf/1606.07480v1.pdf.
\bibitem{KongMultipair}
C. Kong, C. Zhong, M. Matthaiou, E. Bj\"ornson, and Z. Zhang, ``Multipair two-way half-duplex relaying with massive arrays and imperfect CSI,'' \emph{submitted to IEEE Trans. Inf. Theory}, July 2016, [Online]. Available: http://arxiv.org/pdf/1607.01598v1.pdf.
\bibitem{XiaFundamentals}
P. Xia, H.-S. Jo, and J. Andrews, ``Fundamentals of inter-cell overhead signaling in heterogeneous cellular networks,'' \emph{IEEE J. Sel. Topics Sig. Proc.}, vol. 6, no. 3, pp. 257--269, June 2012.
\bibitem{SimeoneUL}
O. Simeone, O. Somekh, H. V. Poor and S. Shamai, ``Local base station cooperation via finite-capacity links for the uplink of linear cellular networks,'' \emph{IEEE Trans Inf. Theory}, vol. 55, no. 1, pp. 190--204, Jan. 2009.
\bibitem{SimeoneDL}
O. Simeone, O. Somekh, H. V. Poor, and S. Shamai, ``Downlink multicell processing with limited-backhaul capacity,'' \emph{EURASIP J. Advances in Signal Process.} vol. 2009, Feb. 2009.

\bibitem{ZhouUplink}
M. Liu, Y. Teng, and M. Song, ``Performance analysis of coordinated multipoint joint transmission in ultra-dense networks with limited backhaul capacity,'' \emph{IEEE J. Sel. Areas Commun.}, vol. 31, no. 10, pp.  2111--2113, Oct. 2013.

\bibitem{DhillonDownlink}
H. S. Dhillon and J. G. Andrews, ``Downlink rate distribution in heterogeneous cellular networks under generalized cell selection,'' \emph{IEEE Wireless Commun. Lett.}, vol. 3, no. 1, pp. 42--45, Feb. 2014. 

\bibitem{LiPilot}
S. M. Kay, \emph{Fundamentals of Statistical Signal Processing: Estimation Theory}. Prentice-Hall, Inc. Upper Saddle River, NJ, USA, 1993.

\bibitem{Argos}
C. Shepard \emph{et al.}, ``Argos: Practical many-antenna base stations,'' in \emph{Proc. 18th Annu. Int. Conf. MobiCom}, 2012, pp. 53–-64.


\bibitem{SharifRBF}
M. Sharif and B. Hassibi  ``On the capacity of MIMO broadcast channel with partial side information",  \emph{IEEE Trans. Inf. Theory},  vol. 51,  no. 2,  pp. 506--522, Feb. 2005.

\bibitem{BillingsleyConvergence}
 P. Billingsley, \emph{Convergence of Probability Measures}, 1968, Wiley.


\bibitem{KnuthBig}
D. E. Knuth, ``Big omicron and big omega and big theta,'' \emph{SIGACT News}, vol. 8, no. 2, pp. 18--24, Apr.-Jun. 1976.



\bibitem{LevequeInformation}
O. Leveque and I.E. Telatar, ``Information-theoretic upper bounds on the capacity of large extended Ad hoc wireless networks,'' \emph{IEEE Trans. Inf. Theory}, vol. 51, no. 3, pp. 858--3211, Sep. 2007.

\end{thebibliography}
\end{document}